\newif\ifabstract
\abstracttrue
\newif\iffull
\ifabstract \fullfalse \else \fulltrue \fi

\newif\iflower
\lowertrue

\documentclass[11pt]{article}
\usepackage{amsfonts}
\usepackage{amssymb}
\usepackage{amstext}
\usepackage{amsmath}
\usepackage{xspace}
\usepackage{theorem}
\usepackage{graphicx}
\usepackage{url}
\usepackage{graphics}
\usepackage{colordvi}
\usepackage{colordvi}
\usepackage{subfigure}

\advance \oddsidemargin by -0.8in
\newcommand{\myparskip}{3pt}
\parskip \myparskip

\newenvironment{proofof}[1]{\noindent{\bf Proof of #1.}}%
        {\hspace*{\fill}$\Box$\par\vspace{4mm}}

\newcommand{\optcro}[1]{\mathsf{cr}(#1)}

\newcommand{\ceil}[1]{\ensuremath{\left\lceil#1\right\rceil}}
\newcommand{\floor}[1]{\ensuremath{\left\lfloor#1\right\rfloor}}

\newcommand{\tone}{\tilde 1}
\newcommand{\ttwo}{\tilde 2}
\newcommand{\tthree}{\tilde 3}

\newcommand{\set}[1]{\left\{ #1 \right\}}

\newcommand{\tset}{{\mathcal T}}

\newcommand{\pset}{{\mathcal{P}}}
\newcommand{\qset}{{\mathcal{Q}}}
\newcommand{\lset}{{\mathcal{L}}}
\newcommand{\bset}{{\mathcal{B}}}
\newcommand{\aset}{{\mathcal{A}}}
\newcommand{\cset}{{\mathcal{C}}}
\newcommand{\nset}{{\mathcal{N}}}

\newcommand{\mset}{{\mathcal M}}

\newcommand{\rset}{{\mathcal{R}}}

\newcommand{\hset}{{\mathcal{H}}}
\newcommand{\sset}{{\mathcal{S}}}

\newcommand{\be}{\begin{enumerate}}
\newcommand{\ee}{\end{enumerate}}
\newcommand{\bd}{\begin{description}}
\newcommand{\ed}{\end{description}}
\newcommand{\bi}{\begin{itemize}}
\newcommand{\ei}{\end{itemize}}

\newtheorem{theorem}{Theorem}[section]
\newtheorem{lemma}[theorem]{Lemma}
\newtheorem{observation}[theorem]{Observation}
\newtheorem{corollary}[theorem]{Corollary}
\newtheorem{claim}[theorem]{Claim}
\newtheorem{definition}{Definition}[section]
\newenvironment{proof}{\par \smallskip{\bf Proof:}}{\hfill\stopproof}
\def\stopproof{\square}
\def\square{\vbox{\hrule height.2pt\hbox{\vrule width.2pt height5pt \kern5pt
\vrule width.2pt} \hrule height.2pt}}

\renewcommand{\phi}{\varphi}

\newenvironment{properties}[2][0]
{
\begin{enumerate} \setcounter{enumi}{#1}}{\end{enumerate}}

\begin{document}

\begin{titlepage}

\title{Improved Bounds for the Flat Wall Theorem\footnote{An extended abstract is to appear in SODA 2015}}
\author{Julia Chuzhoy\thanks{Toyota Technological Institute, Chicago, IL
60637. Email: {\tt cjulia@ttic.edu}. Supported in part by NSF grant CCF-1318242.}
}

\maketitle

\thispagestyle{empty}
\begin{abstract}
The Flat Wall Theorem of Robertson and Seymour states that there is some function $f$, such that for all integers $w,t>1$, every graph $G$ containing a wall of size $f(w,t)$, must contain either (i) a $K_t$-minor; or (ii) a small subset $A\subset V(G)$ of vertices, and a flat wall of size $w$ in $G\setminus A$. Kawarabayashi, Thomas and Wollan recently showed a self-contained proof of this theorem with the following two sets of parameters:  (1) $f(w,t)=\Theta(t^{24}(t^2+w))$ with $|A|=O(t^{24})$, and (2) $f(w,t)=w^{2^{\Theta(t^{24})}}$ with $|A|\leq t-5$. The latter result gives the best possible bound on $|A|$. In this paper we improve their bounds to  $f(w,t)=\Theta(t(t+w))$ with $|A|\leq t-5$. For the special case where the maximum vertex degree in $G$ is bounded by $D$, we show that, if $G$ contains a wall of size $\Omega(Dt(t+w))$, then either $G$ contains a $K_t$-minor, or there is a flat wall of size $w$ in $G$. This setting naturally arises in algorithms for the Edge-Disjoint Paths problem, with $D\leq 4$. Like the proof of Kawarabayashi et al., our proof is self-contained, except for using a well-known theorem on routing pairs of disjoint paths. We also provide efficient algorithms that return either a model of the $K_t$-minor, or a vertex set $A$ and a flat wall of size $w$ in $G\setminus A$.

\iflower
We complement our result for the low-degree scenario by proving an almost matching lower bound: namely, for all integers $w,t>1$, there is a graph $G$, containing a wall of size $\Omega(wt)$, such that the maximum vertex degree in $G$ is $5$, and $G$ contains no flat wall of size $w$, and no $K_t$-minor.
\fi
\end{abstract}
\end{titlepage}

\section{Introduction}
The main combinatorial object studied in this paper is a wall. In order to define a wall $W$ of height $h$ and width $r$, or an $(h\times r)$-wall, we start from a grid of height $h$ and width $2r$. 
Let $C_1,\ldots, C_{2r}$ be the columns of the grid in their natural left-to-right order. For each column $C_j$, let $e_{1}^j,e_{2}^j,\ldots,e_{h-1}^j$ be the edges of $C_j$, in their natural top-to-bottom order. If $j$ is odd, then we delete all edges $e_{i}^j$ where $i$ is even. If $j$ is even, then we delete all edges $e_{i}^j$ where $i$ is odd. We then remove all vertices of the resulting graph whose degree is $1$. This  final graph, denoted by $\hat W$, is called an \emph{elementary $(h\times r)$-wall}  (see Figure~\ref{fig: wall}). The \emph{pegs} of $\hat W$ are all the vertices on its outer boundary that have degree $2$. An $(h\times r)$-wall $W$ is simply a subdivision of the elementary $(h\times r)$-wall $\hat W$, and the pegs of $W$ are defined to be the vertices of $W$ that serve as the pegs of $\hat W$. We sometimes refer to a $(w\times w)$-wall as a \emph{wall of size $w$}.

\begin{figure}[h]
\scalebox{0.4}{\includegraphics{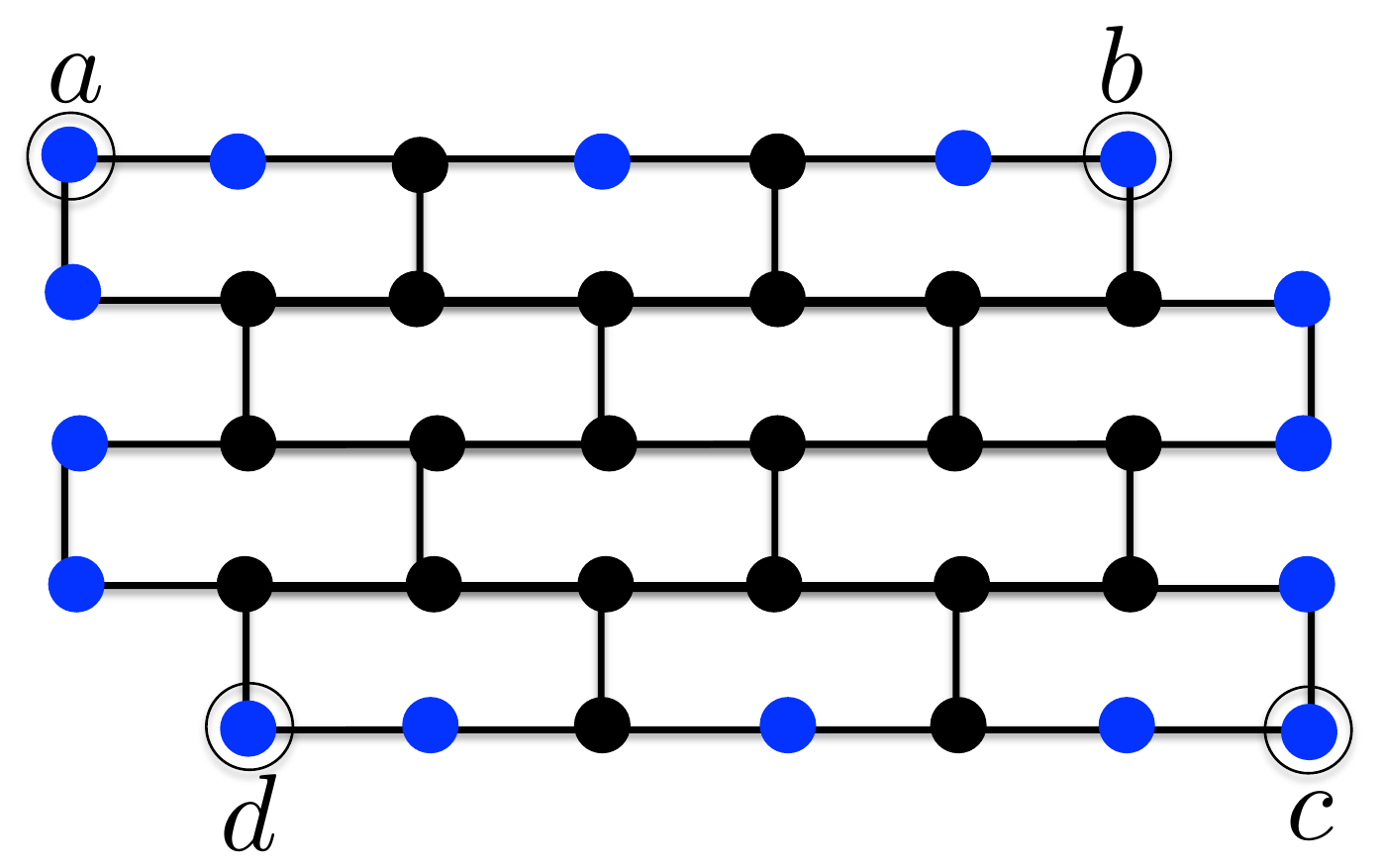}}\caption{An elementary wall of height $5$ and width $4$. The corners are circled, and the pegs are shown in blue.\label{fig: wall}}
\end{figure}

The well-known Excluded Grid Theorem of Robertson and Seymour~\cite{GMT-RS} states that there is some function $g:\mathbb{Z^+}\rightarrow \mathbb{Z^+}$, such that for every integer $w\geq 1$, every graph of treewidth at least $g(w)$ contains the $(w\times w)$-grid as a minor. Equivalently, if the treewidth of $G$ is $g(w)$, then $G$ contains a wall of size $\Omega(w)$. This important theorem has found many applications in graph theory and algorithms. However, in some scenarios it is useful to have more structure than that provided by the presence of a large wall in a graph. The Flat Wall Theorem helps provide this additional structure, and it is used, for example, in  algorithms for the Node-Disjoint Paths problem~\cite{flat-wall-RS}. We start with some basic definitions and results that are needed in order to state the Flat Wall Theorem.

Suppose we are given a graph $G=(V,E)$ with four special vertices $s_1,t_1,s_2,t_2$. In the Two-Disjoint-Paths problem, our goal is to find two disjoint paths $P_1$ and $P_2$ in $G$, with $P_1$ connecting $s_1$ to $t_1$, and $P_2$ connecting $s_2$ to $t_2$. The well-known Two-Disjoint-Paths Theorem~\cite{Jung,RS90, Seymour06, Shiloach80, Thomassen80} states that either there is a solution to the Two-Disjoint-Paths problem, or $G$ can ``almost'' be drawn inside a disc in the plane, with $s_1,s_2,t_1,t_2$ appearing on its boundary in this circular order. In order to define the specific notion of the ``almost'' drawing, we first need to define $C$-reductions.

Recall that a \emph{separation} in a graph $G$ is a pair $X,Y$ of sub-graphs of $G$, such that $G=X\cup Y$, and $E(X)\cap E(Y)=\emptyset$. The \emph{order} of the separation is $|V(X)\cap V(Y)|$. Let $C$ be any set of vertices in graph $G$, and let $(X,Y)$ be any separation of $H$ of order at most $3$ with $C\subseteq V(Y)$. Assume further that all vertices of $X\cap Y$ are connected inside  graph $X$. Let $\tilde G$ be the graph obtained from $Y$ by adding the edges connecting all pairs of vertices in $X\cap Y$. Then we say that $\tilde G$ is an \emph{elementary $C$-reduction} of $H$. Observe that if $s_1,t_1,s_2,t_2\in C$, then there is a solution to the Two-Disjoint-Paths problem in $G$ iff there is such a solution in $\tilde G$. This is since at most one of the two paths $P_1,P_2$ may contain the vertices of $X\setminus Y$. We say that a graph $G^*$ is a \emph{$C$-reduction} of $G$ iff it can be obtained from $G$ by a sequence of elementary $C$-reductions. The Two-Disjoint-Paths Theorem states that either there is a feasible solution to the Two-Disjoint-Paths problem in $G$, or some $C$-reduction of $G$, for $C=\set{s_1,s_2,t_1,t_2}$, can be drawn inside a disc in the plane, with the vertices of $C$ appearing on the boundary of the disc, in the circular order $(s_1,s_2,t_1,t_2)$.

More generally, let $C$ be any set of vertices of $G$, and let $\tilde{C}$ be any circular ordering of the vertices of $C$. A $\tilde{C}$-cross in $G$ is a pair $P_1,P_2$ of disjoint paths, whose endpoints are denoted by $s_1,t_1$ and $s_2,t_2$, respectively, such that $s_1,s_2,t_1,t_2\in C$, and they appear in $\tilde C$ in this circular order. A more general version of the Two-Disjoint-Paths Theorem~\cite{Jung,RS90, Seymour06, Shiloach80, Thomassen80} states that either $G$ contains a $\tilde{C}$-cross, or some $C$-reduction of $G$ can be drawn inside a disc in the plane, with the vertices of $C$ appearing on the boundary of the disc, in the order specified by $\tilde C$. In the latter case, we say that graph $G$ is \emph{$\tilde C$-flat}.


Given a wall $W$, let $\Gamma(W)$ be the outer boundary of $W$, and let $C$ be the set of the pegs of $W$. 
We say that $W$ is a \emph{flat wall} in $G$ iff there is a separation $(X,Y)$ of $G$, with $W\subseteq Y$, $X\cap Y$ contained in $\Gamma(W)$, and the set $C$ of all pegs of $W$ contained in $X\cap Y$, such that, if we denote $Z=X\cap Y$, and $\tilde{Z}$ is the ordering of the vertices of $Z$ induced by $\Gamma(W)$, then graph $Y$ is $\tilde{Z}$-flat.

We are now ready to state the Flat Wall Theorem (a more formal statement appears in Section~\ref{sec: prelims}).  There are two functions $f$ and $g$, such that, for any integers $w,t>1$, for any graph $G$ containing a wall of size $f(w,t)$, either (i) $G$ contains a $K_t$-minor, or (ii) there is a set $A$ of at most $g(t)$ vertices in $G$, and a flat wall of size $w$ in $G\setminus A$. A somewhat stronger version of this theorem was originally proved by Robertson and Seymour~\cite{flat-wall-RS}, with $g(t)=O(t^2)$; however, they do not provide explicit bounds on $f(w,t)$. Giannopoulou and Thilikos~\cite{flat-wall-GT} showed a proof of this theorem with $g(t)=t-5$, obtaining the best possible bound on $|A|$, but they also do not provide explicit bounds on $f(w,t)$. Recently, Kawarabayashi, Thomas and Wollan~\cite{KTW} gave a self-contained proof of the theorem  in the following two settings:  with $g(t)=O(t^{24})$ they achieve $f(w,t)=\Theta(t^{24}(t^2+w))$, and with $g(t)= t-5$, they obtain $f(w,t)=w^{2^{\Theta(t^{24})}}$. 
They also provide an efficient algorithm, that, given a wall of size $\Theta(t^{24}(t^2+w))$, either computes a model of the $K_t$-minor in $G$, or returns a set $A$ of at most $O(t^{24})$ vertices, and a  flat wall of size $w$ in  graph $G\setminus A$.

In this paper we improve their bounds to  $f(w,t)=\Theta(t(t+w))$ with $g(t)= t-5$. We note that this is the best possible bound on $|A|$, since one can construct a graph $G$ containing an arbitrarily large wall and no $K_t$-minor, such that at least $t-5$ vertices need to be removed from $G$ in order to obtain a flat wall of size $w$ for any $w>2$ (see Section~\ref{subsec: thm statements}). For the special case where the maximum vertex degree in $G$ is bounded by $D$, we show that, if $G$ contains a wall of size $\Omega(Dt(t+w))$, then either $G$ contains a $K_t$-minor, or there is a flat wall of size $w$ in $G$. This latter setting naturally arises in algorithms for the Edge-Disjoint Paths problem, with $D\leq 4$. Like the proof of Kawarabayashi et al., our proof is self-contained, except for using the Two-Disjoint-Paths Theorem. We also provide efficient algorithms that return either a model of the $K_t$-minor, or a set $A$ and a flat wall of size $w$ in $G\setminus A$. \iflower We complement our latter result by proving an almost matching lower bound: namely, for all integers $w,t>1$, there is a graph $G$, containing a wall of size $\Omega(wt)$, such that the maximum vertex degree in $G$ is $5$, and $G$ contains no flat wall of size $w$, and no $K_t$-minor.\fi

We now briefly summarize our techniques and compare them to the techniques of Kawarabayashi et al.~\cite{KTW}. The proof of the flat wall theorem in~\cite{KTW} proceeds as follows. Let $W$ be the $(R\times R)$ wall in $G$. Kawarabayashi et al. start by showing that either there is a collection $\pset=\set{P_1,\ldots,P_k}$ of $k=\Omega(t^{12})$ disjoint paths in $G$, where each path $P_i$ connects a pair of vertices $x_i,y_i\in W$, and is internally disjoint from $W$, such that the distance between every pair of vertices in set $\set{x_i,y_i\mid 1\leq i\leq k}$ is large in $W$; or there is a set $A$ of $O(t^{24})$ vertices, such that, if $P$ is a path in graph $G\setminus A$ connecting a pair of vertices $x,y\in W$, such that $P$ is internally disjoint from $W$, 
then the distance between the endpoints of $P$ is small in $W$. In the former case, the paths in $\pset$ are exploited, together with the wall $W$ to find a model of the $K_t$-minor in $G$. Assume now that the latter case happens. The wall $W$ is then partitioned into $O(t^{24})$ disjoint horizontal strips of equal height, so at least one of the strips does not contain any vertex of $A$. Denote this strip by $S$. Strip $S$ is in turn partitioned into a large number of disjoint walls, where each wall spans a number of consecutive columns of the strip $S$. They show that either one of the resulting walls contains a large sub-wall that is flat in graph $G\setminus A$; or we can find a model of a $K_t$-minor in $G$.

The starting point of our proof is somewhat different. Instead of working with a square $(R\times R)$ wall, it is more convenient for us to work with a wall whose width $r$ is much larger than its height $h$. In order to achieve this, we start with the $(R\times R)$ wall, and partition it into horizontal strips of height $h$. We then connect these strips in a snake-like manner to obtain one long wall, of width $\Omega(R^2/h)$ and height $h$. This strip is partitioned into $\Omega(R^2/h^2)$ disjoint walls of size $(h\times h)$, that we call \emph{basic walls}. 
 Let $\bset=(B_1,\ldots,B_N)$ be the resulting sequence of basic walls, for $N=\Omega(R^2/h^2)$. For each such basic wall $B_i$, we define a core sub-wall $B'_i$ of $B_i$, obtained from $B_i$ by deleting the top $2t$ and the bottom $2t$ rows. The construction of the long strip and its partition into basic walls and core walls imposes a convenient structure on the wall, that allows us to improve the parameters of the flat wall theorem. Let $\Gamma'_i$ be the outer boundary of the core wall $B'_i$.
A path $P$ connecting a vertex in $B'_i\setminus \Gamma'_i$ to some vertex of $W\setminus B_i'$, such that $P$ is internally disjoint from $W$, is called a \emph{bridge for $B'_i$}. If the other endpoint of $P$ lies in one of the walls $B_{i-1},B_i,B_{i+1}$, then we call it a \emph{neighborhood bridge}.

We show that if we can find a collection of $\Omega(t^2)$ disjoint neighborhood bridges incident on distinct core walls, or a collection of $\Omega(t^2)$  disjoint non-neighborhood bridges incident on distinct core walls, then we can find a $K_t$-minor in $G$. Our constructions of the $K_t$-minors are  more efficient than those in~\cite{KTW}, in that they require a much smaller number of disjoint bridges. This is achieved by exploiting the convenient structure of a long wall partitioned into basic walls, and several new ways to embed a clique minor into $G$. In a theorem somewhat similar to that of~\cite{KTW}, we show that we can either find a collection of $\Omega(t^2)$ disjoint non-neighborhood bridges incident on distinct core walls, or there is a set $A$ of $O(t^2)$ vertices, and a large subset $\bset'$ of basic walls, such that for each basic wall $B_i\in \bset'$, every bridge for the corresponding core wall $B'_i$ in graph $G\setminus A$ is a neighborhood bridge. In the former case, we use the disjoint non-neighborhood bridges to find a $K_t$-minor. Assume now that the latter case happens.  If many of the walls in $\bset'$ have neighborhood bridges incident on their corresponding core walls, then we construct a collection of $\Omega(t^2)$ disjoint neighborhood bridges incident on distinct core walls, which implies that $G$ contains a $K_t$-minor. Otherwise, for each wall $B_i\in \bset'$, we try to find a pair $P_i,Q_i$ of disjoint paths, with $P_i$ connecting the top left corner of $B'_i$ to its bottom right corner, and $Q_i$ connecting its top right corner to its bottom left corner, such that $P_i,Q_i$ are internally disjoint from $W$. We show that either we can find, for each $B_i\in \bset'$, the desired pair $(P_i,Q_i)$ of paths, such that all these paths are disjoint, or one of the walls $B_i\in \bset'$ contains a large sub-wall that is flat in $G\setminus A$ (a more careful analysis than the one described here leads to an improved bound of $|A|\leq t-5$). In the former case, we again construct a $K_t$-minor, by exploiting the paths $\set{P_i,Q_i}_{B_i\in\bset'}$, while in the latter case we obtain the desired flat wall in $G\setminus A$.
The main difference of our approach from that of~\cite{KTW} is (1) converting the square wall $W$ into a long strip $S$, which is partitioned into smaller square basic walls, and defining a core wall for each basic wall. Performing this step right at the beginning of the algorithm imposes a convenient structure on the wall $W$ that makes the analysis easier; (2) we propose more different ways to embed a $K_t$-minor into $G$, which in turn lead to improved parameters; and (3) careful analysis that allows us to lower $|A|$ from $\Theta(t^2)$ to $t-5$, without increasing the size of the wall we start from.

\paragraph{Organization} We start with preliminaries in Section~\ref{sec: prelims}. Since the formal statements of our main results require defining some graph-theoretic notation, these statements can also be found in Section~\ref{sec: prelims}. In Sections~\ref{sec: two graphs}--\ref{sec: bridges, core walls, wall types} we lay the foundations for proving both upper bounds: in Section~\ref{sec: two graphs} we describe several families of graphs such that, if $G$ contain any such graph as a minor, then it must contain a $K_t$-minor. In Section~\ref{sec: cutting the wall} we describe an algorithm that turns a square $R\times R$ wall into a ``long'' wall of height $h$ and width $R^2/h$. This long wall is then partitioned into $R^2/h^2$ basic walls of size $h\times h$. In Section~\ref{sec: bridges, core walls, wall types} we partition the basic walls into several types, and show how to handle most of these types. Sections~\ref{sec: proof of main week thm} and \ref{sec: proof of main strong theorem} complete the proofs of the two upper bounds, where Section~\ref{sec: proof of main week thm} focuses on the small-degree case, and \ref{sec: proof of main strong theorem} handles general graphs. We provide the proof of our lower bound in Section~\ref{sec: lower bound}.

\section{Preliminaries and Statements of the Main Theorems}\label{sec: prelims}
Throughout the paper, we use two parameters: $t$ and $w$, and our goal is to either find a $K_t$-minor or a flat wall of size $(w\times w)$. We denote $T=t(t-1)/2$ throughout the paper.

We say that a path $P$ is \emph{internally disjoint} from a set $U$ of vertices, if no vertex of $U$ serves as an inner vertex of $P$.
We say that two paths $P,P'$ are \emph{internally disjoint}, iff for each $v\in V(P)\cap V(P')$, $v$ is an endpoint of both paths.

Given a graph $G$ and three sets $A,X,B$ of vertices of $G$, we say that $X$ \emph{separates} $A$ from $B$ iff $G\setminus X$ contains no paths from the vertices of $A\setminus X$ to the vertices of $B\setminus X$.

\begin{definition} A \emph{separation} in graph $G$ is a pair $G_1,G_2$ of subgraphs of $G$, such that $G=G_1\cup G_2$ and $E(G_1)\cap E(G_2)=\emptyset$. The \emph{order of the separation} is $|V(G_1)\cap V(G_2)|$.
\end{definition}

Notice that if $(G_1,G_2)$ is a separation of $G$, then there are no edges in $G$ between $V(G_1\setminus G_2)$ and $V(G_2\setminus G_1)$.

\begin{definition} Given a graph $G$ and a path $P$ in $G$, we say that $P$ is a $2$-path iff every inner vertex of $P$ has degree $2$ in $G$. In other words, $P$ is an induced path in $G$. We say that $P$ is a \emph{maximal $2$-path} iff the degree of each of the two endpoints of $P$ is not $2$.
\end{definition}

\subsection{Minors and Models}
We say that a graph $H$ is a \emph{minor} of a graph $G$, iff $H$ can be obtained from $G$ by a series of edge deletion, vertex deletion, and edge contraction operations. Equivalently, $H$ is a minor of $G$ iff there is a map $f:V(H)\rightarrow 2^{V(G)}$ assigning to each vertex $v\in V(H)$ a subset $f(v)$ of vertices of $G$, such that:

\begin{itemize}
\item For each $v\in V(H)$, the sub-graph of $G$ induced by $f(v)$ is connected;

\item If $u,v\in V(H)$ and $u\neq v$, then $f(u)\cap f(v)=\emptyset$; and

\item For each edge $e=(u,v)\in E(H)$, there is an edge in $E(G)$ with one endpoint in $f(v)$ and the other endpoint in $f(u)$.
\end{itemize}
A map $f$ satisfying these conditions is called \emph{a model of $H$ in $G$}.
The following observation follows easily from the definition of minors.

\begin{observation}\label{obs: transitivity of minors}
If $H$ is a minor of $G$ and $H'$ is a minor of $H$ then $H'$ is a minor of $G$.
\end{observation}

It is sometimes more convenient to use embeddings instead of models for graph minors. A valid embedding of a graph $H$ into a graph $G$ is a map $\phi$, mapping every vertex $v\in V(H)$ to a connected sub-graph $\phi(v)$ of $G$, such that, if $u,v\in V(H)$ with $u\neq v$, then $\phi(v)\cap \phi(u)=\emptyset$. Each edge $e=(u,v)\in E(H)$ is mapped to a path $\phi(e)$ in $G$, such that one endpoint of $\phi(e)$ belongs to $V(\phi(v))$, another endpoint to $V(\phi(u))$, and the path does not contain any other vertices of $\bigcup_{v'\in V(H)}\phi(v')$. We also require that all paths in $\set{\phi(e)\mid e\in E(H)}$ are internally disjoint. A valid embedding of $H$ into $G$ can be easily converted into a model of $H$ in $G$, and can be used to certify that $H$ is a minor of $G$.

\subsection{Walls and Grids}
In this part we formally define grid graphs and wall graphs. We note that Kawarabayashi et al.~\cite{KTW} provide an excellent overview and intuitive definitions for all terminology needed in the statement of the Flat Wall Theorem. Many of our definitions and explanations in this section follow their paper.

We start with a grid graph. A grid of height $h$ and width $r$ (or an $(h\times r)$-grid), is a graph, whose vertex set is:
$\set{v(i,j)\mid 1\leq i\leq h;1\leq j\leq r}$.
The edge set consists of two subsets: a set of \emph{horizontal edges} $E_1=\set{(v(i,j),v(i,j+1))\mid 1\leq i\leq h; 1\leq j<r}$; and a set of \emph{vertical edges} $E_2=\set{(v(i,j),v(i+1,j))\mid 1\leq i<h; 1\leq j\leq r}$. The sub-graph induced by $E_1$ consists of $h$ disjoint paths, that we refer to as \emph{the rows of the grid}. The $i$th row, that we denote by $R_i$, is the row incident on $v(i,1)$. Similarly, the sub-graph induced by $E_2$ consists of $r$ disjoint paths, that we refer to as \emph{the columns of the grid}. The $j$th column, that is denoted by $C_j$, is the column starting from $v(1,j)$. Geometrically, we view the rows $R_1,\ldots,R_h$ as ordered from top to bottom, and the columns $C_1,\ldots,C_r$ as ordered left-to-right in the standard drawing of the grid.
We say that vertices $v(i,j)$ and $v(i',j')$ of the grid are separated by at least $z$ columns iff $|j-j'|>z$.

We now proceed to define a wall graph $W$. In order to do so, it is convenient to first define an \emph{elementary wall} graph, that we denote by $\hat W$.
To construct an elementary wall $\hat W$ of height $h$ and width $r$ (or an $(h\times r)$-elementary wall), we start from a grid of height $h$ and width $2r$. Consider some column $C_j$ of the grid, for $1\leq j\leq r$, and let $e_{1}^j,e_{2}^j,\ldots,e_{h-1}^j$ be the edges of $C_j$, in the order of their appearance on $C_j$, where $e_{1}^j$ is incident on $v(1,j)$. If $j$ is odd, then we delete from the graph all edges $e_{i}^j$ where $i$ is even. If $j$ is even, then we delete from the graph all edges $e_{i}^j$ where $i$ is odd. We process each column $C_j$ of the grid in this manner, and in the end delete all vertices of degree $1$. The resulting graph is an elementary wall of height $h$ and width $r$, that we denote by $\hat W$ (See Figure~\ref{fig: wall}). 

Let $E'_1$ be the set of edges of $\hat W$ that correspond to the horizontal edges of the original grid, and let $E'_2$ be the set of the edges of $\hat W$ that correspond to the vertical edges of the original grid, so $E_1'= E_1,E_2'\subseteq E_2$. Notice that as before, the sub-graph of $\hat W$ induced by $E'_1$ defines a collection of $h$ node-disjoint paths, that we refer to as the rows of $\hat W$. We denote these rows by $R_1,\ldots,R_h$, where for $1\leq i\leq h$, $R_i$ is incident on $v(i,1)$. (It will be clear from context whether we talk about the rows of a wall graph or of a grid graph). Let $V_1$ denote the set of all vertices in the first row of $\hat W$, and $V_h$ the set of vertices in the last row of $\hat W$. There is a unique set $\cset$ of $r$ node-disjoint paths, where each path $C\in \cset$ starts at a vertex of $V_1$, terminates at a vertex of $V_h$, and is internally disjoint from $V_1\cup V_h$. We refer to these paths as the columns of $\hat W$. We order these columns from left to right, and denote by $C_j$ the $j$th column in this ordering, for $1\leq j\leq r$. The sub-graph $Z=R_1\cup C_1\cup R_h\cup C_r$ of $\hat W$ is a simple cycle, that we call the \emph{outer bondary of $W$}. We now define the four corners of the wall. The top left corner $a$ is the unique vertex in the intersection of $R_1$ and $C_1$; the top right corner $b$ is the unique vertex in the intersection of $R_1$ and $C_r$. Similarly, the bottom left and right corners, $d$ and $c$ are defined by $R_h\cap C_1$ and $R_h\cap C_r$, respectively (see Figure~\ref{fig: wall}). All vertices of $Z$ that have degree $2$ are called the \emph{pegs} of $\hat W$. 

We say that a graph $W$ is \emph{a wall of height $h$ and width $r$}, or an $(h\times r)$-wall, iff it is a subdivision of the elementary wall $\hat W$ of height $h$ and width $r$. Notice that in this case, there is a natural mapping $f:V(\hat W)\rightarrow V(W)$, such that for $u\neq v$, $f(u)\neq f(v)$, and for each edge $e=(u,v)\in E(\hat W)$, there is a path $P_e$ in $W$ with endpoints $f(u),f(v)$, such that all paths $\set{P_e\mid e\in E(\hat W)}$ are internally disjoint from each other, and do not contain the vertices of $\set{f(u')\mid u'\in V(\hat W)}$ as inner vertices. We call such a mapping \emph{a good $(\hat W,W)$-mapping}. The corners of $W$ are defined to be the vertices to which the corners of $\hat W$ are mapped, and the pegs of $W$ are the vertices to which the pegs of $\hat W$ are mapped. Notice that the mapping $f$ is not unique, and so the choice of the corners and the pegs of $W$ is not fixed. For convenience, throughout this paper, the paths $P_e$ of $W$ corresponding to the horizontal edges of $\hat W$ are called \emph{blue paths}, and the paths $P_e$ corresponding to the vertical edges of $\hat W$ are called \emph{red paths}. For each $1\leq i\leq h$ and $1\leq j\leq r$, the $i$th row of $W$, $R_i$, and the $j$th column of $W$, $C_j$, are naturally defined as the paths  corresponding (via $f$) to the $i$th row and $j$th column of $\hat W$, respectively. A $(w\times w)$-wall is sometimes called \emph{a wall of size $w$}.

\begin{definition}
Let $W',W$ be two walls, where $W'$ is a sub-graph of $W$. We say that $W'$ is a \emph{sub-wall of $W$} iff every row of $W'$ is a sub-path of a row of $W$, and every column of $W'$ is a sub-path of a column of $W$. 
\end{definition}

Notice that if a wall $W$ is a sub-division of an elementary wall $\hat W$, and we are given some $(\hat W,W)$-good mapping $f: V(\hat W)\rightarrow V(W)$, then any sub-wall $\hat W'$ of $\hat W$ naturally defines a sub-wall $W'$ of $W$: wall $W'$ is the union of all paths $P_e$ for $e\in E(\hat W')$. Moreover, since $f$ is fixed, the corners and the pegs of $W'$ are uniquely defined.

We will often work with a special type of sub-walls of a given wall $W$ --- sub-walls spanned by contiguous sets of rows and columns of $W$. We formally define such sub-walls below.

Consider an $(h\times r)$ elementary-wall $\hat W$, and let $1\leq i_1<i_2\leq h$ be integers. We define a sub-wall of $\hat W$ spanned by rows $(R_{i_1},\ldots,R_{i_2})$ to be the sub-graph of $\hat W$ induced by $\bigcup_{i=i_1}^{i_2}V(R_i)$.
Similarly, for integers $1\leq j_1<j_2\leq r$ we define a sub-wall of $\hat W$ spanned by columns $(C_{j_1},\ldots,C_{j_2})$ to be the graph obtained from $\hat W$, by deleting all vertices in $\left(\bigcup_{j=1}^{j_1-1}V(C_j)\right )\cup\left(\bigcup_{j=j_2+1}^{r}V(C_j)\right )$, and deleting all vertices whose degree is less than $2$ in the resulting graph.
The sub-wall $\hat W''$ of $\hat W$  spanned by rows $(R_{i_1},\ldots,R_{i_2})$  and columns $(C_{j_1},\ldots,C_{j_2})$ is computed as follows: let $\hat W'$ be sub-wall of $\hat W$ spanned by rows $(R_{i_1},\ldots,R_{i_2})$. Then $\hat W''$ is the sub-wall of $\hat W'$ 
 spanned by columns $(C_{j_1},\ldots,C_{j_2})$.

Finally, assume we are given any $(h\times r)$-wall $W$, the corresponding $(h\times r)$-elementary wall $\hat W$ and a $(\hat W,W)$-good mapping $f: V(\hat W)\rightarrow V(W)$. For integers $1\leq i_1<i_2\leq h$, and $1\leq j_1<j_2\leq r$, we define the sub-wall $W'$ of $W$ spanned by rows $(R_{i_1},\ldots,R_{i_2})$ and columns $(C_{j_1},\ldots,C_{j_2})$, as follows. Let $\hat W'$ be the sub-wall of $\hat W$ spanned by 
 rows $(R_{i_1},\ldots,R_{i_2})$ and columns $(C_{j_1},\ldots,C_{j_2})$. We then let $W'$ be the unique sub-wall of $W$ corresponding to $\hat W'$ via the mapping $f$. That is, $W'$ is the union of all paths $P_e$ for $e\in E(\hat W')$. As observed before, since the mapping $f$ is fixed, the corners and the pegs of $W'$ are uniquely defined. Sub-walls of $W$ spanned by sets of consecutive rows, and sub-walls spanned by sets of consecutive columns are defined similarly.

From our definition of an elementary wall, it is clear that the $(h\times 2r)$-grid contains the $(h\times r)$-elementary wall as a minor. It is also easy to see that an $(h\times r)$-wall $W$ contains the $(h\times r)$-grid $G$ as a minor: let $\hat W$ be the $(h\times r)$-elementary wall, and assume that we are given some $(\hat W,W)$-good mapping $f:V(\hat W)\rightarrow V(W)$. Clearly, $\hat W$ is a minor of $W$. For every $1\leq i\leq h$, $1\leq j\leq r$, let $P(i,j)=R_i\cap C_j$, where $R_i$ and $C_j$ are the $i$th row and the $j$th column of $\hat W$, respectively. We contract all edges in $P(i,j)$. Once we process all pairs $R_i,C_j$ in this manner, we obtain the $(h\times r)$-grid $\tilde G$. We call $\tilde G$ \emph{a contraction} of $W$. Notice that if the mapping $f:V(\hat W)\rightarrow V(W)$ is fixed, then this contraction is uniquely defined, and so is the model of $\tilde G$ in $W$.
 
\subsection{Linkedness}  
We now turn to define the notion of $t$-linkedness that we use extensively in our proof.
\begin{definition}
For any integer $t>0$, we say that two disjoint sets $X,Y$ of vertices of are $t$-linked in graph $G$, iff for any pair $X'\subseteq X$, $Y'\subseteq Y$ of vertex subsets, with $|X'|=|Y'|\leq t$, there is a set of $|X'|$ node-disjoint paths in graph $G$, connecting the vertices of $X'$ to the vertices of $Y'$.
\end{definition}

A useful feature of grid graphs is that the sets of vertices in the first and the last columns of the grid are $t$-linked, as long as $t$ is no larger than the smaller of the dimensions of the grid. We show this in the following claim.

\begin{claim}\label{claim: t-linkedness of grid}
Let $G$ be an $(h\times r)$ grid, $t\leq \min\set{h,r}$ an integer, $X$ the set of all vertices on the first column of $G$ and $Y$ the set of all vertices on the last column of $G$. Then $X$ and $Y$ are $t$-linked in $G$.
\end{claim}
\begin{proof}
Let $X'\subseteq X$, $Y'\subseteq Y$ be any pair of vertex subsets with $|X'|=|Y'|=t'\leq t$. We claim that there is a set $\pset$ of $t'$ disjoint paths connecting $X'$ to $Y'$ inside $G$. Assume otherwise. Then there is a set $Z$ of $t'-1$ vertices separating $X'$ from $Y'$. Let $\rset_{X'}$ be the set of all rows on which the vertices of $X'$ lie, and define $\rset_{Y'}$ similarly for $Y'$. Then, since $|\rset_{X'}|=|\rset_{Y'}|=t'$, at least one row $R\in \rset_{X'}$, and at least one row $R'\in \rset_{Y'}$ contain no vertices of $Z$. Moreover, since $G$ contains $r\geq t\geq t'$ columns, at least one column $C$ of $G$ contains no vertices of $Z$. Combining $R,R'$ and $C$, we obtain a path connecting a vertex of $X'$ to a vertex of $Y'$ in $G\setminus Z$, a contradiction.
\end{proof}

A similar claim holds for wall graphs, except that we need to be more careful in defining the sets $X$ and $Y$ of vertices.
The proof of the following claim is identical to the proof of Claim~\ref{claim: t-linkedness of grid}.
\begin{claim}\label{claim: t-linkedness of wall}
Let $\hat W$ be an $(h\times r)$-elementary wall, $t\leq \min\set{h,r}$ a parameter, $X$ a set of vertices lying in the first column of $G$ and $Y$ a set of vertices lying in the last column of $G$, such that for each row $R_i$ of $\hat W$, $|X\cap R_i|\leq 1$ and $|Y\cap R_i|\leq 1$. Then $X$ and $Y$ are $t$-linked in $G$.
\end{claim}

 \subsection{$C$-Reductions  and Flat Walls}

 \begin{definition}
Let $G$ be a graph, $X\subseteq V(G)$, and let $(A,B)$ be a separation of $G$ of order at most $3$ with $X\subseteq A$. Moreover, assume that the vertices of $A\cap B$ are connected in $B$. Let $H$ be the graph obtained from $G[A]$ by adding an edge connecting every pair of vertices in $A\cap B$. We say that $H$ is an elementary $X$-reduction in $G$, determined by $(A,B)$. We say that a graph $J$ is an $X$-reduction of $G$ if it can be obtained from $G$ by a series of elementary $X$-reductions.
\end{definition}

We need a definition of $C$-flat graphs. Intuitively, let $G$ be any graph, and let $C$ be any simple cycle of $G$. Suppose there is some $C$-reduction $H$ of $G$, such that $H$ is a planar graph, and there is a drawing of $H$ in which $C$ bounds its outer face. Then we say that $G$ is $C$-flat. Following is an equivalent way to define $C$-flat graphs, due to~\cite{KTW}, which is somewhat more convenient to work with.

\begin{definition} Let $G$ be a graph, and let $C$ be a cycle in $G$. We say that $G$ is $C$-flat if there exist subgraphs $G_0,G_1,\ldots,G_k$ of $G$, and a plane graph $\tilde G$, such that:

\begin{itemize}
\item $G=G_0\cup G_1\cup\cdots\cup G_k$, and the graphs $G_0,G_1,\ldots,G_k$ are pairwise edge-disjoint;

\item $C$ is a subgraph of $G_0$.

\item $G_0$ is a subgraph of $\tilde G$, with $V(\tilde G)=V(G_0)$. Moreover, $\tilde G$ is a plane graph, and the cycle $C$ bounds its outer face;

\item For all $1\leq i\leq k$, $|V(G_i)\cap V(G_0)|\leq 3$. 

\begin{itemize}
\item If $|V(G_i)\cap V(G_0)|=2$, then $u$ and $v$ are adjacent in $\tilde G$;
\item If $V(G_i)\cap V(G_0)=\set{u,v,w}$, then some finite face of $\tilde G$ is incident with $u,v,w$ and no other vertex;
\end{itemize}

\item For all $1\leq i\neq j\leq k$, $V(G_i)\cap V(G_j)\subseteq V(G_0)$.
\end{itemize}
\end{definition}

We are now ready to define a flat wall.
\begin{definition} Let $G$ be a graph, and let $W$ be a wall in $G$ with outer boundary $D$. Suppose there is a separation $(A,B)$ of $G$, such that $A\cap B\subseteq V(D)$, $V(W)\subseteq B$, and there is a choice of pegs of $W$, such that every peg belongs to $A$. If some $A\cap B$-reduction of $G[B]$ can be drawn in a disc with the vertices of $A\cap B$ drawn on the boundary of the disc in the order determined by $D$, then we say that the wall $W$ is flat in $G$.
\end{definition}

\subsection{Statements of the Main Theorems}\label{subsec: thm statements}
We need one more definition in order to state our main theorems. Let $W$ be a wall in some graph $G$, and assume that $G$ contains a $K_t$-minor. Recall that a model of the $K_t$-minor in $G$ maps each vertex $v\in V(K_t)$ to a subset $f(v)$ of vertices of $G$. We say that the $K_t$-minor is \emph{grasped} by the wall $W$ iff for each $v\in V(K_t)$, $f(v)$ intersects at least $t$ rows of $W$, or at least $t$ columns of $W$.
We will use the following simple observation.

\begin{observation}\label{obs: clique grasped by wall}
Let $W$ be an $(h\times r)$-wall in a graph $H$, and $\tilde G$ an $(h\times r)$-grid, such that $\tilde G$ is a contraction of $W$. Suppose we are given a model $f(\cdot)$ of a $K_t$-minor in $\tilde G$, such that for each $v\in V(K_t)$, $f(v)$ intersects at least $t$ rows or at least $t$ columns of $\tilde G$. Then there is a model of $K_t$ in $H$ grasped by $W$.
\end{observation}

The observation follows from the fact that for each row $R_i$ of $\tilde G$, every vertex on $R_i$ is mapped by the contraction to a set of vertices of $H$ contained in the $i$th row of $W$, and the same holds for the columns of $\tilde G$.

We are now ready to state our main theorems.
Our first theorem is a slightly weaker version of the flat wall theorem, in that it is mostly interesting for graphs whose maximum vertex degree is relatively small. Such graphs arise, for example, in edge-disjoint routing problems, and the guarantees given by this theorem are somewhat better than the guarantees given by the stronger flat wall theorem that appears below, as we do not need to deal with apex vertices.

\begin{theorem}\label{thm: main weak}
 Let  $G$ be any graph with maximum vertex degree $D$, let $w,t> 1$ be integers, set $T=t(t-1)/2$, and let $R=(w+4t)\left (2+ \ceil{\sqrt{8D^2(10T+6)+14T+8}}\right )=\Theta(Dt(w+t))$. Then there is an efficient algorithm, that, given any $(R\times R)$-wall $W\subseteq G$, either computes a model of a $K_t$-minor grasped by $W$ in $G$, or returns a sub-wall $W^*$ of $W$ of size at least $(w\times w)$, such that $W^*$ is a flat wall in $G$. The running time of the algorithm is polynomial in $|V(G)|,D,w$ and $t$.\end{theorem}

\begin{theorem}\label{thm: main strong}
Let $G$ be any graph, let $w,t> 1$ be integers, set $T=t(t-1)/2$, and let $R=(w+4t)\left (2+\ceil{\sqrt{500T+200}}\right )=\Theta(t(w+t))$. Then there is an efficient algorithm, that, given any $(R\times R)$-wall $W\subseteq G$, either finds a model of a $K_t$-minor grasped by $W$ in $G$, or returns a set $A$ of at most $t-5$ vertices, and a sub-wall $W^*$ of $W$ of size at least $(w\times w)$, such that $V(W^*)\cap A=\emptyset$ and $W^*$ is a flat wall in $G\setminus A$. The running time of the algorithm is polynomial in $|V(G)|,w$ and $t$.
\end{theorem}

Observe that the bound of $t-5$ on $|A|$ is the best possible. Indeed, let $G$ be the graph obtained from an $(R\times R)$-elementary wall $W$ (for any value $R$), by adding a set $A=\set{a_1,\ldots,a_{t-5}}$ of new vertices, and connecting every vertex of $A$ to every vertex of $W$. Clearly, in order to obtain a flat wall of size $(w\times w)$ for any $w>2$ in $G$, we need to delete all vertices of $A$ from it. Assume for contradiction that $G$ contains a model $f$ of a $K_t$-minor. Let $\cset=\set{f(v)\mid v\in V(K_t)}$ be the sets of vertices of $G$ to which the vertices of $K_t$ are mapped. Then at most $t-5$ sets in $\cset$ may contain vertices of $A$. So there are at least $5$ sets $S_1,\ldots,S_5\in \cset$ of vertices, where for $1\leq i\leq 5$, $S_i\cap A=\emptyset$. But then sets $S_1,\ldots,S_5$ define a model of a $K_5$-minor in graph $W$, and since $W$ is planar, this is impossible.

Our lower bound is summarized in the following theorem.

\begin{theorem}\label{thm: lower bound}
For all integers $w,t>1$, there is a graph $G$, containing a wall of size $\Omega(wt)$, such that $G$ does not contain a flat wall of size  $w$, and it does not contain a $K_t$-minor. The maximum vertex degree of $G$ is $5$.
\end{theorem}

\subsection{A $C$-cross and a Wall-Cross}

Suppose we are given a graph $G$ and a cycle $C$ in $G$. A \emph{$C$-cross} in $G$ is a pair $P_1$ and $P_2$ of disjoint paths, with ends $s_1,t_1$ and $s_2,t_2$, respectively, such that $s_1,s_2,t_1,t_2$ occur in this order on $C$, and no vertex of $C$ serves as an inner vertex of $P_1$ or $P_2$. The next theorem follows from~\cite{Jung,RS90, Seymour06, Shiloach80, Thomassen80, KTW}.

\begin{theorem}\label{thm: a cross or C-flat}
Let $G$ be a graph and let $C$ be a cycle in $G$. Then the following conditions are equivalent:
\begin{itemize}
\item $G$ has no $C$-cross;
\item Some $C$-reduction of $G$ can be drawn in the plane with $C$ as a boundary of the outer face;
\item $G$ is $C$-flat.
\end{itemize}

Moreover, there is an efficient algorithm, that either computes a $C$-cross in $G$, or returns the subgraphs $G_0,G_1,\ldots,G_k$ of $G$ and the plane graph $\tilde G$, certifying that $G$ is $C$-flat.
\end{theorem}

We will be extensively using a special type of cross, connecting the corners of a wall. For brevity of notation, we define it below, and we call it a \emph{wall-cross}.

\begin{definition} Let $G$ be any graph, and $W$ an $(h\times r)$-wall in $G$. Let $\hat W$ be the corresponding $(h\times r)$-elementary wall, and assume that we are given a $(\hat W,W)$-good mapping $f: V(\hat W)\rightarrow V(W)$. Let $a,b,c,d$ be the four corners of $W$ (whose choice is fixed given $f$), appearing on the boundary of $W$ in this order. A wall-cross for $W$ is a pair $P_1,P_2$ of disjoint paths, where $P_1$ connects $a$ to $c$, and $P_2$ connects $b$ to $d$. 
\end{definition}

Assume that we are given any pair $u,v$ of vertices of a  wall $W$. We say that $u$ and $v$ are separated by a column $C_j$ of $W$, iff $u,v\not\in V(C_j)$, and $V(C_j)$ separates $u$ from $v$ in graph $W$. Similarly, we say that $u,v$ are separated by a row $R_i$ of $W$, iff $u,v\not\in V(R_i)$, and $V(R_i)$ separates $u$ from $v$ in $W$. We will repeatedly use the following simple theorem.

\begin{theorem}\label{thm: cross-wall} Assume that we are given a wall $W$ of height $h\geq 5$ and width $r\geq 5$, with corners $a,b,c,d$ appearing on the boundary of $W$ in this order. Let $u,v\in V(W)$ be any pair of vertices, such that one of the following holds: either (1) neither $u$ nor $v$ lie on the boundary of $W$ and they are separated by some row or some column of $W$; or (2) $u$ lies on the boundary of $W$, and $v$ lies in the sub-wall spanned by rows $(R_3,\ldots,R_{h-2})$ and columns $\set{C_3,\ldots,C_{r-2}}$. Let $W'$ be the graph obtained from $W$ by adding the edge $(u,v)$ to it. Then there is a wall-cross for $W$ in $W'$.
\end{theorem}

\begin{proof}
Assume first that $u,v$ do not lie on the boundary of $W$ and they are separated by some row $R_i$ of $W$. We assume w.l.o.g. that $u$ lies above $R_i$ and $v$ lies below $R_i$ in the natural drawing of $W$. Recall that a blue path of $W$ is a path corresponding to a horizontal edge of the corresponding elementary wall $\hat W$, and a red path corresponds to a vertical edge of $\hat W$. If $u$ is a vertex on a blue path of $W$, then let $u'=u$, $P_u=\emptyset$, and let $R_{i'}$ be the row of $W$ to which $u'$ belongs. Otherwise, let $P'_u$ be the red path to which $u$ belongs, and $u_1,u_2$ the endpoints of $P'_u$. Then $u_1,u_2$ belong to a pair of consecutive rows $R_{i'},R_{i'+1}$ of $W$. We let $u'\in \set{u_1,u_2}$ be the endpoint that belongs to $R_i'$, and we let $P_u$ be the sub-path of $P'_u$ between $u$ and $u'$. Similarly, if $v$ lies on a blue path, then we set $v'=v$, $P_v=\emptyset$, and $R_{i''}$ the row of $W$ to which $v'$ belongs. Otherwise, let $P'_v$ be the red path on which $v$ lies, and let $v_1,v_2$ be its two endpoints. Then $v_1,v_2$ belong to two consecutive rows $R_{i''-1}$, $R_{i''}$ of $W$. Let $v'\in \set{v_1,v_2}$ be the endpoint that belongs to $R_{i''}$, and let $P_v$ be the sub-path of $P'_v$ between $v$ and $v''$. Notice that since $u,v$ are separated by $R_i$, $i'<i<i''$.

We are now ready to define the two paths $P_1$ and $P_2$. We assume w.l.o.g that $a$ is the top left corner of $W$, and that $a,b,c,d$ appear on the outer boundary of $W$ in the clock-wise order. Path $P_1$ follows column $C_1$ until it reaches row $R_i$; it then follows $R_i$ until it reaches $C_r$, and finally it follows $C_r$ to $c$. Path $P_2$ starts at vertex $b$ and follows $C_r$ until it reaches $R_{i'}$, then follows $R_{i'}$ until it reaches $u'$, and then follows $P_u$ to $u$. It then uses the edge $(u,v)$ to reach $v$, follows $P_v$ to $v'$, 
row $R_{i''}$ to $C_1$, and $C_1$ to $d$ (see Figure~\ref{fig: wall-cross}).

The case where $u,v$ do not lie on the boundary of $W$ and they are separated by a column of $W$ is symmetric.
The proof for the case where $u$ lies on the boundary of $W$, and $v$ lies in the sub-wall spanned by rows $(R_3,\ldots,R_{h-2})$ and columns $\set{C_3,\ldots,C_{r-2}}$ is very similar and is omitted here.
\begin{figure}[h]
\scalebox{0.35}{\includegraphics{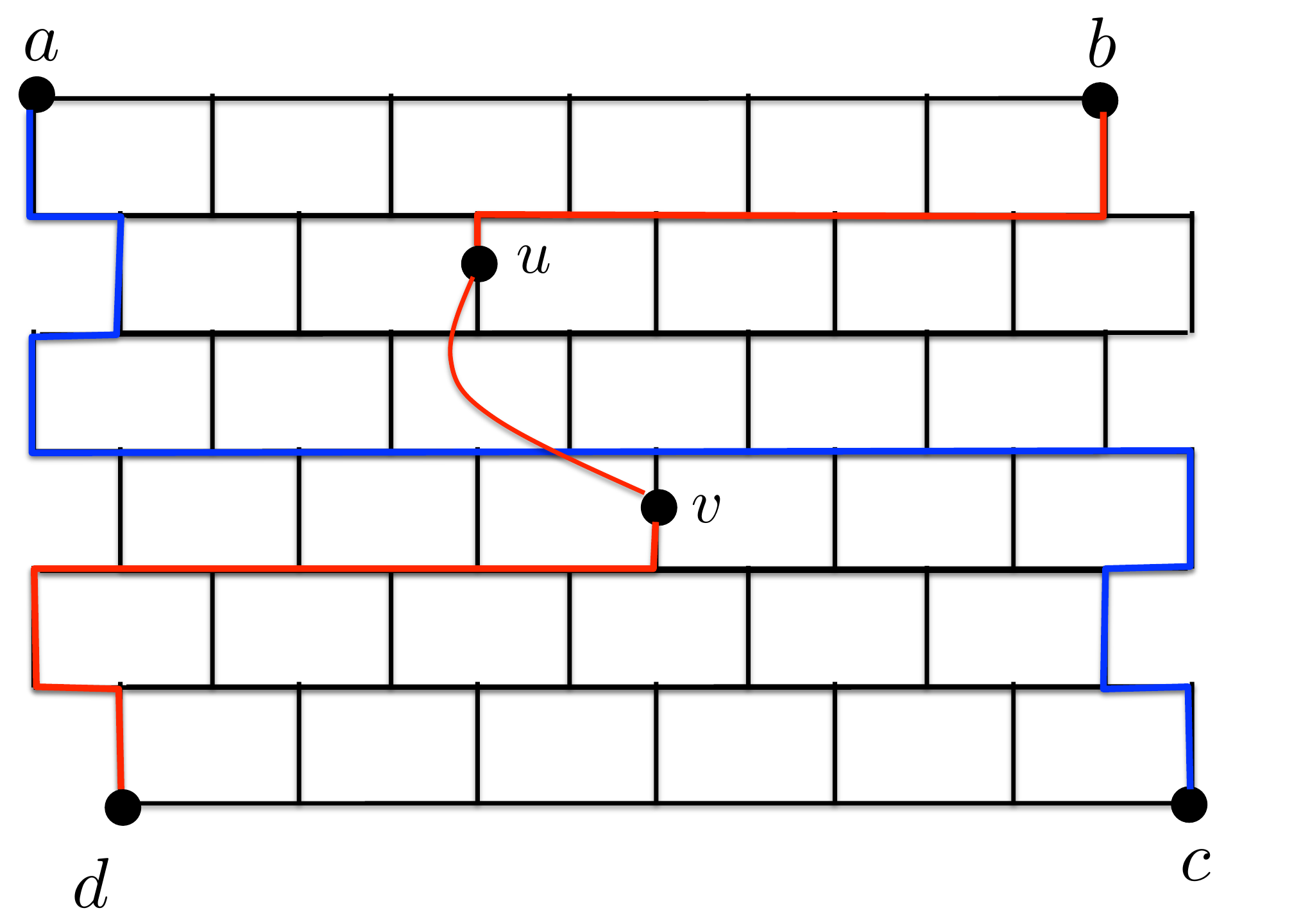}}\caption{Paths $P_1$ and $P_2$ are shown in blue and red, respectively. \label{fig: wall-cross}}
\end{figure}

\end{proof}

\label{--------------------------------------sec: two graphs--------------------------}
\section{Some Useful Graphs}\label{sec: two graphs}
In this section we construct a graph $H^*$, and define three families of graphs $\hset_1,\hset_2,\hset_3$, such that, if $G$ contains $H^*$ or one of the graphs in $\hset_1\cup\hset_2\cup\hset_3$ as a minor, then it contains a $K_t$-minor. 

\subsection{Graph $H^*$}
We start with a grid containing $2t$ rows, that we denote by $R_1,\ldots,R_{2t}$, and $Tt+1$ columns, denoted by $C_1,\ldots,C_{Tt+1}$. The vertex lying at the intersection of row $R_i$ and column $C_j$ is denoted by $v(i,j)$.
 
Consider the vertices $u_i=v(t,ti)$ for $1\leq i\leq T$ (so these are vertices roughly in the middle row of the grid, spaced $t$ apart horizontally). For each such vertex $u_i$, let $L_i$ be the cell of the grid for which $u_i$ is the upper left corner. We add two diagonals to this cell, that is, two edges: $e_i=(v(t,ti),v(t+1,ti+1))$, and $e_{i+1}=(v(t+1,ti),v(t,ti+1))$ (see Figure~\ref{fig: H1}). We call these edges \emph{cross edges}. This completes the definition of the graph $H^*$.

We now partition the graph $H^*$ into blocks. Given any pair $1\leq i<j\leq Tt+1$ of integers, let $B[i,j]$ be the sub-graph of $H^*$, induced by the set $V(C_i)\cup V(C_{i+1})\cup\cdots\cup V(C_j)$ of vertices. We define a sequence $J_1,M_1,J_2,M_2,\ldots,J_T,M_T$ of sub-graphs of $H^*$, as follows. For each $1\leq i\leq T$, $J_i=B[t(i-1)+1,ti]$, and $M_i=B[ti,ti+1]$  (see Figure~\ref{fig: H1}). We denote the set of all horizontal edges of $M_i$, together with the two cross edges contained in $M_i$, by $E_i$. Graphs $J_1,\ldots,J_T$ are called \emph{odd blocks}, and graphs $M_1,\ldots,M_T$ are called \emph{even blocks}.

\begin{figure}[h]
\scalebox{0.5}{\includegraphics{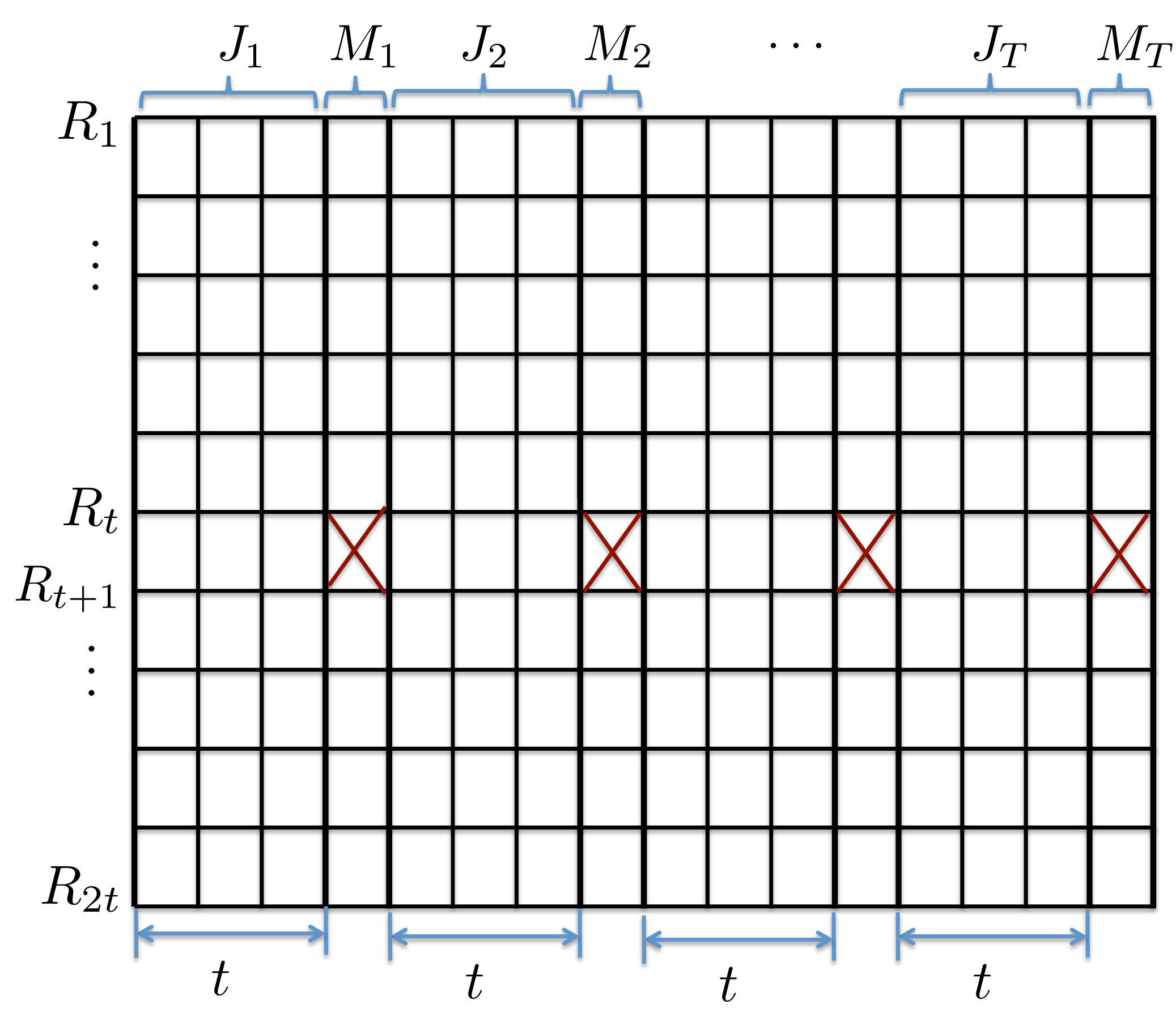}}\caption{Graph $H^*$\label{fig: H1}}
\end{figure}

\begin{lemma}\label{lemma: from H1 to clique}
Let $G$ be any graph and let $W$ be a wall in $G$. Assume that $G$ contains $H^*$ as a minor. Then $G$ contains a $K_t$-minor. Moreover, if $J_1$ is a contraction of some sub-wall of $W$, then there is a model of $K_t$ grasped by $W$ in $G$, and this model can be found efficiently given a model of $H^*$ in $G$.
\end{lemma}

\begin{proof}
For each odd block $J_i$, let $X_i$ be the set of the $2t$ vertices of the first column of $J_i$, and $Y_i$ the set of the $2t$ vertices on the last column of $J_i$. For convenience, let $X_{T+1}$ be the set of vertices of the last column of $M_T$.
For any subset $X'_i\subset X_i$ of vertices, we can view $X'_i$ as an ordered set, where the ordering is given by the order in which the vertices of $X'_i$ appear on the column of the grid induced by $X_i$. Similarly, any subset $Y'_i\subseteq Y_i$ can be viewed as an ordered set of vertices. From Claim~\ref{claim: t-linkedness of grid}, for each $1\leq i\leq T$, $X_i$ and $Y_i$ are $t$-linked inside $J_i$.
For each odd block $J_i$, a \emph{$t$-linkage} is some collection $\lset_i$ of $t$ disjoint paths contained in $J_i$, that connect some subsets $X\subset X_i$, $Y\subset Y_i$ of $t$ vertices each to each other.

For each odd block $J_i$, we will define a $t$-linkage $\lset_i$, and for each even block $M_i$, we will select a subset $E'_i\subset E_i$ of $t$ edges, in such a way that, by concatenating $\lset_1,E'_1,\lset_2,\ldots,\lset_{T},E'_T$, we obtain a collection $\pset$ of $t$ disjoint paths. For every pair $P,P'\in \pset$ of such paths, we will ensure that there is an edge $e_{P,P'}$ in $H^*$, with one endpoint in $P$ and one endpoint in $P'$. It is then easy to see that $H^*$ (and hence $G$) contains a model of a $K_t$-minor, where every vertex of $K_t$ is mapped to a distinct path of $\pset$. We also describe an efficient algorithm to construct the set $\pset$ of paths in graph $H^*$. 

It now remains to define the linkages $\lset_i$, and the subsets $E'_i$ of edges. We perform $T$ iterations. In iteration $i$, we define the linkage $\lset_i$, and the set $E_i'\subset E_i$ of edges. When iteration $i$ is completed, we obtain a collection $\pset_{i}$ of $t$ disjoint paths, by concatenating $\lset_1,E_1',\ldots,\lset_i,E_i'$. Each path $P\in \pset_i$ starts at a vertex of the first column of $H^*$, and terminates at a distinct vertex of $X_{i+1}$. Let $X'_{i+1}\subset X_{i+1}$ be subset of vertices where the paths in $\pset_i$ terminate. We will ensure that $X'_{i+1}$ is a consecutive subset of vertices on $X_{i+1}$ (using the natural ordering of the vertices in $X_{i+1}$ along the corresponding column of the grid). The natural ordering of the vertices of $X_{i+1}'$ then defines an ordering $\sigma_i$ of the paths in $\pset_i$: path $P$ appears before path $P'$ in this ordering iff the last endpoint of $P$ appears before the last endpoint of $P'$ in $\sigma_i$. Notice that any pair $P,P'\in \pset_i$ of paths that appears consecutively in this ordering, has an edge $e$ connecting a vertex of $P$ to a vertex of $P'$ - the edge of the first column of $J_{i+1}$, that connects the two corresponding vertices of $X'_{i+1}$.
We will next define a collection $\sigma_0',\ldots,\sigma_T'$ of permutations of $[t]$, such that for each $1\leq r \neq q\leq t$, there is some index $0\leq i\leq T$, with $r$ and $q$ appearing consecutively in $\sigma'_i$. We will then define the linkages $\lset_1,\ldots,\lset_T$, and the sets $E_1',\ldots,E'_T$ of edges in such a way, that for each $1\leq i\leq T$, the permutation $\sigma_i$ defined by the paths in $\pset_i$ is exactly $\sigma'_i$.
This will ensure that for every pair $P_{r},P_q\in \pset$ of paths, there is some index $i$, such that, if $v$ is a vertex of the first column of $J_{i}$ lying on $P_{r}$, and $v'$ is defined similarly for $P_q$, then $v$ and $v'$ are connected by a vertical edge in the first column of $J_{i}$.
We now proceed to define the desired set of permutations.
 
Suppose we are given two permutations $\sigma,\sigma'$ of $\set{1,\ldots,t}$, where $\sigma=(x_1,\ldots,x_t)$ and $\sigma'=(x'_1,\ldots,x'_t)$ (we view each permutation as an ordered set of elements from $\set{1,\ldots,t}$). We say that $\sigma$ and $\sigma'$ are an $i$-swap, for $1\leq i<t$, iff for all $j\not\in\set{i,i+1}$, $x_j=x'_j$, while $x_i=x'_{i+1}$ and $x_{i+1}=x'_i$. In other words, we can obtain $\sigma'$ from $\sigma$ by swapping the elements at the locations $i$ and $i+1$. We say that $\sigma,\sigma'$ are a \emph{swap}, iff they are an $i$-swap for some $1\leq i<t$.

Suppose we are given a sequence $\sigma_0,\sigma_1,\sigma_2,\ldots,\sigma_T$ of permutations of $\set{1,\ldots,t}$. We say that a pair $\set{r,q}$, for $1\leq r\neq q\leq t$ is \emph{explored}, iff there is some $0\leq i\leq T$, such that $r$ and $q$ appear consecutively in $\sigma_i$. We need the following claim.

\begin{claim}\label{claim: sequence of permutations with swaps}
There is a sequence of $T+1$ permutations $\sigma_0,\sigma_1,\ldots,\sigma_T$ of $\set{1,\ldots,t}$, such that:
\begin{itemize}
\item For each $1\leq i<T$, $\sigma_i$ and $\sigma_{i+1}$ are a swap;

\item Every pair $\set{r,q}$, for $1\leq r\neq q\leq t$ is explored; 

\item $\sigma_0=(1,\ldots,t)$ and $\sigma_1=(2,1,3,4,\ldots,t)$.
\end{itemize}
\end{claim}

\begin{proof}
We partition the sequence $\sigma_0,\ldots,\sigma_T$ of $T+1$ permutations into $t$ subsets, $S_0,S_1,\ldots,S_{t-1}$. Set $S_0$ only contains one permutation $\sigma_0=(1,2,3,\ldots,t)$. Set $S_1$ contains the following $t-1$ permutations, set $S_2$ the following $t-2$ permutations, and so on, until the last set $S_{t-1}$ that contains one permutation.

In order to define the permutations of $S_1$, we start with $\sigma_0$. For each $1\leq i\leq t-1$, the $i$th permutation of $S_1$ is obtained by performing an $i$-swap of the preceding permutation. In other words,  we obtain $\sigma_1$ from $\sigma_0$ by swapping the elements $1$ and $2$; $\sigma_2$ is obtained from $\sigma_1$ by swapping the elements $1$ and $3$, and so on. The last permutation of $S_1$ is $(2,3,4,\ldots,t,1)$. It is clear that by the time we reach the last permutation of $S_1$, all pairs $\set{1,i}$ for all $2\leq i\leq t$ have been explored.

We then define the permutations of $S_2$ similarly: for $1\leq i\leq t-2$, the $i$th permutation of $S_2$ is obtained by an $i$-swap of the preceding permutation. The last permutation of $S_2$ is $(3,4,\ldots,t,2,1)$, and by the end of the last permutation of $S_2$, all pairs $\set{2,i}$, for $3\leq i\leq t$ have been explored. We define the remaining sets $S_i$ of permutations similarly.
\end{proof}

We are now ready to complete the proof of Lemma~\ref{lemma: from H1 to clique}. Let $S=(\sigma_0',\sigma'_1,\ldots,\sigma'_T)$ be the sequence of permutations given by Claim~\ref{claim: sequence of permutations with swaps}.
We start by choosing $X'_1\subset X_1$ to be the subset of the first $t$ vertices on the first column of $J_1$. The first set $\pset_0$ of paths contains $t$ paths, where each path consists of a single distinct vertex from $X'_1$. We label these paths $P_1,\ldots,P_t$ according to the order in which the corresponding vertices appear in $X'_1$. For each $1\leq i\leq t$, we will ensure that the ordering $\sigma_i$ of the paths in $\pset_i$ defined by the ordering of their endpoints in $X'_{i+1}$ is exactly the same as the permutation $\sigma'_i$.

Consider some $1\leq i\leq T$, and assume that we are given a collection $X'_i\subset X_i$ of vertices where the paths in $\pset_{i-1}$ terminate. Assume further that $\sigma'_{i}$ is a $j$-swap of $\sigma'_{i-1}$, where $1\leq j<t$. We define a subset $Y'_{i}\subset Y_i$ of vertices as follows.  Denote the vertices of $Y_i$ by $v_1,\ldots,v_{2t}$, where $v_{\ell}$ is the vertex lying in row $\ell$. The idea is to select a consecutive set of $t$ vertices of $Y_i$, such that $v_t$ is the $j$th vertex in this ordered set. Formally, $Y_i'=\set{v_{t-j+1},v_{t-j+2},\ldots,v_{2t-j}}$. Notice that $v_{t}$ and $v_{t+1}$ are the $j$th and the $(j+1)$th vertices of $Y'_i$, respectively. Since $J_i$ is $t$-linked, we can find an $X'_i$-$Y'_i$ linkage $\lset_i$ in $J_i$. Let $\pset_i'$ be the set of paths obtained by concatenating $\pset_{i-1}$ and $\lset_i$. Since $J_i$ is a planar graph, the ordering in which the endpoints of the paths of $\pset_{i-1}$ appear in $X'_i$, and the ordering in which the endpoints of paths in $\pset_i'$ appear in $Y'_i$ are the same. In order to define $E_i'$, we select, for each $v\in Y_i'$, a single edge $e\in E_i$ incident on $v$. For $v\not\in\set{ v_{t},v_{t+1}}$, we select the unique edge of $E_i$ (the horizontal edge of the grid) incident on $v$, and for $v_{t},v_{t+1}$, we select the cross edges incident on these vertices. By concatenating $\pset_i'$ with $E_i'$, we obtain a new set $\pset_i$ of paths, and their corresponding ordering $\sigma_i$. It is easy to see that $\sigma_i$ is obtained from $\sigma_{i-1}$ via a $j$-swap, and therefore, $\sigma_i=\sigma'_i$ as required.

This completes the construction of the $K_t$-minor in $H^*$. Assume now that $J_1$ is a contraction of a sub-wall $W'$ of $W$. We claim that $W$ grasps $K_t$. Indeed, it is easy to see that every path in $\pset$ intersects at least $t$ columns of $J_1$. Therefore, from Observation~\ref{obs: clique grasped by wall}, there is a model of $K_t$ grasped by $W$. Our proof also gives an efficient algorithm, that finds the model of $K_t$ in $H^*$. Therefore, if we are given a model of $H^*$ in $G$, then we can efficiently compute the model of $K_t$ in $G$ with the required properties.
\end{proof}

\subsection{Graph Families $\hset_1,\hset_2,\hset_3$}

In this section, we define three graph families $\hset_1,\hset_2,\hset_3$. We will show that if $G$ contains a graph from any of these families as a minor, then it contains a $K_t$-minor. Before we proceed to define these families of graphs, we need a few definitions.

Let $G'$ be the $(h\times r)$-grid. As with walls, we define sub-grids of $G'$ spanned by subsets of rows and columns of $G'$. For any consecutive subset $\rset'$ of the rows of $G'$, the sub-grid of $G'$ spanned by $\rset'$ is $G'[S]$, where $S$ contains all vertices $v(i,j)$ with $R_i\in \rset'$. Similarly, given any consecutive subset $\rset'$ of the rows of $G'$, and a consecutive subset $\cset'$ of the columns of $G'$, the sub-grid of $G$ spanned by the rows in $\rset'$ and the columns in $\cset'$ is $G[S']$, where $S'$ contains all vertices $v(i,j)$ with $R_i\in \rset'$ and $C_j\in \cset'$.

Assume now that we are given two vertices $v(i,j)$ and $v(i',j')$ of the grid $G'$, where $j\leq j'$. We say that $v(i,j)$ and $v(i',j')$ are \emph{separated} by column $C_{j''}$ of the grid iff $j<j''<j'$. We say that they are separated by $x$ columns of $G'$ iff at least $x$ distinct columns $C_{j''}$ separate $v(i,j)$ from $v(i',j')$, or, equivalently, $j'-j>x$.

\subsection{Graph Family $\hset_1$}
A graph $H$ belongs to the family $\hset_1$ iff $H$ is the union of the $(h\times r)$ grid $G'$, where $h>2t$, and a set $E'$ of $T$ edges, such that the following additional conditions hold. Let $G_1$ be the sub-grid of $G'$ spanned by the top $t$ rows, $G_2$ the sub-grid of $G'$ spanned by the bottom $t$ rows, and $G_3$ the sub-grid of $G'$ spanned by the remaining rows. Let $X$ be the set of all endpoints of the edges in $E'$. Then the following conditions must hold: 

\begin{itemize}
\item $X\subseteq V(G_3)$, and $|X|=2T$, so all edges in $E'$ have distinct endpoints.

\item Every pair of vertices in $X$ is separated by at least $t+2$ columns, and no vertex of $X$ belongs to the first column of $G'$.
\end{itemize}

Let $B_1$ be the sub-grid of $G'$ spanned by the first $t$ columns of $G'$ and all rows of $G'$.

\begin{theorem}\label{thm: family H1 to clique}
Let $G$ be any graph, and assume that it contains a graph $H\in \hset_1$ as a minor. Then $G$ contains a $K_t$-minor. Moreover, if $G$ contains a wall $W$, and $B_1$ is a contraction of a sub-wall of $W$, then $G$ contains a model of a $K_t$-minor grasped by $W$, and this model can be found efficiently given a model of $H$ in $G$.
\end{theorem}

\begin{proof}
Let $\mset$ be the set of all unordered pairs of distinct elements of $\set{1,\ldots,t}$, that is, $\mset=\set{\set{i,j}\mid 1\leq i,j\leq t, i\neq j}$,  and $|\mset|=|E'|=T$. We assign labels $\ell(e)\in \mset$ to the edges $e\in E'$, so that every label in $\mset$ is assigned to exactly one edge. We then assign labels $\ell(u)$ to the vertices $u\in X$, as follows. For each edge $e=(u,u')\in E'$, if $\ell(u,u')=\set{i,j}$, then we assign label $i$ to $u$ and label $j$ to $u'$, breaking the symmetry arbitrarily. 

We will construct $t$ disjoint paths $P_1,\ldots,P_t$ in the grid $G'$ in such a way that for each $1\leq i\leq t$, every vertex of $X$ with label $i$ belongs to $P_i$. It is then immediate to find a model of the $K_t$-minor in $H$: every vertex of $K_t$ is mapped to a distinct path $P_i$, and the edges of $K_t$ are mapped to the edges of $E'$. Since for each pair $\set{i,j}$ with $1\leq i\neq j\leq t$, there is an edge $e\in E'$ whose endpoints are labeled $i$ and $j$ respectively, every pair $P_i$, $P_j$ of distinct paths is connected with an edge in $E'$.

Therefore, in order to show that $G$ contains a $K_t$-minor, it is now enough to construct the set $\set{P_1,\ldots,P_t}$ of disjoint paths, such that for each $1\leq i\leq t$, every vertex of $X$ with label $i$ belongs to $P_i$. We assume that $X=\set{v(i_1,j_1),v(i_2,j_2),\ldots,v(i_{2T},j_{2T})}$, where $j_1<j_2<\ldots<j_{2T}$. For each $1\leq q\leq 2T$, let $M_q$ be the sub-grid of $G'$ spanned by columns $C_{j_q-1},C_{j_q},C_{j_q+1}$. For $1<q\leq 2T$,  let $J_q$ be the sub-grid of $G'$ spanned by columns $C_{j_{q-1}+1},\ldots,C_{j_q-1}$, and let $J_1$ be the sub-grid of $G'$ spanned by columns $C_1,\ldots,C_{j_1-1}$. Recall that from the definition of $\hset_1$, for $1< q\leq 2T$, $J_q$ contains at least $t+2$ columns.

Fix some $1\leq q\leq 2T$ and let $\ell(v(i_q,j_q))=\ell$. Since all vertices of $X$ lie in $G_3$, $t<i_q\leq  h-t$. Let $\rset_q$ be a set of $t$ consecutive rows, such that $R_{i_q}$ is the $\ell$th row in this set. That is, $\rset_q=\set{R_{i_q-\ell+1},\ldots,R_{i_q-\ell +t}}$. Let $S_q$ be the set of $t$ vertices in column $C_{j_q-1}$ that belong to the rows of $\rset_q$, and let $T_q$ be the set of $t$ vertices in column $C_{j_q+1}$ that belong to the rows of $\rset_q$.

We are now ready to define our set $\pset=\set{P_1,\ldots,P_t}$ of paths. The paths in $\pset$ are obtained by concatenating $4T$ sets of paths: $\pset_1,\pset_1',\pset_2,\pset_2',\ldots,\pset_{2T},\pset_{2T}'$. For each $1\leq q\leq 2T$, set $\pset_q'$ contains, for each row $R\in \rset_q$, the segment of $R$ from the unique vertex in $S_q\cap V(R)$ to the unique vertex in $T_q\cap V(R)$. From the definition of $\pset_q'$, if $\ell(v(i_q,j_q))=i$, then the $i$th path of $\pset_q'$ from the top contains $v(i_q,j_q)$. We now turn to define the path sets $\pset_q$ for $1\leq q\leq 2T$. Assume first that $q>1$. The set $\pset_q$ of paths is defined to be a set of $q$ disjoint paths connecting the vertices of $T_{q-1}$ to the vertices of $S_q$ in graph $J_q$. Notice that from Claim~\ref{claim: t-linkedness of grid}, since $J_q$ contains at least $t+2$ columns, such a set of paths exists. Finally, assume that $q=1$. Then $\pset_1$ contains $t$ rows of $J_1$, which are sub-paths of the rows in $\rset_1$. The final set $\pset$ of paths is obtained by concatenating the paths in $\pset_1,\pset_1',\pset_2,\pset_2',\ldots,\pset_{2T},\pset_{2T}'$. It is immediate to verify that $\pset$ contains $t$ disjoint paths. We assume that $\pset=\set{P_1,\ldots,P_t}$, where the paths are indexed in their natural top-to-bottom ordering. From our construction, for each $1\leq i\leq t$, every vertex of $X$ with label $i$ belongs to $P_i$.  Notice that every path in $\pset$ intersects every column of $B_1$. Therefore, if $B_1$ is a contraction of a sub-wall of $W$, then the $K_t$-minor that we have constructed is grasped by $W$. Our proof also gives an efficient algorithm, that finds the model of $K_t$ in $H$. Therefore, if we are given a model of $H$ in $G$, then we can efficiently compute the model of $K_t$ in $G$ with the required properties.
\end{proof}

\subsection{Graph Family $\hset_2$}
A graph $H$ belongs to the family $\hset_2$ iff $H$ is the union of the $(h\times r)$ grid $G'$, where $h>4t$, and a set $E'$ of $2T+2$ edges, and the following conditions hold. Let $G_1$ be the sub-grid of $G'$ spanned by the top $t$ rows, $G_2$ the sub-grid of $G'$ spanned by the bottom $t$ rows, and $G_3$ the sub-grid of $G'$ spanned by rows $\set{R_{2t+1},\ldots,R_{h-2t}}$.
We assume that $E'=\set{e_1,\ldots,e_{2T+2}}$, and for each $1\leq i\leq 2T+2$, the endpoints of $e_i$ are labeled as $x_i$ and $y_i$. Let $X=\set{x_i\mid 1\leq i\leq 2T+2}$, and $Y=\set{y_i\mid 1\leq i\leq 2T+2}$.
Then the following conditions must hold: 

\begin{itemize}

\item $X\cup Y$ contains $4T+4$ distinct vertices.

\item $X\subseteq V(G_3)$, and every pair of vertices in $X$ is separated by at least $t+2$ columns.

\item $Y\subseteq V(G_1)\cup V(G_2)$.
\end{itemize}

Let $B_1$ be the sub-grid of $G'$ spanned by the first $t$ columns of $G'$ and all rows of $G'$.

\begin{theorem}\label{thm: family H2 to clique}
Let $G$ be any graph, that contains a graph $H\in \hset_2$ as a minor. Then $G$ contains a $K_t$-minor. Moreover, if $G$ contains a wall $W$, and $B_1$ is a contraction of a sub-wall of $W$, then $G$ contains a model of a $K_t$-minor grasped by $W$, and this model can be found efficiently given a model of $H$ in $G$.
\end{theorem}

\begin{proof}
Let $Y_1=Y\cap V(G_1)$ and $Y_2=Y\cap V(G_2)$. We build two paths: path $P_1\subseteq G_1$ containing all vertices of $G_1$, and path $P_2\subseteq G_2$, containing all vertices of $G_2$. In order to construct $P_1$, we start with $R_1$, and add the edge $(v(1,r),v(2,r))$ to $P_1$. We then add the row $R_2$ and the edge $(v(2,1),v(3,1))$. We continue in this fashion, until $P_1$ contains all rows $R_1,\ldots, R_t$. Path $P_2$ is constructed similarly in $G_2$. Clearly, $Y_1\subseteq V(P_1)$ and $Y_2\subseteq V(P_2)$. We can find $\floor{|Y_1|/2}$ disjoint sub-paths of $P_1$, such that for each sub-path, both its endpoints belong to $Y_1$. Let $\mset_1$ be this set of paths. Similarly, we can find a set $\mset_2$ of $\floor{|Y_2|/2}$ disjoint sub-paths of $P_2$, such that for each sub-path, both its endpoints belong to $Y_2$. Let $\mset=\mset_1\cup \mset_2$. Then $|\mset|\geq T$. If $|\mset|>T$, then we remove arbitrary paths from $\mset$ until $|\mset|=T$ holds. Consider some path $P\in\mset$, and let $y,y'$ be its endpoints. Let $e\in E'$ be the unique edge incident on $y$, and let $e'\in E'$ be the unique edge incident on $y'$. We define $P'$ to be the concatenation of $e,P$ and $e'$. Let $\mset'=\set{P'\mid P\in \mset}$.

We now construct a minor $H'$ of $H$ as follows. Graph $H'$ is the union of the sub-grid $G'\setminus (G_1\cup G_2)$ of $H$, and a set $E''$ of edges, where for each path $P'\in \mset'$ with endpoints $x$ and $y$, $E''$ contains the edge $(x,y)$. It is immediate to verify that $H'\in \hset_1$. Since $H'$ is a minor of $G$, from Theorem~\ref{thm: family H1 to clique}, it contains a $K_t$-minor. Let $B'_1$ be the sub-grid of $G'\setminus (G_1\cup G_2)$, spanned by its first $t$ columns. Then, if $B_1$ is a contraction of a sub-wall of $W$, so is $B_1'$, and therefore there is a model of a $K_t$-minor grasped by $W$ in $G$. Our proof also gives an efficient algorithm, that finds the model of $K_t$ in $H$. Therefore, if we are given a model of $H$ in $G$, then we can efficiently compute the model of $K_t$ in $G$ with the required properties.
\end{proof}

\subsection{Graph Family $\hset_3$}
A graph $H$ belongs to the family $\hset_3$ iff $H$ is the union of the $(h\times r)$ grid $G'$, where $h>4t$, and a set $E'$ of $10T+6$ edges, and the following conditions hold. Let $G_1$ be the sub-grid of $G'$ spanned by the top $t$ rows, $G_2$ the sub-grid of $G'$ spanned by the bottom $t$ rows, and $G_3$ the sub-grid of $G'$ spanned by rows $\set{R_{2t+1},\ldots,R_{h-2t}}$.
We assume that $E'=\set{e_1,\ldots,e_{10T+6}}$, and for each $1\leq i\leq 10T+6$, the endpoints of $e_i$ are labeled as $x_i$ and $y_i$. We call $x_i$ the $x$-endpoint of $e_i$ and $y_i$ its $y$-endpoint. Let $X=\set{x_i\mid 1\leq i\leq 10T+6}$, and $Y=\set{y_i\mid 1\leq i\leq 10T+6}$.
Then the following conditions must hold: 

\begin{itemize}

\item $X\cup Y$ contains $20T+12$ distinct vertices.

\item $X\subseteq V(G_3)$, and every pair of vertices in $X$ is separated by at least $2t+2$ columns in $G'$.

\item For each $1\leq i\leq 10T+6$, $x_i$ and $y_i$ are separated by at least $t+1$ columns in $G'$.

\item No vertex of $X\cup Y$ lies in the first $t$ columns, or in the last column of $G'$.
\end{itemize}

Let $J_1$ be the sub-grid of $G'$ spanned by the first $t+1$ columns of $G'$ and all rows of $G'$.

\begin{theorem}\label{thm: family H3 to clique}
Let $G$ be any graph, that contains a graph $H\in \hset_3$ as a minor. Then $G$ contains a $K_t$-minor. Moreover, if $G$ contains a wall $W$, and $J_1$ is a contraction of a sub-wall of $W$, then $G$ contains a model of a $K_t$-minor grasped by $W$, and this model can be found efficiently given a model of $H$ in $G$.
\end{theorem}

\begin{proof}
Let $E_1\subseteq E'$ be the set of edges $(x_i,y_i)$, where $y_i\in V(G_1)\cup V(G_2)$. Assume first that $|E_1|\geq 2T+2$. Then $G'\cup E_1$ is a graph from family $\hset_2$, and invoking Theorem~\ref{thm: family H2 to clique} finishes the proof. From now on we assume that $|E_1|<2T+2$.

Let $E''=E'\setminus E_1$, so $|E''|\geq 8T+4$.
Our next step is to partition the grid $G'$ into blocks, spanned by disjoint consecutive sets of columns. Consider some edge $e=(x,y)\in E''$, and let $C_{i_e}$ be the column to which $x$ belongs. We define the block $B(e)$ to be the sub-grid of $G'$ spanned by columns $C_{i_e-t},\ldots,C_{i_e+t}$. Notice that since every pair of vertices of $X$ is separated by at least $2t$ columns, for $e\neq e'$, the blocks $B(e)$ and $B(e')$ are disjoint. We also define a special block $B^*$, that is spanned by the first $t$ columns of $G'$.
Let $\bset$ be the set of blocks we have constructed so far. Consider any consecutive pair $B,B'$ of blocks in $\bset$. Assume that $B$ appears to the left of $B'$, and let $C_i$ be the last column of $B$ and $C_{i'}$ the first column of $B'$.
If $i'=i+2$, then we extend $B$ to contain column $C_{i+1}$. Otherwise,
 if $i'>i+2$, then we add a new block $B''$ to $\bset$, spanned by columns $C_{i+1},\ldots,C_{i'-1}$. 
 Let $B$ be the last block in $\bset$, and let $C_{i}$ be the last column of $B$. If $i<r$, then we add a new block spanned by columns $C_{i+1},\ldots, C_r$ to $\bset$. This finishes the definition of the set $\bset$ of blocks. Clearly, all blocks in $\bset$ are disjoint, and every column of $G'$ belongs to some block. Each block, except for possibly the last one, contains at least two columns. As observed before, each block $B\in \bset$ may contain at most one vertex of $X$ (and possibly several vertices of $Y$). However, for each edge $e=(x,y)\in E''$, since $x$ and $y$ are separated by at least $t+1$ columns,  they must belong to distinct blocks. Block $B^*$ contains no vertices of $X\cup Y$. Our next step is to select a subset $E^*\subseteq E''$ of at least $2T+1$ edges, such that each block $B\in \bset$ may contain an $x$-endpoint of an edge in $E^*$, or $y$-endpoints of edges in $E^*$, but not both. 

\begin{claim}
We can efficiently find a subset $E^*\subseteq E''$ of $|E''|/4\geq 2T+1$ edges, such that, if we denote by $X^*$ the $x$-endpoints of the edges in $E^*$, and by $Y^*$ their $y$-endpoints, then for each block $B\in \bset$, either $V(B)\cap X^*=\emptyset$, or $V(B)\cap Y^*=\emptyset$.
\end{claim}

\begin{proof}
Consider any edge $e=(x,y)\in E''$. Let $B$ be the block to which $x$ belongs, and $B'$ the block to which $y$ belongs, so, as observed before, $B\neq B'$. We say that $e$ is \emph{directed towards right} if $B$ appears before $B'$ in the natural left-to-right ordering of the blocks; otherwise, we say that it is directed towards left. At least half the edges in $E''$ are all directed towards right, or are all directed towards left. We assume w.l.o.g. that at least half the edges of $E''$ are directed towards right. Let $\tilde E$ be this set of edges.


Next, we build a directed graph $\tilde H$, with $V(\tilde H)=\set{v_B\mid B\in \bset}$. There is an edge $(v_B,v_{B'})$ in $\tilde H$  iff there is an edge $(x_i,y_i)\in \tilde E$ with $x_i\in B$ and $y_i\in B'$. Therefore, there is a natural mapping between the edges of $\tilde H$ and the edges of $\tilde E$, and we do not distinguish between those edges. Let $\tilde H'$ be the underlying undirected graph for $\tilde H$. We claim that graph $\tilde H'$ is a forest. Indeed, recall that each block may contain at most one vertex of $X$, and so each vertex of $\tilde H$ has at most one outgoing edge. Assume for contradiction that $\tilde H'$ is not a forest. Then there is a cycle $C$ in $\tilde H'$. Let $\bset'=\set{B\mid v_B\in V(C)}$, and let $\hat B\in \bset$ be the left-most block in the set $\bset'$. Then there are two edges incident on $v_{\hat B}$ in $C$. Both of these edges must be directed away from $v_{\hat B}$ in $\tilde H$, since all edges of $\tilde E$ are directed towards right. But every vertex of $\tilde H$ has at most one outgoing edge, a contradiction.
We conclude that $\tilde H'$ is a forest. Consider some tree $\tau$ of $\tilde H'$, and let $\tau'$ be the corresponding sub-graph of $\tilde H$. Then, since every vertex of $\tilde H$ has at most one outgoing edge, and there are no cycles in $\tau$, there is a vertex $x\in V(\tau')$ that has no outgoing edges. Rooting $\tau'$ at vertex $x$, it is easy to see that $\tau'$ is an arborescence, where all edges are directed towards the root.

We now turn to construct the final set $E^*$ of edges. We start with $E^*=\emptyset$, and then process the arborescences of $\tilde H$ one-by-one. Consider some such rooted arborescence $\tau$. Let $e=(u,v)$ be a directed edge of $\tau$. We say that $e$ belongs to level $j$ iff the unique path in $\tau$ from $u$ to the root of $\tau$ has $j$ edges. Let $E_1(\tau)$ be the subset of edges of $\tau$ lying in odd-indexed levels, and $E_2(\tau)$ the subset of edges of $\tau$ lying in even-indexed levels. If $|E_1(\tau)|\geq |E_2(\tau)|$, then we add the edges of $E_1(\tau)$ to $E^*$, and otherwise we add the edges of $E_2(\tau)$ to $E^*$. Notice that for each vertex $v$ of $\tau$, each of the sets $E_1(\tau)$ and $E_2(\tau)$ may contain either the edge leaving $v$, or all the edges entering $v$, but not both types of edges. Once all arborescences $\tau$ in $\tilde H$ are processed, we obtain our final set $E^*$ of edges with $|E^*|\geq |\tilde E|/2\geq |E'' |/4$. For every block $B\in \bset$, $E^*$ may contain an edge leaving $v_B$, or a number of edges entering $v_B$, but not both. Therefore, if $X^*=\set{x_i\mid (x_i,y_i)\in E^*}$ and $Y^*=\set{y_i\mid (x_i,y_i)\in E^*}$, then for each block $B\in \bset$, $V(B)\cap X^*=\emptyset$ or $V(B)\cap Y^*=\emptyset$ must hold.
\end{proof}

Let $X^*$ be the set of the $x$-endpoints of the edges in $E^*$, and $Y^*$ the set of their $y$-endpoints. We say that a block $B\in \bset$ is \emph{an $x$-block} iff $V(B)\cap X^*\neq \emptyset$. Otherwise, we call it a \emph{$y$-block}. In the rest of the proof, we will assume that $h-t$ is an odd number. If this is not the case, then we delete the bottom row from the grid $G'$, and adjust $h$ accordingly. Our next step is to partition the grid $G'$ into two sub-graphs. We do so by first partitioning the vertices of $G'$ into two subsets: $U$ and $U'$. Set $U$ of vertices contains, for each $x$-block $B$, all vertices in the top $h-1$ rows of $B$ (that is, all but the last row of $B$). It also contains all vertices of the special block $B^*$. For each $y$-block $B$, set $U$ contains all vertices lying in the top $t$ rows of $B$. We then set $U'=V(G')\setminus U$, and we define two sub-graphs of $G'$: $G_1'=G'[U]$ and $G_2'=G'[U']$.
Notice that for each edge $e=(x_i,y_i)\in E^*$, $x_i\in U$ and $y_i\in U'$ must hold.
 The rest of the proof consists of three steps. 
 
 In the first step, we will construct a path $P$ in graph $G_2'$ that contains all vertices of $G_2'$. In order to do so, we will view $P$ as a graph. We will start with $P=\emptyset$, and will gradually add edges and vertices to $P$, until it becomes a path with the desired properties. For each $x$-block $B$, we add the bottom row of $B$ to $P$. Consider now some $y$-block $B'$, and assume that it contains at least two columns.  Let $a$ and $b$ be its bottom left and right corners of $B'$, respectively. We build a path $P(B')$ whose vertex set is $V(B')\cap U'$, and whose endpoints are $a$ and $b$. Such a path can be constructed, for example, by following the first column of $B'$ until the row $R_{t+1}$ of $B'$ is reached, and then traversing the rows $R_{t+1},\ldots,R_h$ of $B'$ in a snake-like fashion (see Figure~\ref{fig: traversing the grid}). (We use the fact that $h-t$ is odd here). Since all blocks, except for possibly the last one, contain at least $2$ columns, we can build path $P(B')$ for each such block. If the last block $B''$ contains only one column, then we set $P(B'')=B''\cap G_2'$.

 \begin{figure}[h]
 \scalebox{0.4}{\includegraphics{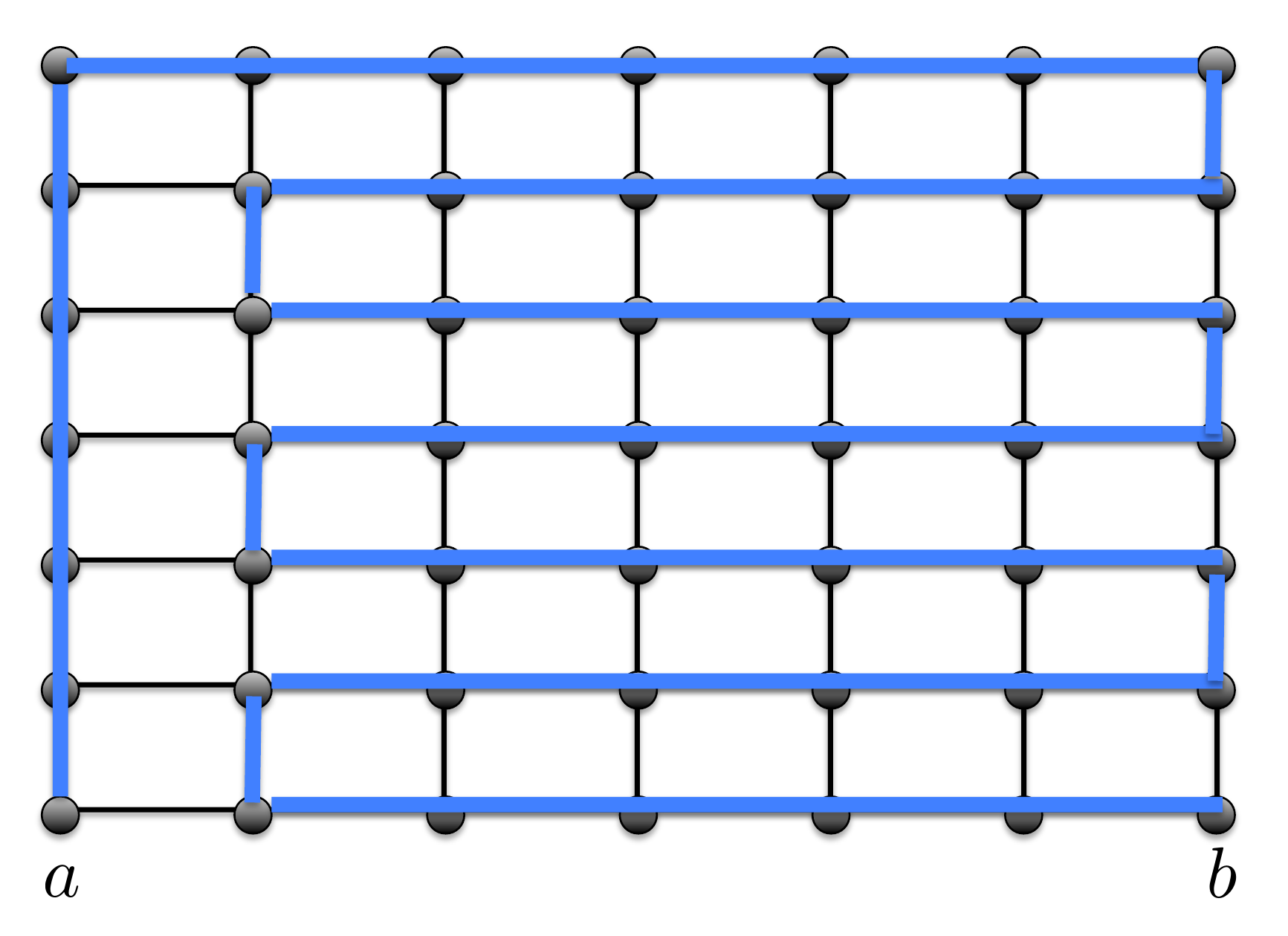}}\caption{Path $P(B')$ \label{fig: traversing the grid}}
 \end{figure}
 
 Finally, for every consecutive pair $B,B'$ of blocks, we add an edge connecting the bottom right corner of $B$ to the bottom left corner of $B'$. It is easy to see that the resulting graph $P$ is a path contained in $G_2'$, with $V(P)=V(G_2')$. In particular, $Y^*\subseteq V(P)$. Therefore, we can find $\floor{ |Y^*|/2}\geq T$ disjoint segments of $P$, where for each segment both endpoints belong to $Y^*$, and it is internally disjoint from $Y^*$. Let $\Sigma$ be this set of segments of $P$. If $|\Sigma|>T$, we discard paths from $\Sigma$ until $|\Sigma|=T$ holds. Consider some segment $\sigma\in \Sigma$, and assume that its endpoints are $y_i$ and $y_j$, where $y_i$ is an endpoint of some edge $e_i=(x_i,y_i)\in E^*$, and $y_j$ is an endpoint of some edge $e_j=(x_j,y_j)\in E^*$. Concatenating $e_i,\sigma,e_j$, we obtain a path $Q(\sigma)$ whose endpoints are $x_i$ and $x_j$, and $Q(\sigma)$ is internally disjoint from $G_1'$. Let $\qset=\set{Q(\sigma)\mid \sigma\in \Sigma}$.
 
 The rest of the proof is very similar to the proof of Theorem~\ref{thm: family H1 to clique}. Let $\mset$ be the set of all unordered pairs of distinct elements of $\set{1,\ldots,t}$, $\mset=\set{\set{i,j}\mid 1\leq i,j\leq t, i\neq j}$, so $|\mset|=|\qset|=T$. We assign labels $\ell(Q)\in \mset$ to the paths of $\qset$, so that every label in $\mset$ is assigned to exactly one path. We then assign labels $\ell(u)$ to all vertices  $u\in X^*$, as follows. For each path $Q\in \qset$ with endpoints $u$ and $u'$, if $\ell(Q)=\set{i,j}$, then we assign label $i$ to $u$ and label $j$ to $u'$, breaking the symmetry arbitrarily. 

As in the proof of Theorem~\ref{thm: family H1 to clique},
we will construct $t$ disjoint  simple paths $P_1,\ldots,P_t$ in the graph $G'_1$ in such a way that for each $1\leq i\leq t$, every vertex of $X^*$ with label $i$ belongs to $P_i$. It is then immediate to find a model of a $K_t$-minor in $H$: every vertex of $K_t$ is mapped to a distinct path $P_i$, and the edges of $K_t$ are mapped to the paths of $\qset$. Since for each pair $\set{i,j}$ with $1\leq i\neq j\leq t$, there is a path $Q\in \qset$ whose endpoints are labeled $i$ and $j$ respectively, every pair $P_i$, $P_j$ of distinct paths is connected with a path in $\qset$.

In order to construct the set $\pset=\set{P_1,\ldots,P_t}$ of paths, we will view $\pset$ as a graph. We start with $\pset=\emptyset$, and gradually add vertices and edges to it, until it becomes a set of $t$ disjoint paths with the required properties. In our first step, for each $y$-block $B$, we add the top $t$ rows of $B$ to $\pset$. We also add the top $t$ rows of the special block $B^*$ to $\pset$. Consider now some $x$-block $B$. Recall that $B$ contains at least $2t+1$ columns, that we label $C_1,\ldots,C_{t'}$ with $t'\geq 2t+1$ in their natural order, and there is one vertex $x\in X^*$ that belongs to $B$, which must lie in column $C_{t+1}$. Let $R_i$ be the row in which $x$ lies. Since $x$ belongs to $G_3$, $2t<i<h-2t+2$ (this takes into account the possibility that the bottom row has been deleted if $h-t$ was initially even). We denote the label of $X$ by $\ell$. Let $L$ be the set of the top $t$ vertices in the column $C_1$ and $M$ the set of the top $t$ vertices in the column $C_{t'}$ of $B$. We will construct a set $\pset(B)$ of $t$ disjoint paths contained in $B\cap G_1'$, such that each path starts from a vertex of $L$ and terminates at a vertex of $M$. Moreover, the $\ell$th path from the top will contain the vertex $x$.

In order to do so, we partition $B\cap G'_1$ into three sub-blocks: $B_1$ spanned by columns $C_1,\ldots,C_{t}$, $B_2$ spanned by columns $C_{t},\ldots,C_{t+2}$, and $B_3$ spanned by columns $C_{t+2},\ldots,C_{t'}$. Let $\rset(B)$ be a set of $t$ consecutive rows of $B$, such that $R_i$ is the $\ell$th row of $\rset(B)$. That is, $\rset(B)=\set{R_{i-\ell+1},\ldots,R_{i-\ell+t}}$ (since $2t< i< h-2t+2$, all rows of $\rset_B$ are contained in $G_1'$). Let $L'\subseteq V(C_{t})$ be the set of the $t$ vertices lying in the rows of $\rset(B)$, and let $M'\subseteq V(C_{t+2})$ be the set of the $t$ vertices lying in the rows of $\rset_B$. Set $\pset(B)$ is a concatenation of three sets of paths: $\pset(B_1),\pset(B_2),\pset(B_3)$, that are defined as follows. Set $\pset(B_1)$ is a set of $t$ disjoint paths contained in $B_1$, that connect the vertices of $L$ to the vertices of $L'$ (since $B_1$ has more than $t$ rows and at least $t$ columns, from Claim~\ref{claim: t-linkedness of grid} such a set of paths exists). Similarly, $\pset(B_3)$ is a set of $t$ disjoint paths contained in $B_3$, connecting the vertices of $M$ to the vertices of $M'$. Finally, $\pset(B_2)$ is the set of $t$ rows of $B_2$, that are sub-paths of the rows in $\rset(B)$. Let $\pset(B)$ be the set of paths obtained by concatenating the paths in $\pset(B_1),\pset(B_2)$ and $\pset(B_3)$. Then $\pset(B)$ contains $t$ paths, connecting the vertices of $L$ to the vertices of $M$, such the $\ell$th path from the top contains $x$. From the construction, all paths of $\pset(B)$ are contained in $B\cap G_1'$. We add the paths in $\pset(B)$ to $\pset$.

Finally, in order to turn $\pset$ into a collection of $t$ disjoint paths with the required properties, we add, for every consecutive pair $B,B'$ of blocks, the $t$ edges of $G'$ that lie in rows $R_1,\ldots,R_t$ and connect the vertices of the last column of $B$ to the vertices of the first column of $B'$.
This completes the description of the set $\pset$ of paths. It is now immediate to see that $H$ contains a $K_t$-minor. Moreover, since each path in $\pset$ intersects every column of $J_1$, if $B^*=J_1$ is a contraction of a sub-wall of $W$, then the $K_t$-minor is grasped by $W$. Our proof also gives an efficient algorithm, that finds the model of $K_t$ in $H$. Therefore, if we are given a model of $H$ in $G$, then we can efficiently compute the model of $K_t$ in $G$ with the required properties.
\end{proof}

 \label{---------------------------------------sec: cutting the wall-------------------------------------}
\section{Cutting the Wall}\label{sec: cutting the wall}
In the statement of the Flat Wall Theorem, the starting point is a ``square'' wall - that is, a wall whose height and width are the same. However, in our algorithm, we need a wall whose width $r$ is much larger than its height $h$. One can, of course, simply use a wall of size $(r\times r)$ as a starting point, but this is very wasteful. In this section we provide an algorithm that starts with a square wall of size $(\Theta(h\sqrt r)\times \Theta(h\sqrt r))$, and turns it into a wall of height $h$ and width $r\cdot h$. We do so by first cutting the square wall into thin horizontal strips, and then connecting those strips in a snake-like fashion. This allows us to save significantly on the size of the wall in the statement of the Flat Wall Theorem.

We use two parameters $z$ and $N$, whose values will be specified later. Our starting point is a wall $W$ of size $(Nz\times Nz)$. Let $\hat W$ be the corresponding elementary wall of the same size, and let $f:V(\hat W)\rightarrow W$ be a good $(\hat W,W)$-mapping. Our goal in this section is to construct a sub-graph $W'$ of $W$, such that $W'$ is a wall of height $z$ and width $\Omega(N^2z)$. We will then partition $W'$ into $\Omega(N^2)$ disjoint walls, that we will call \emph{basic walls}, of size $z\times z$ each, where each basic wall is a sub-wall of $W'$ spanned by $z$ consecutive columns. Our construction will ensure that each such basic wall is also a sub-wall of $W$. 
In order to define the sub-graph $W'$ of $W$, we first construct a sub-graph $\hat W'$ of $\hat W$. This graph then naturally defines the sub-graph $W'$ of $W$, where $W'$ is the union of all paths $P_e$ for edges $e\in E(\hat W')$. (Recall that the paths $P_e$ are uniquely defined via the good $(\hat W,W)$-mapping $f$).

We now turn to define the sub-graph $\hat W'$ of $\hat W$.
We start by deleting, for each $1\leq i< N$, all vertical edges of $\hat W$, connecting rows $R_{iz}$, $R_{iz+1}$, and then deleting any resulting degree-1 vertices.
As a result, $\hat W$ decomposes into $N$ disjoint subgraphs $S_1,\ldots,S_N$, where for each $1\leq j\leq N$, $S_j$ is a sub-wall of $\hat W$ of height $z$ and width $Nz$. We sometimes call the walls $S_j$ \emph{horizontal strips}. 

Consider any such strip $S_j$. We partition $S_j$ into $N$ disjoint sub-walls of size $z\times z$ each, as follows. For each $1\leq i\leq N$, wall $B_{j,i}$ is the sub-wall of $S_j$ spanned by columns $(C_{(i-1)z+1},\ldots,C_{iz})$, and rows $(R_1,\ldots,R_z)$ of $S_j$. We call the walls $B_{j,i}$ \emph{basic elementary walls}. Therefore, each strip $S_j$ is now partitioned into $N$ basic elementary walls of size $(z\times z)$ each (see Figure~\ref{fig: cutting the wall1}). Let $\bset=\set{B_{j,i}\mid 1\leq i\leq N; 1\leq j\leq N}$ be the set of all basic elementary walls. For each $1\leq j\leq N$, we remove from $\bset$ the walls $B_{j,1}$ and $B_{j,N}$ - the first and the last basic elementary walls of the strip $S_j$. We will use these walls to connect the strips to each other. Therefore, $|\bset|=N(N-2)$.

\begin{figure}[h]
\scalebox{0.5}{\includegraphics{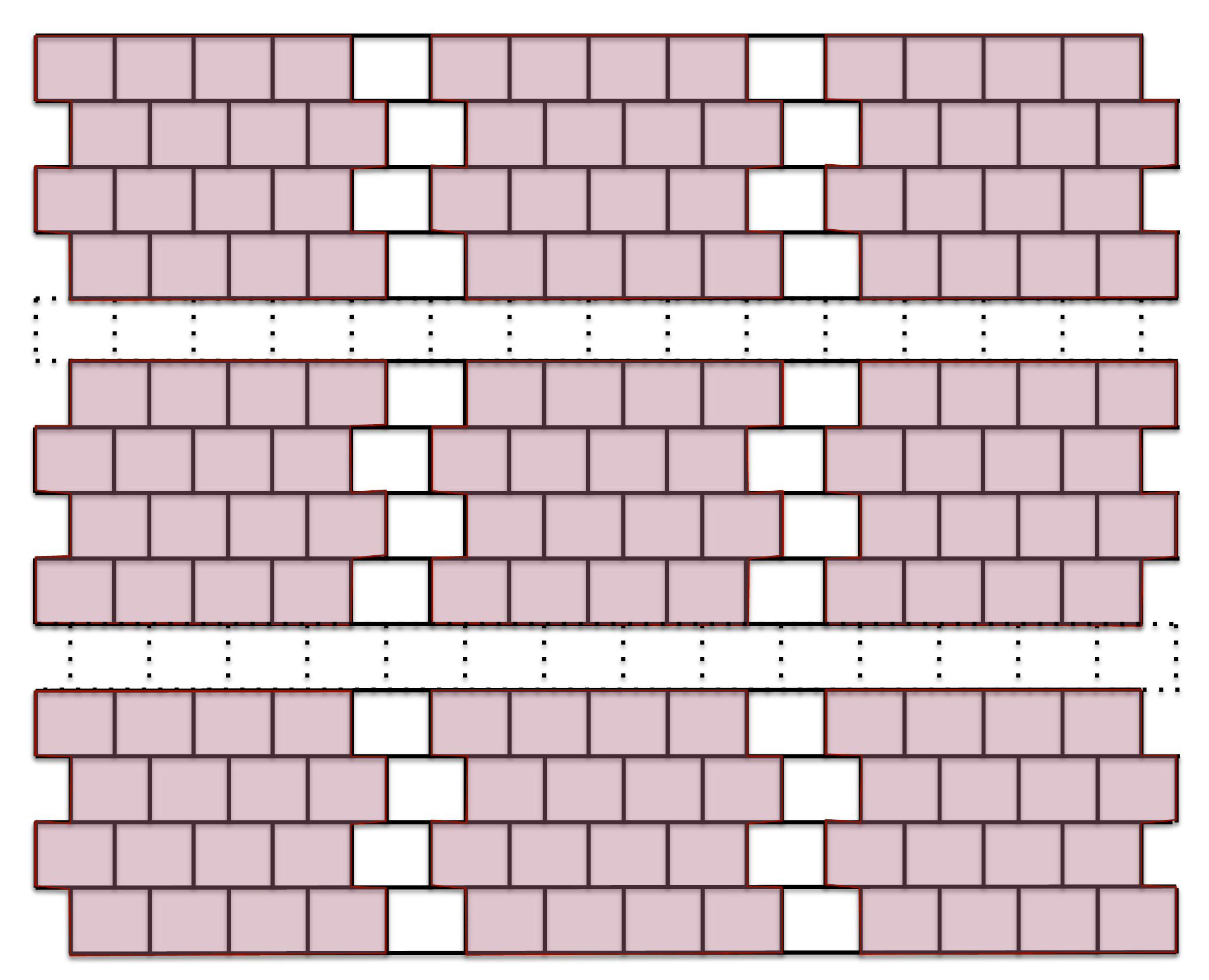}}\caption{Cutting the wall into basic elementary walls.\label{fig: cutting the wall1}}
\end{figure}

Finally, we would like to connect the strips in a snake-like fashion to create one long strip. Fix some $1\leq j<N$, where $j$ is odd.
We construct a graph $G_j$, that will be used to connect the strips $S_j$ and $S_{j+1}$ via the walls $B_{j,N}$ and $B_{j+1,N}$.
In order to construct $G_j$, we start with the sub-graph of $\hat W$ induced by $V(B_{j,N})\cup V(B_{j+1,N})$. Consider now the last column $C_z$ of $B_{j,N-1}$. We construct a subset $X\subseteq V(C_z)$ of vertices, as follows. For each row $R_i$ of $S_j$, if $|C_z\cap R_i|=1$, then we add the unique vertex of $C_z\cap R_i$ to $X$. Otherwise, $|C_z\cap R_i|=2$, and among the two vertices, we add the one that lies more to the right on $R_i$ to set $X$. We define a subset $Y$ of the vertices of the last column of $B_{j+1,N-1}$ similarly. Notice that $|X|=|Y|=z$. We add to $G_j$ the vertices of $X\cup Y$, and all horizontal edges connecting the vertices of $X\cup Y$ to the vertices of $B_{j,N}\cup B_{j+1,N}$. This completes the description of the graph $G_j$. Clearly, $G_j$ is a planar graph. Moreover, using the same reasoning as in Claim~\ref{claim: t-linkedness of wall}, it is easy to see that $X$ and $Y$ are $z$-linked in $G_j$. Therefore, there is a set $\pset_j$ of $z$ node-disjoint paths, contained in $G_j$, that connect the vertices of $X$ to the vertices of $Y$ (see Figure~\ref{fig: cutting the wall2}). Since $G_j$ is a planar graph, if $v\in X$ lies in row $R_i$ of $S_j$, then $v$ is connected to the vertex $u\in Y$ lying in row $z-i+1$ of $S_{j+1}$.

Similarly, if $1<j<N$ is an even number, then we construct a graph $G_j$ that will be used to connect strips $S_j$ and $S_{j+1}$ via the walls $B_{j,1}$ and $B_{j+1,1}$. In order to construct $G_j$, we start with the sub-graph of $\hat W$ induced by $V(B_{j,1})\cup V(B_{j+1,1})$. We also define the sets $X$ and $Y$ of $z$ vertices each as before: $X$ is a subset of vertices of the first column of $B_{j,2}$, containing exactly one vertex lying in each row of $S_j$. For each row $R_i$, if $|R_i\cap C_1|=2$, we add the vertex of $R_i\cap C_1$ that lies more to the left on $R_i$, to $X$. Set $Y$ is defined similarly for the first column of $B_{j+1,2}$. We add to $G_j$ the vertices $X\cup Y$, and all horizontal edges of $S_j$, connecting these vertices to the vertices of $B_{j,1}\cup B_{j+1,1}$. As before, $X$ and $Y$ are $z$-linked in $G_j$, and therefore there is a set $\pset_j$ of $z$ node-disjoint paths connecting them in $G_j$.

\begin{figure}[h]
\scalebox{0.6}{\includegraphics{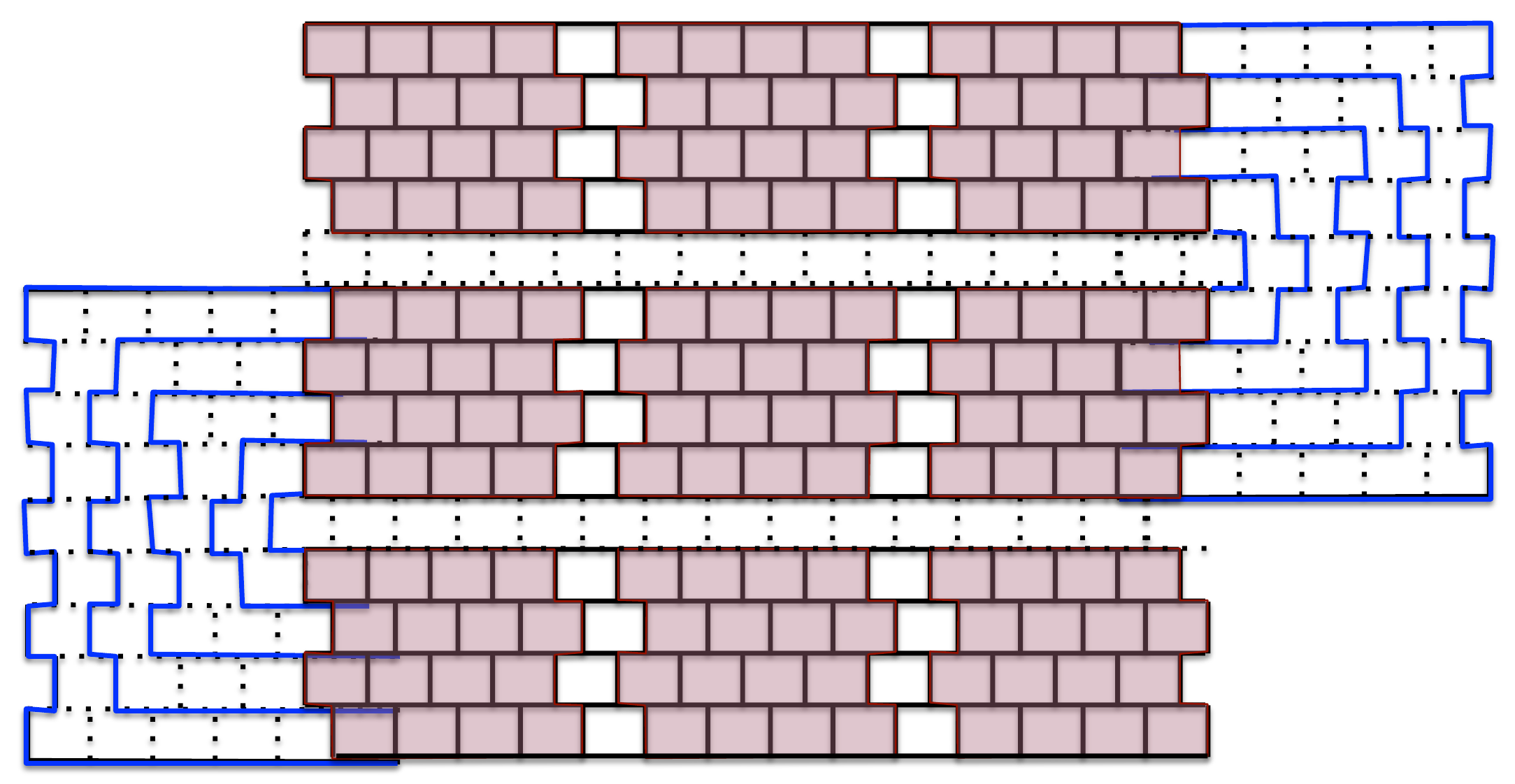}}\caption{Connecting the strips into one long strip. \label{fig: cutting the wall2}}
\end{figure}

The final graph $\hat W'$ is constructed by taking the union of all basic elementary walls in $\bset$. For each consecutive pair $B_{j,i},B_{j,i+1}$ of such walls (where $1\leq j\leq N$, $2\leq i\leq N-1$), we add all horizontal edges of $\hat W$ connecting $V(B_{j,i})$ and $V(B_{j,i+1})$. Finally, for all $1\leq j<N$, we add all vertices and edges of paths in $\pset_j$ to $\hat W'$. This completes the definition of the graph $\hat W'$.

Recall that each edge $e$ of $\hat W$ corresponds to some path $P_e$ of the original wall $W$ via the good $(\hat W,W)$-mapping $f$, where the paths in $\set{P_e}_{e\in E(\hat W)}$ are pairwise internally node-disjoint. The subgraph $W'$ of $W$ is obtained by replacing every edge $e\in E(\hat W')$ by the corresponding path $P_e$. The $z\times z$ walls in $W'$, that correspond to the basic elementary walls in $\bset$ are called \emph{basic walls}. Therefore, $W'$ is a wall of height $z$ and width $zN(N-2)$. 
We call the corresponding object \emph{a chain of $N(N-2)$ basic walls of height $z$}, and we formally define it below.

\begin{definition}
A \emph{chain $(\bset,\pset)$ of $N$ basic walls of height $z$}  consists of:

\begin{itemize}
\item A collection $\bset$ of $N$ disjoint walls $B_1,\ldots,B_N$, that we call \emph{basic walls}, where each wall $B_i$ has height $z$ and width at least $z$.

\item A set $\pset=\bigcup_{j=1}^{N-1}\pset_j$ of disjoint paths, where for each $1\leq j<N$, $\pset_j=\set{P_1^j,\ldots,P^j_z}$ is a set of $z$  paths, connecting the pegs of $B_j$ lying in the last column of $B_j$ to the pegs of $B_{j+1}$ lying in the first column of $B_{j+1}$, and for $1\leq i\leq z$, $P_i^j$ connects a vertex in the $i$th row of $B_j$ to a vertex in the $i$th row of $B_{j+1}$. Moreover, the paths in $\pset$ do not contain the vertices of $\bigcup_{j'=1}^{N}V(B_{j'})$ as inner vertices. 
\end{itemize}
We denote by $W'(\bset,\pset)$ the corresponding graph $\left(\bigcup_{j=1}^N\bset_j\right )\cup\left(\bigcup_{j=1}^{N-1}\pset_j\right )$.
\end{definition}

Observe that graph $W'=W'(\bset,\pset)$ is a wall of height $z$ and width at least $Nz$. We will always assume that we are given some fixed choice of the four corners $(a,b,c,d)$ of the wall $W'$, that appear along the boundary of $W'$ in this order clock-wise, and $a$ is the top left corner of $W'$. Therefore, for each basic wall $B_i$, the four corners of $B_i$ are also fixed, and are denoted by $a_i,b_i,c_i,d_i$, where $a_i$ is the top left corner, and the four corners appear in this order clock-wise along the boundary of $B_i$. From the discussion in this section, we obtain the following theorem.

\begin{theorem}\label{thm: getting the chain}
For any integers $N,z\geq 2$, given a wall $W$ of size $(Nz\times Nz)$, there is an efficient algorithm to construct a chain $(\bset,\pset)$ of $N(N-2)$ basic walls of height $z$, such that $W'(\bset,\pset)$ is a sub-graph of $W$, and each basic wall $B\in \bset$ is a sub-wall of $W$.
\end{theorem}

\label{---------------------------------------bridges, core walls, wall types--------------------------}
\section{Bridges, Core Walls, and Wall types}\label{sec: bridges, core walls, wall types}

Let $G$ be any graph, $(\bset,\pset)$ a chain of $N$ basic walls of height $z$ in $G$, and let $W'=W'(\bset,\pset)$ be the corresponding sub-graph of $G$.
Let $\cset$ be the set of all connected components of $G\setminus V(W')$. We say that a component $F\in \cset$ \emph{touches} a vertex $v\in V(W')$ iff $G$ contains an edge from a vertex of $F$ to $v$.

Given an integer parameter $1\leq \tau<z/2$, for each $1\leq i\leq N$, we define a \emph{$\tau$-core sub-wall $B'_i$} of $B_i$, as follows. Wall $B'_i$ is the sub-wall of $B_i$ spanned by rows $(R_{\tau},\ldots,R_{z-\tau+1})$ and all columns of $B_i$ (see Figure~\ref{fig: core}). 
The boundary of the $\tau$-core wall $B'_i$ is denoted by $\Gamma_i'$, and its four corners are denoted by $a_i',b_i',c_i',d_i'$. We assume that the four corners appear in this order clock-wise along the boundary of $B'_i$, and $a_i'$ is the top left corner. We will assume throughout this section that the value of $\tau$ is fixed, and we will sometimes refer to the $\tau$-core walls simply as core walls.

\begin{figure}[h]
\scalebox{0.4}{\includegraphics{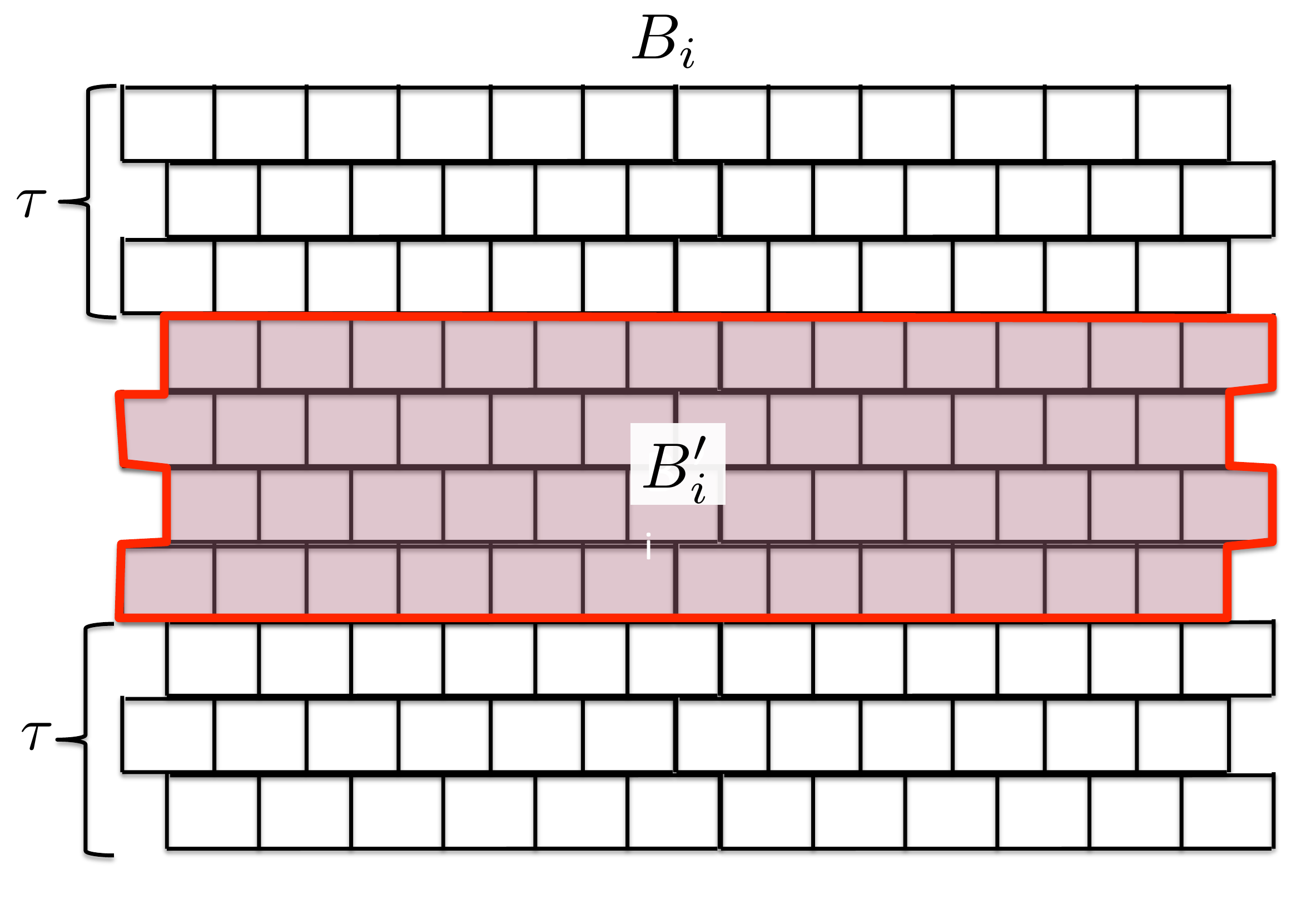}}\caption{Graphs $B_i$ and $B_i'$, with $\Gamma_i'$ shown in red.\label{fig: core}}
\end{figure}

\begin{definition}
A \emph{bridge} incident on a core wall $B'_i$ is one of the following: either an edge with one endpoint in $V(B'_i\setminus \Gamma_i')$, and another in $V(W'\setminus B'_i)$; or a connected component $F\in \cset$, that touches a vertex of $V(B'_i\setminus \Gamma_i')$, and a vertex of $V(W'\setminus B'_i)$.
\end{definition}

For $1<i<N$, we define the neighborhood of $B_i$ as follows: $\nset(B_i)=V(B_{i-1}\cup \pset_{i-1}\cup B_i\cup \pset_i\cup B_{i+1})$. We say that a bridge $F$ incident on the core wall $B'_i$ is a \emph{neighborhood bridge} for $B'_i$ iff either $F$ is an edge whose both endpoints lie in $\nset(B_i)$, or $F$ is a connected component of $\cset$, such that all vertices of $W'$ that $F$ touches belong to $\nset(B_i)$. Otherwise, we say that $F$ is a \emph{non-neighborhood} bridge for $B_i'$.

For $1<i<N$, we say that the core wall $B'_i$ is a type-1 wall iff there is at least one bridge $F$ incident on $B'_i$, such that $F$ is a neighborhood bridge for $B'_i$. We say that it is a type-2 wall if it is not a type-1 wall, and at least one bridge is incident on $B'_i$. Therefore, if $B'_i$ is a type-2 wall, then at least one bridge $F$ incident on $B'_i$ is a non-neighborhood bridge for $B'_i$.

Assume now that $B_i'$ is not type-1 and not type-2 wall. Then no bridge is incident on $B'_i$, so graph $G\setminus \Gamma_i'$ consists of at least two connected components, with one of them containing $B'_i\setminus \Gamma'_i$. Therefore, there is a separation $(X,Y)$ of $G$, with $B'_i\subseteq X$, $X\cap Y\subseteq \Gamma_i'$, and for each $1\leq j\leq N$ with $j\neq i$, $B_j\subseteq Y$. Recall that the corners $a_i',b'_i,c_i',d_i'$ of the wall $B'_i$ are fixed. We assume that they appear on $\Gamma_i'$ in this order clockwise, with $a_i'$ being the top left corner. If graph $X$ contains a wall-cross for $B'_i$ (that is, a pair of disjoint paths connecting $a_i'$ to $c_i'$ and $b_i'$ to $d_i'$), then we say that wall $B_i'$ is of type 3. Otherwise, it is of type 4. 
Notice that given $W'$, for each $1<i<N$, we can efficiently determine what type wall $B_i'$ belongs to, and if it is a type-3 wall, then we can find the corresponding wall-cross efficiently.

The proof of both Theorems~\ref{thm: main weak} and \ref{thm: main strong} proceeds in a similar way: we start with a wall $W$ of an appropriate size, and apply Theorem~\ref{thm: getting the chain} to obtain a chain of walls $(\bset,\pset)$ with some parameters $z$ and $N$. We then show that if, for any of the basic walls in the chain, the corresponding $\tau$-core wall is of type $4$, then $G$ contains a flat sub-wall of $W$ of size $((z-2\tau)\times (z-2\tau))$. Therefore, with an appropriate choice of $z$ and $\tau$, we can assume that all $\tau$-core walls are of types $1$, $2$ or $3$. For each one of these three types, we show that if there are many $\tau$-core walls of that type, then $G$ contains a $K_t$-minor grasped by $W$. Therefore, large parts of the proofs of both theorems are similar, and only differ in the specific parameters we choose. In the rest of this section we formally state and prove theorems that allow us to handle walls of each one of the four types. The statements of the theorems are generic enough so we can apply them in the different settings with the different choices of the parameters that we need. 

We start by observing that if at least one basic wall of $\bset$ is of type $4$, then $G$ contains a flat wall of size $((z-2\tau)\times (z-2\tau))$.

\begin{figure}[h]
\scalebox{0.4}{\includegraphics{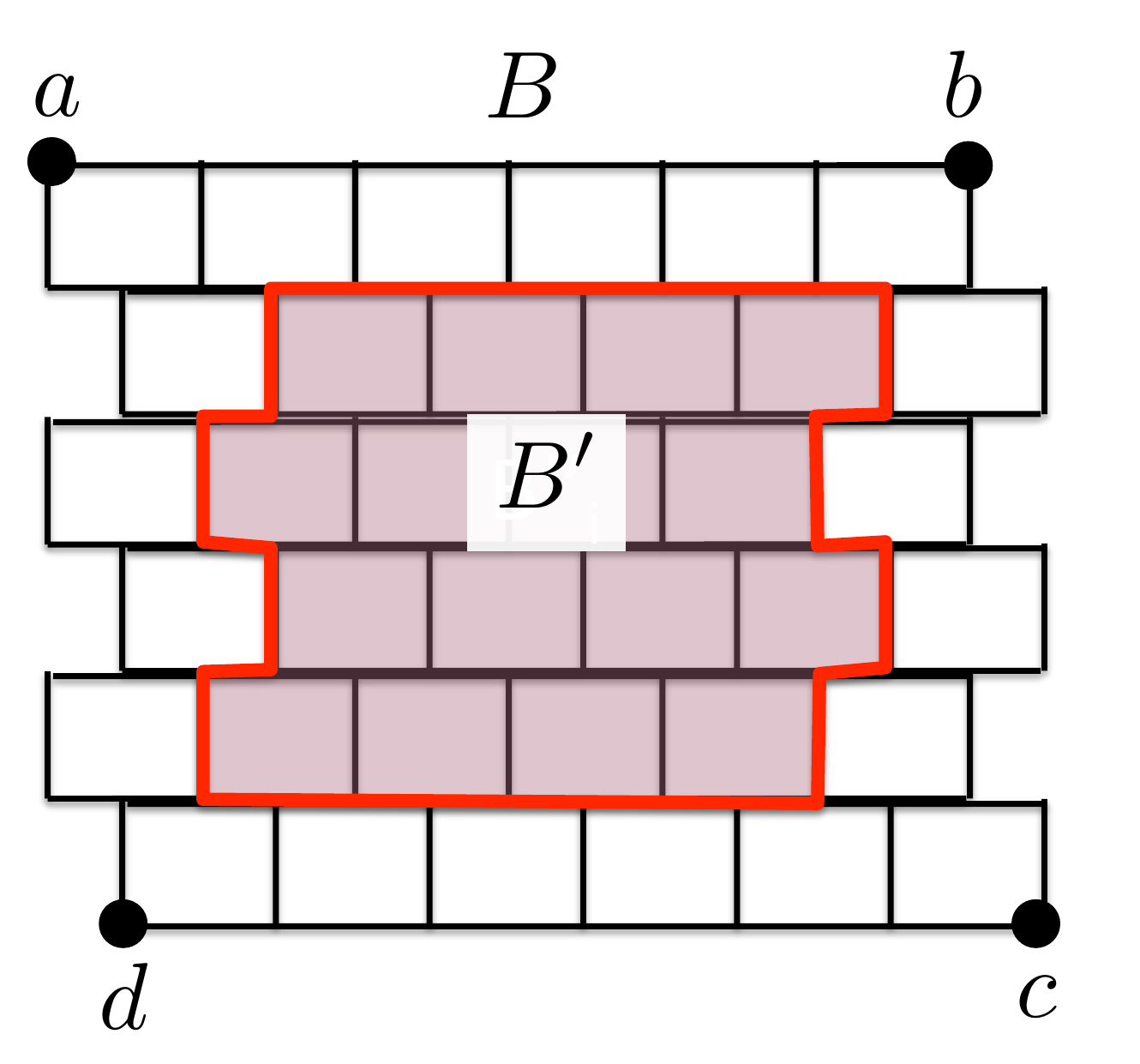}}\caption{Graphs $B$ and $B'$, with $\Gamma'$ shown in red.\label{fig: Bi'}}
\end{figure}

\begin{lemma}\label{lemma: type 4 gives flat wall}
Let $G$ be any connected graph, $B$ a wall of size $((z+2)\times (z+2))$ in $G$, $\Gamma$ an outer boundary of $B$, and $a,b,c,d$ the corners of $B$ appearing on $\Gamma$ in this order circularly.  Assume further that there is a separation $(X,Y)$ of $G$, with $B\subseteq X$, $X\cap Y\subseteq \Gamma$, such that $X$ does not contain a wall-cross for $B$. Then there is an efficient algorithm to find a flat wall $B'$ of size $(z\times z)$ in $G$, such that $B'$ is a sub-wall of $B$. 
\end{lemma}

We then obtain the following immediate corollary:
\begin{corollary}\label{cor: type 4 gives flat wall}
Let $G$ be any connected graph, $W$ a wall in $G$, $(\bset,\pset)$ a chain of $N$ walls of height at least $z$ in $G$, and $\tau<z/2$ some integer. Assume that for some $1<i<N$, the $\tau$-core wall $B_i'$ is a type-4 wall. Then there is an efficient algorithm to find a flat wall of size $((z-2\tau)\times (z-2\tau))$ in $G$. Moreover, if $B_i$ is a sub-wall of $W$, then so is the flat wall.
\end{corollary}

\begin{proofof}{Lemma~\ref{lemma: type 4 gives flat wall}}
Let $B'$ be the sub-wall of $B$, obtained from $B\setminus \Gamma$, by deleting all degree-$1$ vertices (see Figure~\ref{fig: Bi'}). The size of $B'$ is $(z\times z)$. We will show that $B'$ is a flat wall in $G$.
We denote the boundary of $B'$ by $\Gamma'$. Recall that we are given a separation $(X,Y)$ of $G$, with $X\cap Y\subseteq \Gamma$, and $B\subseteq X$. 
We assume that the four corners of $B$ appear in clock-wise order $(a,b,c,d)$ along $\Gamma$, with $a$ being the top left corner. We denote the four corners of $B'$ by $a',b',c',d'$ similarly. Recall that graph $X$ contains no wall-cross for $B$, that is, we cannot connect $a$ to $c$ and $b$ to $d$ by disjoint paths in $X$.

We augment the graph $X$ by adding four edges: $(a,b),(b,c),(c,d)$, and $(d,a)$, and we denote by $L$ the cycle consisting of these four edges. Let $\tilde X$ denote the resulting graph.
Since there is no pair of disjoint paths in $X$, connecting $a$ to $c$ and $b$ to $d$, there is no $L$-cross in $\tilde X$, and so, from Theorem~\ref{thm: a cross or C-flat}, graph $\tilde X$ is $L$-flat. Let $G_0,G_1,\ldots,G_k,\tilde G$ be the corresponding collection of graphs, certifying that $\tilde X$ is $L$-flat (recall that we can find these graphs efficiently). We exploit this collection of graphs below in order to prove that $B'$ is a flat wall. 
We need the following simple observation.

\begin{observation}\label{obs: no direct path}
Let $P$ be any path in $G$, connecting a vertex $v\in \Gamma$ to a vertex $u\in B'\setminus \Gamma'$. Then $P$ must contain a vertex of $\Gamma'$.
\end{observation}

\begin{proof}
Assume otherwise. Since $\Gamma$ separates $X\setminus \Gamma$ from $Y\setminus \Gamma$, there is a sub-path $P'$ of $P$ that starts at a vertex of $\Gamma$, ends at a vertex of $B'\setminus \Gamma'$, and only uses the vertices of $X$. Moreover, there is a sub-path $P''$ of $P'$, that starts at a vertex $v'$ of $B\setminus B'$, ends at a vertex $u'$ of $B'\setminus \Gamma'$, and does not use any vertices of $B$ as inner vertices (where possibly $P''$ consists of only one edge). From Theorem~\ref{thm: cross-wall}, we conclude that $X$ contains a wall-cross for $B$, a contradiction.
 \end{proof}

Let $\rset$ be the set of all connected components of $G\setminus\Gamma'$, and let $R\in \rset$ be the component containing $B'\setminus \Gamma'$. Then, from Observation~\ref{obs: no direct path}, $R$ does not contain any vertices of $B\setminus B'$. Let $R'\in \rset$ be the connected component containing the vertices of $B\setminus B'$. Then $Y\subseteq R'$ must hold, since $\Gamma'$ cannot separate $\Gamma$ from $Y$, as $G$ is connected, $\Gamma$ separates $X\setminus \Gamma$ from $Y\setminus\Gamma$, while $\Gamma'\subseteq X\setminus \Gamma$.

We now turn to define the pegs of the wall $B'$. Consider the vertices that lie in the first row of $B'$. We can partition these vertices into three types.

\begin{itemize}

\item Vertices whose degree in $B'$ is $3$;

\item Vertices of $B'$ whose degree in both $B'$ and $B$ is $2$;

\item Vertices of $B'$ whose degree in $B'$ is $2$ and in $B$ is $3$.
\end{itemize}

Let $Z_1$ be the set of all vertices of the latter type. Define $Z_2,Z_3,Z_4$ similarly for the last row, first column, and last column of $B'$, respectively, and set $Z=\bigcup_{i=1}^4Z_i$. Notice that the corners of $B'$ belong to $Z$.
We choose the vertices of $Z$ to be the pegs of the wall $B'$ (see Figure~\ref{fig: pegs}). Notice that for each peg $v\in Z$, there is a path $P_v$ connecting $v$ to a vertex of $\Gamma$, such that all vertices of $P_v$, except for $v$, belong to $B\setminus B'\subseteq R'$.

\begin{figure}[h]
\scalebox{0.6}{\includegraphics{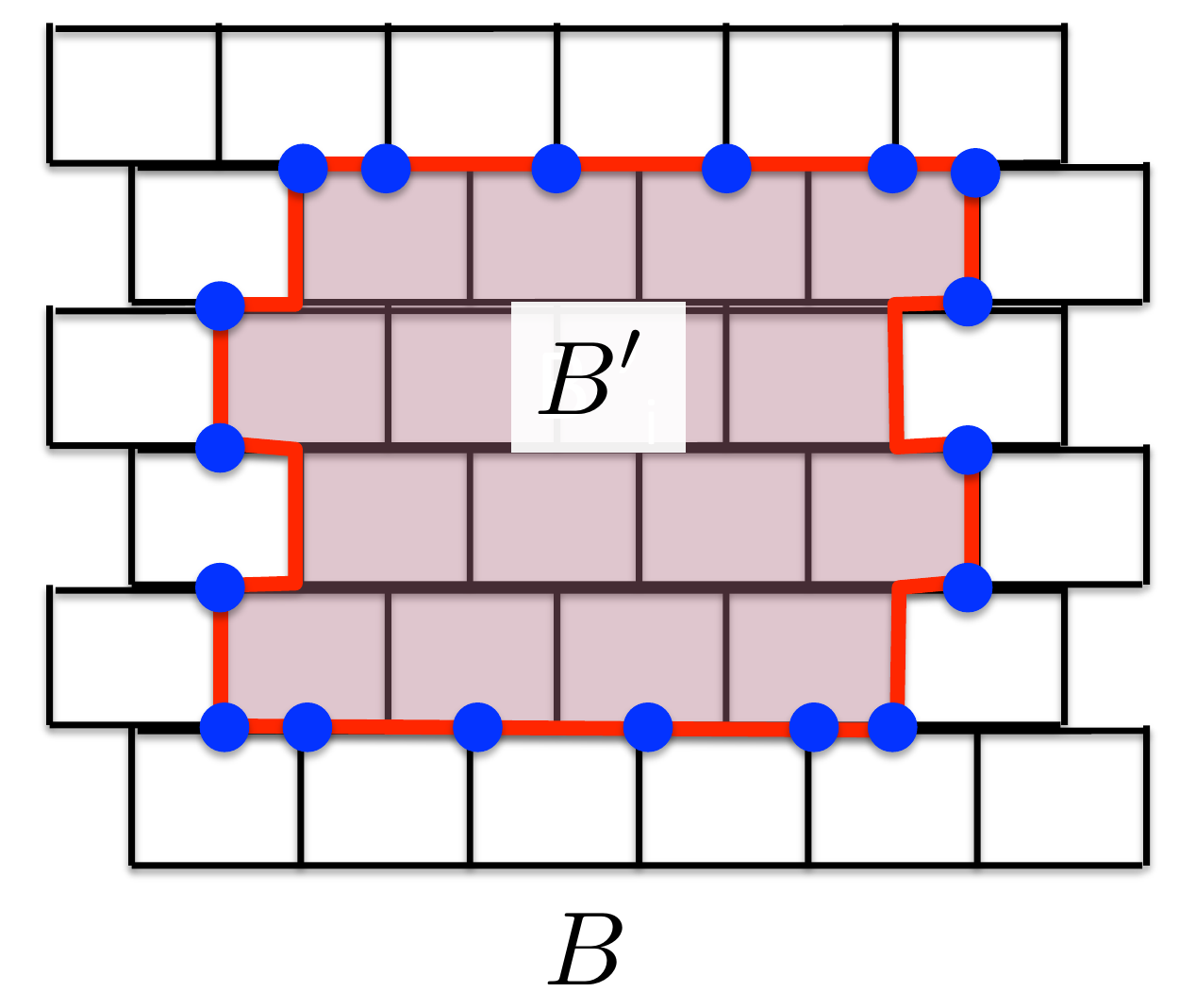}}\caption{The vertices of $Z$ are shown in blue.\label{fig: pegs}}
\end{figure}

In order to prove that $B'$ is a flat wall, we gradually construct a separation $(\aset,\bset)$ of $G$, with $B'\subseteq \bset$, $\aset\cap \bset\subseteq \Gamma'$, $Z\subseteq V(\aset)$, together with an $\aset\cap \bset$-reduction of $G[\bset]$, that is drawn in a disc with the vertices of $\aset\cap \bset$ drawn on its boundary in the order determined by $\Gamma'$.

We will use the fact that $\tilde X$ is $L$-flat, and exploit the corresponding graphs $G_0,G_1,\ldots,G_k,\tilde G$ in order to find the separation $(\aset,\bset)$ of $G$, and the $\aset\cap \bset$-reduction of $G[\bset]$. Consider the graphs $G_j$, for $1\leq j\leq k$. For each such graph $G_j$, the vertices of $V(G_j)\setminus V(G_0)$ are called the \emph{inner vertices} of $G_j$, and the vertices of $V(G_j)\cap V(G_0)$ are called the \emph{endpoints} of $G_j$. We partition the graphs $G_j$ into four types. For $1\leq j\leq k$:

\begin{itemize}
\item Graph $G_j$ is of type 1 iff none of its inner vertices belongs to $\Gamma'$;

\item Graph $G_j$ is of type $2$, iff it is not of type 1, and it has exactly two endpoints;

Notice that if $G_j$ does not belong to types $1$ and $2$, then it must have an inner vertex in $\Gamma'$, and three endpoints. Since $\Gamma'$ is a cycle, two of the endpoints of $G_j$ must belong to $\Gamma'$. 

\item We say that graph $G_j$ is of type $3$, iff it is not of types 1 and 2, and its third endpoint belongs to $R'$.
\item Otherwise, graph $G_j$ is of type $4$ (see Figure~\ref{fig: graph-types-sub}).
\end{itemize}

\begin{figure}[h]
\centering
\subfigure[The four graph types: $G_1$ is of type 1, $G_2$ is of type 2, $G_3$ is of type 3, and $G_4$ is of type 4.]{\scalebox{0.35}{\includegraphics{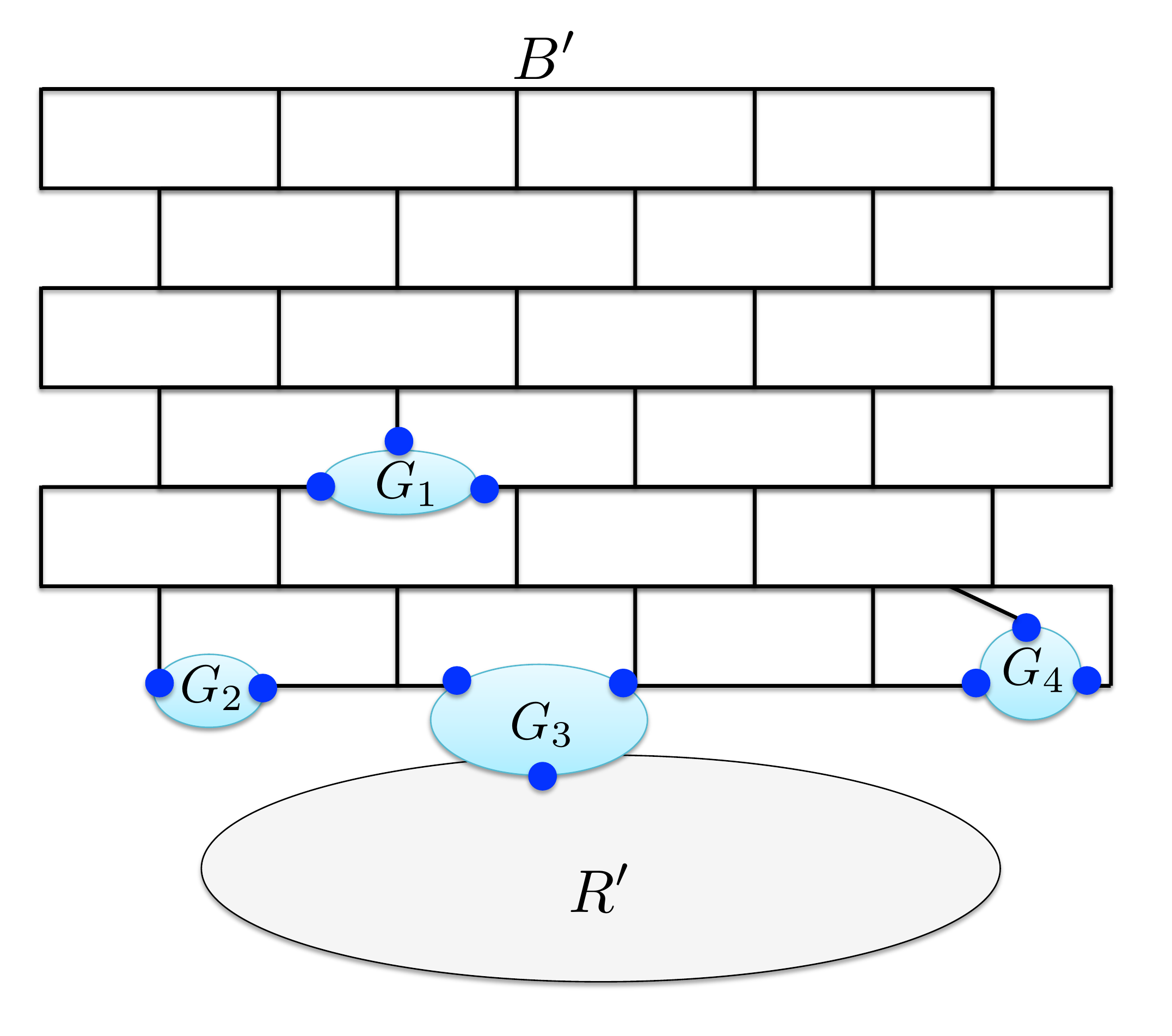}}\label{fig: graph-types-sub}}
\hspace{1cm}
\subfigure[Handling a type-3 graph]{
\scalebox{0.35}{\includegraphics{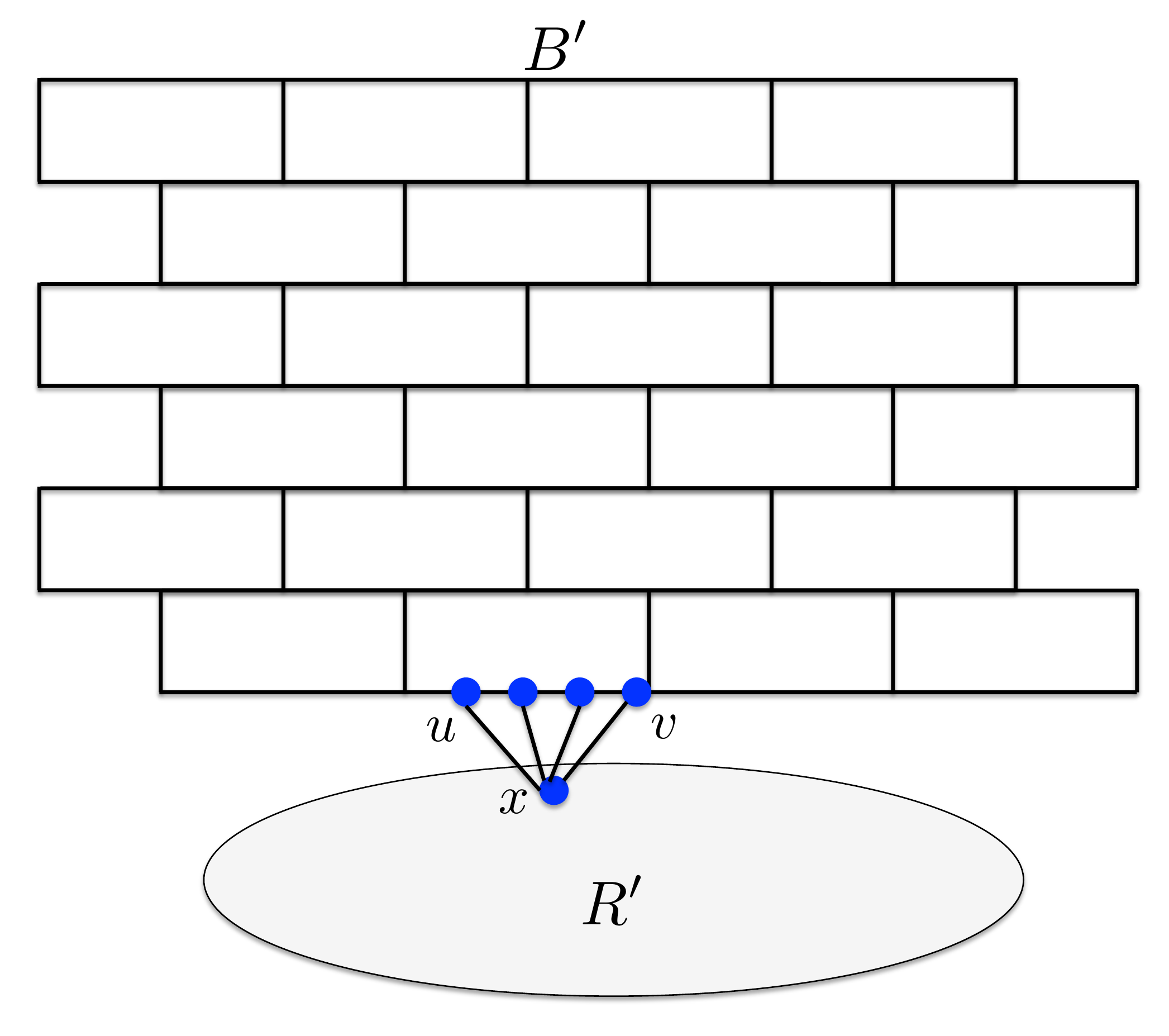}}\label{fig: type-3-graph}}
\caption{The four types of graphs \label{fig: graph-types}}
\end{figure}

Observe that if $G_j$ is of type 2 or 4, then it does not contain the vertices of $Z$ as its inner vertices, since for each $v\in Z$, there is a path $P_v$, connecting $v$ to $\Gamma$, and $P_v\setminus \set{v}\subseteq R'$, as observed above.
We need the following simple observation.

\begin{observation}\label{obs: type-3 gives a path}
For each $1\leq j\leq k$, if $G_j$ is a graph of type $3$, and $D=G_j\cap B'$, then $D$ is a sub-path of $\Gamma'$.
\end{observation}
The observation follows since exactly two endpoints of $G_j$ belong to $B'$, and they must lie on $\Gamma'$. Any such 2-vertex cut of $B'$ can only disconnect a sub-path of $\Gamma'$ from the rest of the wall $B'$.

We are now ready to define the separation $(\aset,\bset)$ of $G$. First, we define a set $Z'=V(\aset)\cap V(\bset)$, as follows. Set $Z'$ contains all the vertices of $v$ of $\Gamma'$, such that there is an edge $(u,v)\in E(G)$ with $u\in R'$ (notice that this includes all vertices of $Z$). Additionally, for each type-3 graph $G_j$, we add all the vertices of $G_j\cap B'$ to $Z'$. From Observation~\ref{obs: type-3 gives a path}, $Z'\subseteq V(\Gamma')$.

Let $\rset'$ be the set of all connected components of $G\setminus Z'$. Then since $Z'$ separates $B'\setminus \Gamma'$ from $R'$, there is a connected component $\tilde R\in \rset'$ with $R\subseteq \tilde R$, and a connected component $\tilde R'\in \rset'$ with $R'\subseteq \tilde R'$.
The set $V(\bset)$ of vertices consists of all the vertices of $Z'$ and $V(\tilde R)$. We let $E(\bset)$ be all the edges of $G$ with both endpoints in $V(\bset)$. Clearly, $B'\subseteq \bset$. We set $V(\aset)=V(G)\setminus V(\tilde R)$. Set $E(\aset)$ contains all edges with both endpoints in $V(\aset)$, except for the edges with both endpoints in $\Gamma'$.

It is easy to see that $\aset\cap \bset=Z'\subseteq V(\Gamma')$, and $Z\subseteq \aset\cap \bset$. It now only remains to show that $G[\bset]$ is $\aset\cap \bset$-flat. 
In order to do so, we exploit the planar drawing of $\tilde G$, and turn it into a drawing of an $\aset\cap \bset$-reduction of $G[\bset]$.

Consider again the graphs $G_1,\ldots,G_k$. The only graphs $G_j$ that may contain the vertices of $\aset\cap \bset$ as inner vertices are graphs  of type $3$. Let $G_j$ be any such graph. Recall that $G_j\cap B'$ is a sub-path of $\Gamma'$. Denote this path by $P_j$, and its endpoints by $u$ and $v$. We claim that $G_j\cap \bset=P_j$. Indeed, if we let $S=V(G_j\setminus P_j)$, then all vertices of $S$ are disconnected from the vertices of $B'$ in $G\setminus Z'$, and so they do not belong to $\tilde R$. Therefore, $G_j\cap \bset=P_j$.

 Let $x$ denote the endpoint of $G_j$ that does not lie on $\Gamma'$, so $x\in \tilde R'$. Recall that the drawing of $\tilde G$ contains a face whose only endpoints are $u,v$ and $x$. Consider the curve on the boundary of that face, connecting $u$ to $v$, that does not contain $x$. We subdivide this curve with the vertices of $P_j$, so it now corresponds to a drawing of the path $P_j$. For each vertex $w$ of $P_j$, we add an edge $(x,w)$ to the graph, and add the drawing of this edge to our current planar drawing (see Figure~\ref{fig: type-3-graph}). Once we process all graphs $G_j$ of type 3 in this manner, we obtain a planar drawing of the resulting graph $\tilde G'$. Every vertex of $\aset\cap \bset$ now belongs to $\tilde G'$. Moreover, each such vertex $v\in V(\aset\cap \bset)$ has the following property: there is an edge in $\tilde G'$ connecting $v$ to some vertex of $\tilde R'$. Since $\tilde R'\cap \tilde{G}'$ is a connected graph, we can unify all vertices of $\tilde R'\cap \tilde G'$ into a single vertex $v_0$, and obtain a planar drawing of the resulting graph, where each vertex of $\aset\cap \bset$ is connected to $v_0$ with an edge. We can now immediately obtain a drawing of $\bset$ inside a disc, with the vertices of $\aset\cap \bset$ appearing on the boundary of the disc, in the order specified by $\Gamma'$. 
\end{proofof}


\subsection{Type-3 Core Walls}


The goal of this section is to prove the following theorem.

\begin{theorem}\label{thm: type-3 walls}
Let $G$ be a connected graph, $\tau>t$ an integer, $W$ a wall in $G$, and $(\bset,\pset)$ a chain of $N$ walls of height $z>2\tau$ in $G$. If the number of $\tau$-core walls $B'_i$ of type 3 with $1<i<N$ in the chain of walls is at least $2T$, then $G$ contains a $K_t$-minor. Moreover, if $B_1\in \bset$ is a sub-wall of $W$, then $G$ contains a model of a $K_t$-minor that is grasped by $W$, and it can be found efficiently given $(\bset,\pset)$.
\end{theorem}

\begin{proof}
We will show that $G$ contains the graph $H^*$, defined in Section~\ref{sec: two graphs}, as a minor. We will also ensure that the first odd block $J_1$ of $H^*$ is a contraction of a sub-wall of $B_1$. The theorem will then follow from Lemma~\ref{lemma: from H1 to clique} and Observation~\ref{obs: transitivity of minors}. 

Let $\sset$ be the set of all $\tau$-core walls $B'_i$ of type 3, with $1<i<N$. We select a subset $\sset'\subseteq \sset$ of $T$ core walls, such that every consecutive pair of such walls is separated by at least one core wall. In other words, if we denote $\sset'=\set{B'_{i_1},B'_{i_2},\ldots,B'_{i_T}}$, where $1<i_1<i_2<\ldots<i_T<N$, then for all $1\leq r<T$, $i_{r+1}\geq i_r+2$. Subset $\sset'$ can be found using a standard greedy procedure, by ordering the walls in $\sset$ in the ascending order of their indices, and then choosing all walls whose location in this ordering is even. We then delete walls from $\sset'$ as necessary to ensure that $|\sset'|=T$.

We will use each wall $B'_{i_r}\in \sset'$ to ``implement'' the even block $M_r$ of $H_1$, and the wall $B_{i_{r}-1}$ (whose corresponding core wall cannot belong to $\sset'$ due to the choice of $\sset'$) to implement $J_r$. The only exception is when $r=1$, where we will use $B_1$ to implement $J_1$.

Fix some $1\leq r\leq T$. Recall that $a_{i_r}',b_{i_r}',c_{i_r}', d_{i_r}'$ are the corners of the wall $B'_{i_r}$, and they appear on its boundary $\Gamma'_{i_r}$ in this order. Moreover, since there is no bridge incident to $B'_{i_r}$, there is a separation $(X_{i_r},Y_{i_r})$ of $G$, with $B'_{i_r}\subseteq X_{i_r}$, $X_{i_r}\cap Y_{i_r}\subseteq \Gamma_{i_r}'$, such that all graphs $\set{X_{i_r}}_{j=1}^T$ are disjoint from each other. Since $B_{i_r}'$ is a type-3 core wall, we can efficiently find a pair $Q^1_{i_r},Q^2_{i_r}$ of disjoint paths in $X_{i_r}$ connecting $a_{i_r}'$ to $c_{i_r}'$ and $b_{i_r}'$ to $d_{i_r}'$, respectively. 


Fix some $1\leq r\leq T$. We define a set $\cset_r$ of $t$ columns of $W'$, as follows.
Let $r_1$ be the index of the last column of $B_{i_{r-1}}$ (for $r=1$, set $r_1=0$), and let $r_2$ be the index of the first column of $B_{i_r}$. Let $\cset_r'$ be the consecutive set of the columns of $W'$, starting with $C_{r_1+1}$, and ending with $C_{r_2-1}$. Observe that $|\cset_r'|>t$, since all columns of the basic wall $B_{i_r-1}$ belong to $\cset_r'$ (for $r=1$, this holds for $B_1$). 
Finally, let $\cset_r\subseteq \cset_r'$ be the set of the first $t$ columns of $\cset_r'$.
Notice that for all $1<r\leq T$, all columns of $\cset_r$ lie between the basic walls $B_{i_r-1}$ and $B_{i_r}$, while the columns of $\cset_1$ lie before $B_{i_1}$ and are contained in $B_1$.

Let $W''$ be the sub-graph of $W'$, obtained by taking the union of all columns in $\bigcup_{r=1}^T\cset_r$, the top $t$ rows and the bottom $t$ rows of $W'$. For each $1\leq r\leq T$, consider the $t$-core sub-wall $B_{i_r}''$ of $B_i$. We extend the paths $Q^1_{i_r}$ and $Q^2_{i_r}$ along the first and the last columns of $B_{i_r}''$, so that these paths become a wall-cross for $B_{i_r}''$, and they are internally disjoint from the top and the bottom row of $B_{i_r}''$.

Let $W^*$ be the union of $W''$ with $\set{Q^1_{i_r},Q^2_{i_r}\mid 1\leq i_r\leq T}$. 
In our final step, for each column $C_q$ each row $R_{\ell}$ of $W''$, we contract all edges in $C_q\cap R_{\ell}$. 
It is easy to see that this graph is a sub-division of $H^*$, and it is a minor of $G$. Moreover,  the first odd block $ J_1$ of the resulting copy of $H^*$ is a contraction of a sub-wall of $B_1$. Therefore, from Lemma~\ref{lemma: from H1 to clique}, $G$ contains a $K_t$-minor. Moreover, if $B_1$ is a sub-wall of $W$, then there is a model of $K_t$ grasped by $W$ in $G$, and this model can be found efficiently.
\end{proof}

\subsection{Type-1 Core Walls}


In this section we take care of type-1 core walls, by proving the following theorem.

\begin{theorem}\label{thm: type-1 walls}
Let $G$ be a connected graph, $\tau\geq 2t$ an integer, $W$ a wall in $G$, and $(\bset,\pset)$ a chain of $N$ walls of height $z>2\tau$ in $G$. If the number of $\tau$-core walls $B'_i\in \bset$  of type 1 with $1<i<N$ is at least $12T+6$, then $G$ contains a $K_t$-minor. Moreover, if $B_1\in \bset$ is a sub-wall of $W$, then $G$ contains a model of a $K_t$-minor grasped by $W$, and it can be found efficiently, given $(\bset,\pset)$.\end{theorem}

\begin{proof}
Let $\sset$ be the set of $\tau$-core walls $B'_i$ of type 1, with $1<i<N$.  We select a subset $\sset'\subseteq \sset$ of $4T+2$ core walls, such that, if we denote $\sset'=\set{B'_{i_1},B'_{i_2},\ldots,B'_{4T+2}}$, where $1<i_1<i_2<\ldots<i_{4T+2}<N$, then $i_1>2$, and for all $1\leq r<4T+2$, $i_{r+1}\geq i_r+3$. In other words, $B'_1,B'_2,B'_N\not \in \sset'$, and every consecutive pair of core walls is separated by at least two walls. In order to construct $\sset'$, we can use a greedy algorithm similar to the one in the proof of Theorem~\ref{thm: type-3 walls}: order the walls in $\sset$ in the ascending order of their indices; add to $\sset'$ all walls whose indices are $0$ modulo $3$ even in this ordering; delete walls from $\sset'$ as necessary until $|\sset'|=4T+2$ holds.

Consider now some wall $B'_{i_r}\in \sset'$. Recall that there is at least one neighborhood bridge $F_r$ incident on $B'_{i_r}$. This bridge must contain a path $P_{i_r}$, connecting a vertex of $B'_{i_r}\setminus \Gamma_{i_r}'$ to a vertex of $\nset(B_{i_r})\setminus B'_{i_r}$, such that $P_{i_r}$ does not contain any vertices of $W'$ as internal vertices. We denote the endpoint of $P_{i_r}$ that belongs to $B'_{i_r}\setminus \Gamma'_{i_r}$ by $x_{i_r}$, and its other endpoint by $y_{i_r}$. Notice that since $\tau\geq 2t$, $x_{i_r}$ cannot belong to the sub-walls of $W'$ spanned by the top $2t$ or the bottom $2t$ rows of $W'$. All paths $\set{P_i}_{B'_i\in \sset'}$ are completely disjoint from each other, since for $r\neq r'$, $|i_r-i_{r'}|\geq 3$, so the neighborhood bridges $F_{i_r},F_{i_{r'}}$ must be disjoint.

Let $\rset_1$ be the set of the top $t$ rows of $W'$, $\rset_2$ the set of the bottom $t$ rows of $W'$, and $\rset_3$ the set of all remaining $z-2t$ rows of $W'$. For $1\leq j\leq 3$, let $W_j$ be the sub-wall of $W'$ spanned by the rows in $\rset_i$ and all columns of $W'$.
We partition $\sset'$ into two subsets, $\sset_1,\sset_2$, where $\sset_1$ contains all walls $B'_{i_r}\in\sset'$ with $y_{i_r}\in V(W_1\cup W_2)$, and $\sset_2$ containing all remaining walls. The following two lemmas will finish the proof.

\begin{lemma}
If $|\sset_1|\geq 2T+2$, then $G$ contains a $K_t$-minor. Moreover, if $B_1$ is a sub-wall of $W$, then $G$ contains a model of $K_t$-minor grasped by $W$, and it can be found efficiently, given $(\bset,\pset)$.
\end{lemma}

\begin{proof}
We will show that $G$ contains a graph $H\in \hset_2$ as a minor, and then invoke Theorem~\ref{thm: family H2 to clique}. We start with the graph $G'=W'\cup\left(\bigcup_{B'_i\in \sset_1}P_i\right)$. For each column $C_j$ and row $R_i$ of $W'$, we contract all edges in $C_j\cap R_i$, obtaining a graph $H''$, which is a subdivision of a grid. We then turn $H''$ into a grid $H'$, as follows: for each maximal $2$-path $P$ of $H''$ that does not contain the corners of $H''$, we contract all but one edges of $P$. So far we have contracted edges of $W'$ to turn it into a grid, but we made no changes in the paths $P_i$, for $B'_i\in \sset_1$. Our last step is to contract all but one edges in each such path $P_i$. It is immediate to verify that the resulting graph belongs to the family $\hset_2$, and therefore, from Theorem~\ref{thm: family H2 to clique}, it contains a $K_t$-minor.

Let $J_1$ be the sub-graph of $H'$ spanned by the first $t$ columns of $H'$. If $B_1$ is a sub-wall of $W$, then $J_1$ is a contraction of a sub-wall of $W$, and so from Theorem~\ref{thm: family H2 to clique}, graph $G$ contains a $K_t$-minor grasped by $W$, and it can be found efficiently, given $(\bset,\pset)$.
\end{proof}

\begin{lemma}\label{lemma: type-1 walls easy case}
If $|\sset_2|\geq 2T$, then $G$ contains a $K_t$-minor. Moreover, if $B_1$ is a sub-wall of $W$, then $G$ contains a model of a $K_t$-minor grasped by $W$, and it can be found efficiently, given $(\bset,\pset)$.
\end{lemma}

\begin{proof}
The idea of the proof is to define a different partition $\bset'$ of $W'$ into basic walls, such that for each original basic wall $B_i$ with $B_i'\in \sset_2$, there is a basic wall  $\tilde B_i\in \bset'$ with $\nset(B_i)\subseteq \tilde B_i$. Let $\tilde B'_i$ be the $(t+1)$-core sub-wall of $\tilde B_i$, and let $Z_i$ be the graph $\tilde B'_i\cup P_i$. We show that $Z_i$ contains a wall-cross for $\tilde B'_i$, and so we reduce the problem to the case where the number of $(t+1)$-core walls of type $3$ is at least $2T$. Applying Theorem~\ref{thm: type-3 walls} then finishes the proof.

Formally, for each wall $B'_i\in \sset_2$, let $i_1$ be the index of the first column of $W'$ whose vertices are contained in $\nset(B_i)$, and let $i_2$ be the index of the last column of $W'$  whose vertices are contained in $\nset(B_i)$. We define a set $\cset_i$ of consecutive columns of $W'$ to contain all columns starting from $C_{i_1}$ and ending with $C_{i_2}$. Let $\tilde B_i$ be the sub-wall of $W'$ spanned by columns in $\cset_i$ and all rows of $W'$. Notice that for $B_i',B_j'\in \sset_2$, if $i\neq j$, then $\cset_i\cap \cset_j=\emptyset$ due to the choice of $\sset'$.  Let $\cset^*$ be the set of the first $t$ columns of $W'$. We define one additional basic wall $\tilde B^*$ to be the sub-wall of $W'$ spanned by the columns in $\cset^*$ and all rows of $W'$. Let $\bset'$ be the set of these new basic walls.

Let $W''$ be the sub-wall of $W'$ obtained by taking the union of  all the columns of $W'$ that are contained in the walls in $\bset'$, and all the rows of $W'$. Then $W''$ and the walls in $\bset'$ define a chain of $2T$ walls of height $z$. 
Let $G'$ be the union of $W''$ with all paths $P_i$ for $B_i'\in \sset_2$. 

For convenience, for each $B_i'\in \sset_2$, we re-name the path $P_i$ to be $P(\tilde B)$ where $\tilde B\in \bset'$ is the basic wall containing $B_i'$.

For each $\tilde B\in \bset'$, where $\tilde B\neq \tilde B^*$, let $\tilde B'$ be the $(t+1)$-core sub-wall of $\tilde B$, and let $G(\tilde B)$ be the union of $\tilde B$ with the path $P(\tilde B)$ (whose both endpoints must be contained in $\tilde B'$). Notice that the endpoints of $P(\tilde B)$ are separated by at least one row or at least one column, as one endpoint of $P(\tilde B)$ belongs to $B'_j\setminus \Gamma'_j$ and the other to $\nset(B_j)\setminus V(B_j')$ for the corresponding $2t$-core wall $B'_j\in \sset_2$. From Theorem~\ref{thm: cross-wall}, graph $G(\tilde B)$ contains a wall-cross for $\tilde B'$. Therefore, the chain of walls defined by $W''$ and $\bset'$ contains $2T$ $(t+1)$-core walls, that are type-3 walls in graph $G'$. From Theorem~\ref{thm: type-3 walls}, $G'$ must contain a $K_t$-minor. Moreover, if $B_1\in \bset$ is a sub-wall of $W$, then so is $\tilde B^*\in \bset'$. Therefore, from Theorem~\ref{thm: type-3 walls}, $G'$ must contain a $K_t$-minor grasped by $W$, and it can be found efficiently given $(\bset,\pset)$.
\end{proof}
\end{proof}

\label{----------------------------------------Proof of main weak thm--------------------------------}
\section{Proof of Theorem~\ref{thm: main weak}}\label{sec: proof of main week thm}
We assume w.l.o.g. that graph $G$ is connected - otherwise it is enough to prove the theorem for the connected component of $G$ containing $W$.
We set $z'=w+4t$. Let $W$ be the $R\times R$ wall in $G$. Using Theorem~\ref{thm: getting the chain}, we can build a chain $(\bset,\pset)$ of $N=8D^2(10T+6)+14T+8$ basic walls of height at least $z'$ in $G$, such that each wall $B_i\in \bset$ is a sub-wall of $W$. 
Let $W'=W'(\bset,\pset)$ be the sub-graph corresponding to the chain of walls.
Throughout the proof, we will set $\tau=2t$, and we will consider the set $\sset^*$ of all $\tau$-core walls $B'_i$ with $1<i<N$, so $|\sset^*|\geq 8D^2(10T+6)+14T+6$.

If at least one of the core walls $B'_i\in \sset^*$ is a type-4 wall, then from Corollary~\ref{cor: type 4 gives flat wall}, we can find a flat sub-wall of $W$ of size $((z'-2\tau)\times (z'-2\tau))=(w\times w)$. Therefore, we assume from now on that no core wall in $\sset^*$ is a type-4 wall.

If the number of type-3 core walls in $\sset^*$ is at least $2T$, then from Theorem~\ref{thm: type-3 walls}, we can efficiently construct a $K_t$-minor in $G$ grasped by $W$. If the number of type-1 core walls in $\sset^*$ is at least $12T+6$, then from Theorem~\ref{thm: type-1 walls}, we can construct a $K_t$-minor in $G$ grasped by $W$. Therefore, we assume from now on that the number of type-$2$ walls in $\sset^*$ is at least $8D^2(10T+6)+14T+6-2T-(12T+6)=8D^2(10T+6)$.

Let $\sset\subseteq \sset^*$ be the set of all core walls $B'_i$ of type 2, with $1<i<N$. We select a subset $\sset'\subseteq \sset$ of $4D^2(10T+6)$ core walls, such that every consecutive pair of such walls is separated by at least one wall. In other words,  if we denote $\sset'=\set{B'_{i_1},B'_{i_2},\ldots,B'_{i_{4D^2(10T+6)}}}$, where $1<i_1<i_2<\ldots<i_{4D^2(10T+6)}<N$, then for all $1\leq r<4D^2(10T+6)$, $i_{r+1}\geq i_r+2$. Subset $\sset'$ can be found using a standard greedy procedure, by ordering the walls in $\sset$ in the ascending order of their indices, and then choosing all walls whose location in this ordering is even. We then delete walls from $\sset'$ as necessary to ensure that $|\sset'|=4D^2(10T+6)$.
 We need the following claim.

\begin{claim}\label{claim: find the matching}
There is an efficient algorithm to find a set $\pset$ of $10T+6$ disjoint paths in $G$, such that:

\begin{itemize}
\item For each path $P\in \pset$, its endpoints are labeled $x_P$ and $y_P$. There is a core wall $B'_{i_P}\in \sset'$, such $x_P\in V(B'_{i_P}\setminus \Gamma'_{i_P})$, and $y_P\in V(W'\setminus \nset(B_{i_P}))$. Moreover, if $P,P'\in \pset$ are distinct, then $i_P\neq i_{P'}$.

\item The paths in $\pset$ are internally disjoint from $W'$.
\end{itemize}
\end{claim}

\begin{proof}
We start by constructing a set $\pset'$ of paths that have all the required properties, except that they may share endpoints (that is, instead of being completely disjoint, paths in $\pset'$ are only guaranteed to be internally disjoint). 

Consider some core wall $B'_i\in \sset'$. Recall that, since $B'_i$ is a type-2 wall, there is at least one non-neighborhood bridge incident on $B'_i$. Fix any such bridge $F$, and let $x(B'_i)$ be any vertex of $B'_i\setminus \Gamma'_i$ that $F$ touches. If $F$ consists of a single edge $e$, then we add this edge to $\pset'$ as one of the paths $P=\set{e}$. We let $x_P=x(B'_i)$, and $y_P$ the other endpoint of $e$. Since $F$ is a non-neighborhood bridge, $y_P\in V(W')\setminus \nset(B_i)$ as required. Otherwise, the bridge $F$ is a connected component of $G\setminus W'$ that touches $x(B'_i)$. Let $e_i$ be any edge connecting $x(B'_i)$ to some vertex $u_i\in V(F)$. We say that $B'_i$ \emph{tags} the component $F$ via the edge $e_i$ and the vertex $u_i$, and we call $x(B'_i)$ a \emph{terminal of $F$}.

Consider now some connected component $F$ of $G\setminus W'$. Assume first that exactly one core wall $B'_i$ tags $F$. Let $e_i$ be the edge via which $B'_i$ tags $F$, and $u_i$ the endpoint of $e_i$ that lies in $F$. Since $F$ is a non-neighborhood bridge incident on $B'_i$, $F$ must touch some vertex $y\in W'\setminus \nset(B_i)$. Let $P$ be any path connecting the vertex $u_i$ to $y$, where all inner vertices of $P$ belong to $F$, and let $P_i$ be the path obtained by concatenating $e_i$ and $P$. We then add $P_i$ to $\pset'$, setting $x_P=x(B'_i)$ and $y_P=y$.

Finally, assume that more than one core wall tags $F$. Let $F'$ be the union of $F$ and all edges via which $F$ is tagged. Let $T$ be any spanning tree of $F'$, rooted at some vertex whose degree is more than $1$. Notice that the terminals of $F$ must serve as the leaves of $T$, since the degree of each terminal is $1$ in $F'$. We now construct a collection $\pset(T)$ of node-disjoint paths connecting some terminals of $F$ to each other, such that all paths in $\pset(T)$ are contained in $F'$. We start with $\pset(T)=\emptyset$. While $T$ contains at least two terminals, we select the lowest vertex $v$ of $T$ whose subtree contains at least two terminals. Since the maximum vertex degree is bounded by $D$, the subtree of $v$ contains at most $D$ terminals if $v$ is the root of $T$, and at most $D-1$ terminals otherwise. We select any pair of terminals in the sub-tree of $v$ and connect them with a path $P$ contained in the subtree of $v$. This path is then added to $\pset(T)$, and all vertices of the subtree of $v$ (including $v$ itself) are deleted from $T$. We label the endpoints of $P$ as $x_P$ and $y_P$ arbitrarily. Notice that if $x_P=x(B'_i)$ for some $B'_i\in \sset'$, then $y_P\not\in \nset(B_i)$, since all core walls in $\sset'$ are separated by at least one wall. We continue this procedure until $T$ contains at most one terminal, and we terminate the algorithm at that point. Let $n_F$ be the number of terminals of $F$. Then in every iteration except for the last one, we add one path to $\pset(T)$ and delete at most $D-1$ terminals from the tree $T$. In the last iteration, we may add one path to $\pset(T)$ and delete up to $D$ terminals, or we may delete one terminal and not add any paths to $\pset(T)$. Therefore, $|\pset(T)|\geq (n_F-1)/D\geq n_F/(2D)$. We add the paths of $\pset(T)$ to $\pset'$. Once we process all connected components of $G\setminus W'$, set $\pset'$ must contain at least $\frac{|\sset'|}{2D}\geq 2D(10T+6)$ paths. It is easy to see from the construction that all these paths are internally node-disjoint, and each such path is internally disjoint from $W'$. For each path $P\in \pset'$, there is a distinct core wall $B'_i\in \sset'$ such that $x_P\in V(B'_i)\setminus V(\Gamma'_i)$, while $y_P\in W'\setminus \nset(B_i)$. 

Our last step is to select a subset $\pset\subseteq \pset'$ such that all paths in $\pset$ are completely disjoint.
 Start with $\pset=\emptyset$. While $\pset'\neq \emptyset$, select any path $P\in \pset'$ and add it to $\pset$. Delete from $\pset'$ the path $P$, and all paths sharing endpoints with $P$. Since the degree of every vertex is at most $D$, and the paths in $\pset'$ are internally disjoint, in every iteration we delete at most $2D-1$ paths from $\pset'$, and add one path to $\pset$. Therefore, $|\pset|\geq |\pset'|/2D\geq 10T+6$ as required, and the paths in $\pset$ are mutually disjoint.
\end{proof}

Finally, we show that $G$ contains a graph $H\in \hset_3$ as a minor, and then invoke Theorem~\ref{thm: family H3 to clique}. We start with the graph $G'=W'\cup\left(\bigcup_{P\in \pset}P\right)$. For each column $C_j$ and row $R_i$ of $W'$, we contract all edges in $C_j\cap R_i$, obtaining a graph $H''$, which is a subdivision of a grid. We then turn $H''$ into a grid $H'$, as follows: for each maximal $2$-path $P$ of $H''$ that does not contain the corners of $H''$, we contract all but one edges of $P$. So far we have contracted edges of $W'$ to turn it into a grid, but we have made no changes in the paths $P\in \pset$. Our last step is to contract, for each path $P\in \pset$, all but one edges of $P$. It is immediate to verify that the resulting graph belongs to the family $\hset_3$, and therefore, from Theorem~\ref{thm: family H3 to clique}, it contains a $K_t$-minor.
Let $J_1$ be the sub-graph of $H'$ spanned by the first $t$ columns of $H'$. Since $B_1$ is a sub-wall of $W$,  $J_1$ is a contraction of a sub-wall of $W$, and so from Theorem~\ref{thm: family H3 to clique}, graph $G$ contains a $K_t$-minor grasped by $W$, and it can be found efficiently.


\label{----------------------------------------Proof of main strong thm--------------------------------}
\section{Proof of Theorem~\ref{thm: main strong}}\label{sec: proof of main strong theorem}
We assume w.l.o.g. that graph $G$ is connected - otherwise it is enough to prove the theorem for the connected component of $G$ containing $W$.
We set $z'=w+4t$. Let $W$ be the $R\times R$ wall in $G$. Using Theorem~\ref{thm: getting the chain}, we can build a chain $(\bset,\pset)$ of $N=500T+200$ basic walls of height at least $z'$ in $G$, such that each wall $B_i\in \bset$ is a sub-wall of $W$. We set $\tau=2t$, and we will consider the set $\sset^*$ of all $\tau$-core walls $B'_i$ with $1\leq i\leq N$. Let $\Gamma'_i$ be the boundary of the $\tau$-core wall $B'_i$, and let $W'=W'(\bset,\pset)$ be the sub-graph of $G$ corresponding to $(\bset,\pset)$.

For each $B_i\in \bset$, we define a pair of vertex subsets $X_i=V(B_i')\setminus V(\Gamma_i')$ and $Y_i=V(W')\setminus \nset(B_i)$. Notice that $X_i\cap Y_i=\emptyset$ and each pair $(x,y)$ of vertices with $x\in X_i$, $y\in Y_i$ is separated by at least $t+1$ columns in $W'$. We denote $M_i=(X_i,Y_i)$, and we call it a \emph{demand pair} for $B_i$. We say that a path $P$ \emph{routes} the pair $M_i$, iff one of the endpoints of $P$ belongs to $X_i$, the other endpoint belongs to $Y_i$, and $P$ is internally disjoint from $W'$. Notice that if the endpoints of a path $P$ belong to two distinct sets $X_i,X_j$, then it is possible that $P$ routes both the pairs $M_i,M_j$. 
We say that the pair $M_i$ is \emph{routable} in a sub-graph $H$ of $G$, iff there is a path $P$ in $H$ that routes $M_i$. We start with the following theorem, whose proof is almost identical to the proof of Lemma 2.1 in~\cite{KTW}.

\begin{theorem}\label{thm: find apex vertices}
There is an efficient algorithm, that returns one of the following: either (1) a set $A$ of at most $40T+20$ vertices of $G$, and a set $\bset'\subseteq \bset\setminus\set{B_1,B_N}$ of least $396T+3t+150$ walls, such that for each $B_i\in \bset'$, $V(B_i)\cap A=\emptyset$, and $M_i$ is not routable in $G\setminus A$, or (2) a set $\pset^*$ of $10T+6$ disjoint paths in $G$, such that:

\begin{itemize}
\item For each path $P\in \pset^*$, the endpoints are labeled $x_P$ and $y_P$, and there is $1<i_P<N$, such that $x_P\in X_{i_P}$ and $y_P\in Y_{i_P}$;
\item If $P,P'\in \pset^*$ and $P\neq P'$, then $|i_P- i_{P'}|>1$; and

\item All paths in $\pset^*$ are internally disjoint from $W'$.
\end{itemize}
\end{theorem}

\begin{proof}
The proof uses arguments similar to those in~\cite{KTW}. Throughout the algorithm, we maintain two sets of paths, $\pset^*$ and $\qset$, such that the following invariants hold:
\begin{itemize}
\item Paths in $\pset^*$ are disjoint from each other, and are internally disjoint from $W'$.
\item For each path $P\in \pset^*$, the endpoints are labeled $x_P$ and $y_P$, and there is $1<i_P<N$, such that $x_P\in X_{i_P}$ and $y_P\in Y_{i_P}$;
\item If $P,P'\in \pset^*$ and $P\neq P'$, then $|i_P- i_{P'}|>1$.

Let $S=\set{B_i\in \bset \mid \exists P\in \pset^*\mbox{ such that $x_P$ or $y_P$ belong to $\nset(B_i)$}}$. Then:

\item For each path $Q\in\qset$, its endpoints are labeled $x_Q,a_Q$, such that there is $1<i_Q<N$ with $B_{i_Q}\not\in S$ and $x_Q\in X_{i_Q}$. There is also a path $P_Q\in \pset^*$, such that $a_Q$ is an {\bf inner} vertex of $P_Q$.
\item If $Q,Q'\in \qset$ are distinct paths, then $i_Q\neq i_{Q'}$ and $P_Q\neq P_{Q'}$.
\item All paths in $\qset$ are disjoint from each other and do not contain the vertices of $W'\cup (\bigcup_{P\in \pset^*}P)$ as inner vertices.
\end{itemize}

We start with $\pset^*,\qset=\emptyset$. Clearly, all invariants hold for this choice of $\pset^*$ and $\qset$. In every iteration, we either add one path to $\pset^*$ and set $\qset=\emptyset$, or we add one path to $\qset$ without changing $\pset^*$. An iteration is executed as follows. If $|\pset^*|\geq 10T+6$, then we stop and output $\pset^*$. From our invariants, it is clear that $\pset^*$ has all the desired properties. We assume from now on that $|\pset^*|<10T+6$, and so $|\qset|\leq |\pset^*|<10T+6$. Let $\bset^*\subset \bset\setminus\set{B_1,B_N}$ be the set of all walls $B_i$, such that an endpoint of some path in $\pset^*\cup \qset$ belongs to $\nset(B_i)$, and let $\bset'=\bset\setminus (\bset^*\cup\set{B_1,B_N})$. Then $|\bset^*|\leq 3(2|\pset^*|+|\qset|)\leq 9|\pset^*|\leq 90T+45$, and $|\bset'|\geq |\bset|-90T-45-2\geq 396T+3t+150$. 
Let $A$ be the set of all endpoints of the paths in $\pset\cup \qset$, so $|A|\leq 40T+20$.
We consider two cases.
We say that Case 1 happens if for some wall $B_i\in \bset'$, there is a path $Z$ routing $M_i$ in graph $G\setminus A$, and otherwise we say that Case 2 happens. If Case 2 happens, then we terminate the algorithm and return $\bset'$ and $A$. Clearly, this is a valid output. Otherwise, let $Z$ be any path routing some demand pair $M_i$ with $B_i\in \bset'$ in graph $G\setminus A$. Let $x_Z$ be the endpoint of $Z$ that belongs to $X_i$, and $y_Z$ its other endpoint. If $Z$ is disjoint from all paths in $\pset^*$, then we simply add $Z$ to $\pset^*$ and set $\qset=\emptyset$. It is easy to verify that all invariants continue to hold. In particular, since no vertex of $A$ belongs to $\nset(B_i)$, for each path $P\in \pset^*$ with $P\neq Z$, $|i_P-i|>1$. Therefore, we assume from now on that there is some path $P\in \pset^*$, such that $Z\cap P\neq \emptyset$.
Let $v$ be the first vertex on path $Z$ (counting from $x_Z$) that belongs to $V(\pset^*\cup \qset)$. Assume first that $v$ lies on some path $Q\in \qset$.
From the definition of $\bset'$, $x_Q\not\in \nset(B_i)$, and so $x_Q\in Y_{i}$. Moreover, since $a_Q,x_Q\in A$, $v\neq a_Q$. Let $P^*$ be the concatenation of the segment of $Z$ from $x_Z$ to $v$, and the segment of $Q$ from $x_Q$ to $v$. We add path $P^*$ to $\pset^*$, setting $\qset=\emptyset$. We set $x_{P^*}=x_Z$, $y_{P^*}=x_Q$ and $i_{P^*}=i$. It is easy to verify that $P^*$ is disjoint from all other paths in $\pset^*$, and all other invariants continue to hold.

Assume now that $v\in P'$ for some path $P'\in \pset^*$. Since $x_{P'},y_{P'}\in A$, $v$ is not an endpoint of $P'$. If there is no path $Q\in \qset$ with $P_Q=P'$, then we add $Z$ to $\qset$, setting $a_{Z}=v$ and $P_{Z}=P'$. Therefore, we assume from now on that there is some path $Q\in \qset$ with $P_Q=P'$. Assume w.l.o.g. that $v$ lies closer to $x_{P'}$ on path $P'$ than $a_Q$. We define two new paths: path $P_1$ is the concatenation of the segment of $Z$ from $x_Z$ to $v$ and the segment of $P'$ from $x_{P'}$ to $v$, and path $P_2$ is the concatenation of $Q$ and the segment of $P'$ from $a_Q$ to $y_{P'}$ (see Figure~\ref{fig: two paths}). We set $x_{P_1}=x_{Z}$, $y_{P_1}=x_{P'}$, $i_{P_1}=i$ (the index of the wall $B_i$ with $x_Z\in X_i$), $x_{P_2}=x_{Q}$, $y_{P_2}=y_{P'}$, and $i_{P_2}$ the index of the wall $B_{i_Q}$ with $x_Q\in X_{i_Q}$. We then delete $P'$ from $\pset^*$, add $P_1,P_2$ to $\pset^*$, and set $\qset=\emptyset$. It is immediate to verify that all invariants continue to hold.

\begin{figure}[h]
\scalebox{0.4}{\includegraphics{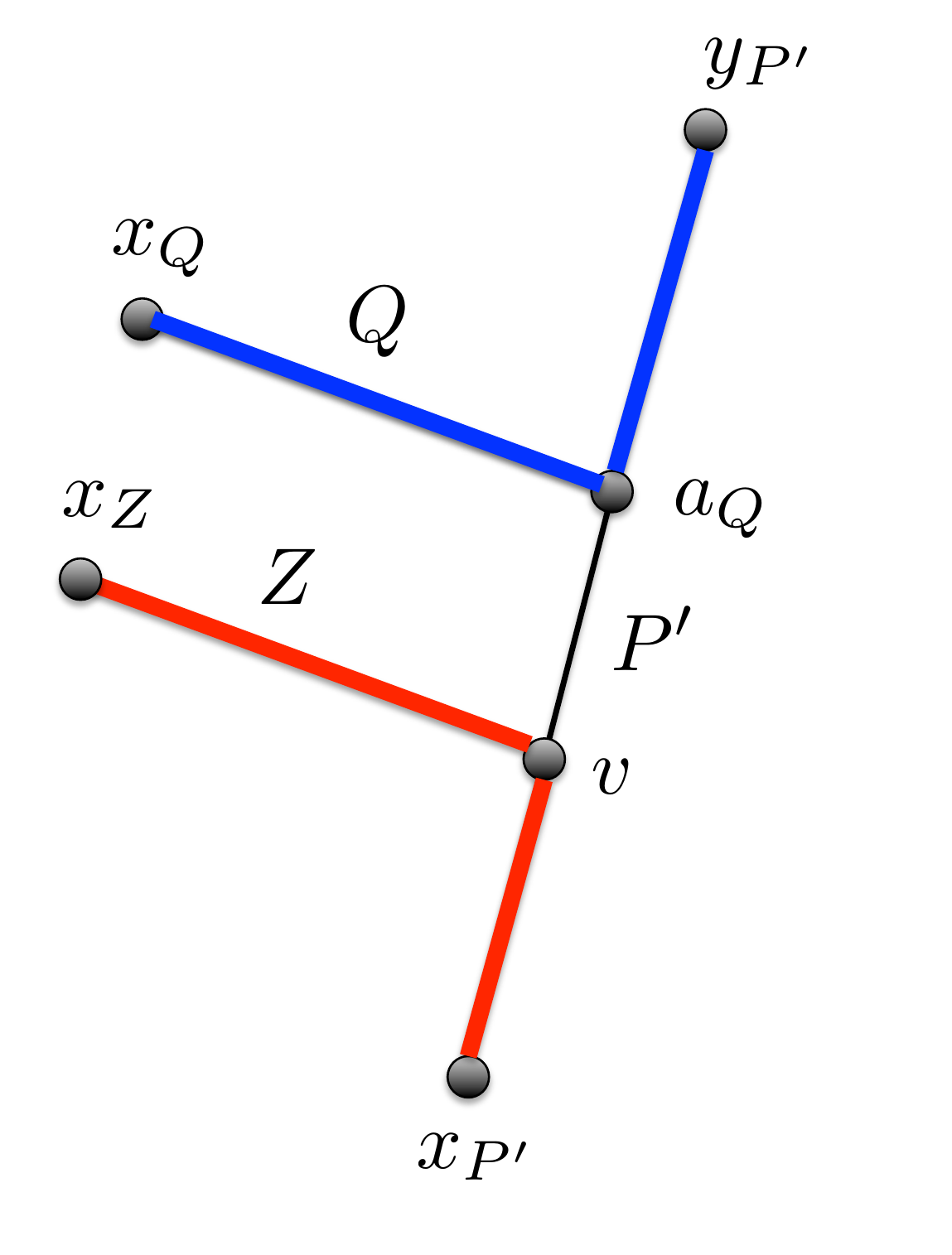}}\caption{Path $P_1$ is shown in red and $P_2$ in blue.\label{fig: two paths}}
\end{figure}

From the above discussion, the algorithm returns a valid output. Each iteration can be implemented efficiently. Since in every iteration, either $|\pset^*|$ increases, or $|\qset|$ increases while $|\pset^*|$ remains the same, and since $|\qset|\leq |\pset^*|$ always holds, the number of iterations is bounded by $(10T+6)^2$.
\end{proof}

We apply Theorem~\ref{thm: find apex vertices} to our chain of walls $(\bset,\pset)$.
 Assume first that the outcome of Theorem~\ref{thm: find apex vertices} is a set $\pset^*$ of $10T+6$ paths.
We show that $G$ contains a graph $H\in \hset_3$ as a minor, exactly as in the proof of Theorem~\ref{thm: main weak}, and then invoke Theorem~\ref{thm: family H3 to clique}. We start with the graph $G'=W'\cup\left(\bigcup_{P\in \pset^*}P\right)$. For each column $C_j$ and row $R_i$ of $W'$, we contract all edges in $C_j\cap R_i$, obtaining a graph $H''$, which is a subdivision of a grid. We then turn $H''$ into a grid $H'$, as follows: for each maximal $2$-path $P$ of $H''$ that does not contain the corners of $H''$, we contract all but one edges of $P$. So far we have contracted edges of $W'$ to turn it into a grid, but we made no changes in the paths $P\in \pset^*$. Our last step is to contract, for each path $P\in \pset^*$, all but one edges of $P$. It is easy to verify that the resulting graph belongs to the family $\hset_3$. Indeed, Theorem~\ref{thm: find apex vertices} ensures that for $P,P'\in \pset^*$ where $P\neq P'$, the vertices $x_P$ and $x_{P'}$ belong to core walls $B'_{i_P},B'_{i_{P'}}$, with $|i_P-i_{P'}|>1$. In other words, the two walls are separated by at least one wall, and $x_P,x_{P'}$ are separated by at least $2t$ columns. The definition of the pairs $M_i$ ensures that for each path $P\in \pset^*$, $y_P$ and $x_P$ are also separated by at least $2t$ columns. For every $P\in \pset^*$, $x_P\in X_{i_P}=B'_{i_P}\setminus \Gamma'_{i_P}$, so $x_P$ does not lie in the top $2t$ or the bottom $2t$ rows of the grid.
We can now apply Theorem~\ref{thm: family H3 to clique} to find a $K_t$-minor grasped by $W$. 

Assume now that the outcome of Theorem~\ref{thm: find apex vertices} is a set $A$ of at most $40T+20$ vertices of $G$, and a set $\bset'\subseteq \bset\setminus\set{B_1,B_N}$ of at least $396T+3t+150$ walls, such that for each $B_i\in \bset'$, $V(B_i)\cap A=\emptyset$, and $M_i$ is not routable in $G\setminus A$. Let $\bset''\subseteq \bset'$ be a subset of $132T+t+50$ walls, such that for each pair $B_i,B_{i'}\in \bset''$, with $i\neq i'$, $|i-i'|\geq 3$. We can find $\bset''$ by standard methods: order the walls in $\bset'$ in their natural left-to-right order, and select all walls whose index is $1$ modulo $3$ in this ordering, discarding any excess walls as necessary, so $|\bset''|=132T+t+50$. Lastly, we would like to ensure that $\bset''$ does not contain walls $B_i$ with $1\leq i\leq t-3$ and $N-t+4\leq i\leq N$, by simply removing all such walls from $\bset''$. Since there are at most $t-3$ such walls in $\bset''$, the final size of $\bset''$ is at least $132T+50$.  Notice that if $B_i,B_j\in \bset''$ with $i\neq j$, then there is no path $P$ in $G\setminus A$, such that $P$ is internally disjoint from $W'$ and it connects $X_i$ to $V(B_j)$, since $V(B_i)\subseteq Y_i$ and  $M_i$ is not routable in $G\setminus A$.

Assume that $A=\set{a_1,\ldots,a_{m}}$, where $m\leq 40T+20$. Our next step is to gradually construct, for each $1\leq j\leq m$, a collection $\qset_j$ of paths, where for each path $Q\in \bigcup_{j=1}^m\qset_j$ there is an index $i_Q$ with $B_{i_Q}\in \bset''$, such that $Q$ starts at a vertex of $X_{i_Q}$, and for $Q\neq Q'$, $i_Q\neq i_{Q'}$. All paths in set $\qset_j$ must terminate at $a_j$, and all paths in $ \bigcup_{j=1}^m\qset_j$ are internally disjoint from $W'\cup A$, and mutually disjoint from each other, except for possibly sharing their last endpoint (we view the paths as directed towards the vertices of $A$).

We start with $\qset_j=\emptyset$ for all $1\leq j\leq m$. We say that vertex $a_j\in A$ is \emph{active} iff $|\qset_j|<2t$. Let $A^*\subseteq A$ be the set of all  vertices that are inactive in the current iteration. We say that a wall $B_i\in \bset''$ is \emph{active}, iff no path of $\bigcup_{j=1}^m\qset_j$ starts at a vertex of $X_i$. An iteration is executed as follows. Assume that there is a path $Q$ in $G\setminus A^*$, connecting a vertex $v\in X_i$, for some active wall $B_i$, to some active vertex $a_j\in A\setminus A^*$, such that $Q$ contains no vertices of $W'\cup A$ as inner vertices. We claim that $Q$ is disjoint from all paths $Q'\in \bigcup_{j=1}^m\qset_j$, except possibly for sharing the last vertex $a_j$ with $Q'$. Indeed, assume for contradiction that $Q'$ and $Q$ share some vertex other than $a_j$, say vertex $u$. Let $B_{i_{Q'}}$ be the wall to which the first vertex of $Q'$ belongs. Then, since $B_i$ is still active, $i_{Q'}\neq i$ must hold, and so $V(B_{i_{Q'}})\subseteq Y_i$. Concatenating the segments of $Q$ and $Q'$ between their starting endpoints and $u$, we obtain a path connecting $X_i$ to $Y_i$, that does not contain any vertices of $A$ and is internally disjoint from $W'$, contradicting the fact that $M_i$ is not routable in $G\setminus A$. Therefore, $Q$ is disjoint from all paths $Q'\in \bigcup_{j=1}^m\qset_j$, except for possibly sharing its last vertex $a_j$ with $Q'$. We then add $Q$ to $\qset_j$, and continue to the next iteration. It is easy to see that each iteration can be computed efficiently. The algorithm terminates when we cannot make progress anymore: that is, for each active wall $B_i$, there is no path  $Q$ connecting a vertex in $X_i$ to a vertex in $A\setminus A^*$, such that $Q$ is internally disjoint from $W'\cup A$. It is easy to see that the number of iterations is bounded by $|\bset''|$, and so the algorithm can be executed efficiently. Consider the final set $A^*$ of inactive vertices, and the final set $\bset^*$ of active walls. We say that Case 1 happens if $|A^*|\geq t-4$; we say that Case 2 happens if $|A^*|\leq t-5$, but $|\bset''\setminus \bset^*|\geq 80T+40+6t^2$; otherwise we say that Case 3 happens. We analyze each of the three cases separately.

\paragraph{Case 1}
We show that if Case 1 happens, then we can find a model of $K_t$ grasped by $W$.
We define a new graph $Z$, whose vertex set is $V(Z)=\set{v_1,\ldots,v_t,u_1,\ldots,u_{t-4}}$, and the set of edges is a union of two subsets:

\[E_1=\set{(v_i,u_j)\mid 1\leq i\leq t; 1\leq j\leq t-4}\quad\quad\mbox{and}\quad\quad E_2=\set{(v_i,v_j)\mid 1\leq i<j\leq 4}.\]

 In other words, $Z$ is obtained from $K_{t,t-4}$, by adding the $6$ edges connecting all pairs of vertices in $\set{v_1,\ldots,v_4}$. It is easy to see that $Z$ contains a $K_t$-minor, by contracting, for each $5\leq i\leq t-4$, the edge $(v_i,u_{i+4})$. We will show that $G$ contains graph $Z$ as a minor, and provide an efficient algorithm for embedding $Z$ into $G$. It is then easy to find a model of $Z$, and consequently, a model of $K_t$ in $G$.

Since we assume that Case 1 happened, $|A^*|\geq t-4$ when the algorithm terminates. Let $a_1,\ldots,a_{t-4}$ be arbitrary $t-4$ vertices of $A^*$. We will embed, for each $1\leq i\leq t-4$, vertex $u_i$ of $Z$ into $\set{a_i}$: the connected sub-graph of $G$ consisting of only the vertex $a_i$.

We say that a basic wall $B_i\in \bset''$ is \emph{bad} iff some vertex of $a_1,\ldots,a_{t-4}$ belongs to $\nset(B_i)$. Since every pair of walls in $\bset''$ is separated by at least two walls, the number of bad walls in $\bset''$ is at most $t-4$.

 Consider now some vertex $a_j$, for $1\leq j\leq t-4$, and the corresponding set $\qset_j$ of $2t$ paths.  We discard from $\qset_j$ all paths that originate at a vertex that belongs to a bad wall. We also discard additional paths from $\qset_j$ as needed, until $|\qset_j|=t$ holds.
Let $\qset_j=\set{Q_1^j,\ldots,Q_t^j}$ be this final set of paths, and for each $1\leq i\leq t$, we denote by $x_i^j$ the first endpoint of path $Q_i^j$. We assign the label $i$ to $x_i^j$, denoting $\ell(x_i^j)=i$. Let $\qset=\bigcup_{j=1}^{t-4}\qset_j$, and let $\tilde A=\set{a_1,\ldots,a_{t-4}}$.

Our next step is to define a collection $\lset=\set{L_1,\ldots,L_t}$ of $t$ disjoint paths contained in $W'\setminus \tilde A$, such that, for each $1\leq i\leq t$, path $L_i$ contains all vertices whose label is $i$. We will then embed each vertex $v_i$ of $Z$ into the path $L_i$. The edge $(v_i,u_j)$ of $Z$, for $1\leq i\leq t$, $1\leq j\leq t-4$, will then be embedded into $Q_i^j$. Finally, we will define a new set $\sset$ of $6$ disjoint paths contained in $W'\setminus \tilde A$, that connect every pair of paths in $\set{L_1,\ldots,L_4}$. We will ensure that the paths in $\sset$ are internally disjoint from the paths in $\lset$. They are also guaranteed to be internally disjoint from the paths in $\qset$, since all paths in $\sset$ are contained in $W'$. The paths in $\sset$ will be used to embed the edges of $E_2$.

Recall that the walls $B_1,\ldots,B_{t-3}$ do not belong to $\bset''$. Let $B_{i^*}\in \set{B_1,\ldots,B_{t-3}}$ be any of these walls that does not contain vertices of $\tilde A$. Similarly, let $B_{i^{**}}\in\set{B_{N-t+4,\ldots,B_N}}$ be any wall that does not contain vertices of  $\tilde A$. Notice that for each wall $B_i\in \bset''$, $i^*<i<i^{**}$ must hold. Let $\tilde{\bset}\subseteq \bset''$ be the set of all walls $B_i$ that contain the vertices $x_q^j$, for $1\leq j\leq t-4$, $1\leq q\leq t$.

Let $\rset$ be any set of $t$ rows of $W'$, such that no vertex in $\tilde A$ lies in a row of $\rset$. Since there are $4t+w$ rows in $W'$, and $|\tilde A|=t-4$, such a set exists. We assume that the paths in $\rset$ are ordered in their natural top-to-bottom order.

Consider now some vertex $x^j_q$, for some $1\leq j\leq t-4$, $1\leq q\leq t$, and assume that $x^j_q$ belongs to some basic wall $B_i\in \tilde \bset$. Recall that the label of $x^j_q$ is $q$, and $\nset(B_i)$ does not contain vertices of $\tilde A$. Let $S_1$ be the set of $t$ vertices lying in the first column of $B_{i-1}$, that belong to the rows in $\rset$, such that exactly one vertex from each row in $\rset$ belongs to $S_1$. Define $S_2$ similarly for the last column of $B_{i+1}$. 
We will construct a set $\lset(B_i)$ of $t$ disjoint paths, contained in $W'[\nset(B_i)]$, connecting the vertices of $S_1$ to the vertices of $S_2$, such that the $q$th path of $\lset(B_i)$ in their natural top-to-bottom order contains $x^j_q$.

 Assume first that $x^j_q$ belongs to some row $R_{s}$ of $B_{i}$. Then $2t<s<z'-2t+1$ must hold, as $x_q^j\in X_i$.
Let $\rset'$ be the set of the top $q-1$ rows of $W'$, the bottom $t-q$ rows of $W'$, and the row $R_s$ (so row $R_s$ is the $q$th row in the set $\rset'$ in their natural top-to-bottom order). Let $T_1$ be a set of $t$ vertices lying in the last column of $B_{i-1}$, that belong to the rows of $\rset'$, such that exactly one vertex from each row in $\rset'$ belongs to $T_1$. Define $T_2$ similarly for the first column of $B_{i+1}$. We now build three sets of paths: $\lset^1$ is a set of $t$ disjoint paths contained in $B_{i-1}$, connecting the vertices of $S_1$ to the vertices of $T_1$; $\lset^2$ is the set of $t$ disjoint paths containing, for each row $R\in \rset'$, the segment of $R$ between the unique vertex of $\rset'\cap T_1$ and the unique vertex of $\rset'\cap T_2$, and $\lset^3$ is a set of $t$ disjoint paths contained in $B_{i+1}$, connecting the vertices of $T_2$ to the vertices of $S_2$ (the existence of the sets $\lset^1,\lset^3$ of paths follows from Claim~\ref{claim: t-linkedness of wall}). Let $\lset(B_i)$ be the concatenation of the paths in $\lset^1,\lset^2,\lset^3$. Then the $q$th path of $\lset(B_i)$ in their natural top-to-bottom order must contain $x^j_q$.

If $x^j_q$ does not belong to a row of $B_i$, but instead lies on a red path of $B_i$, let $P$ be that red path, and assume that its two endpoints, $u$ and $u'$ belong to rows $R_s$ and $R_{s+1}$, respectively. We define sets $S_1, T_1$ and $S_2$ exactly as before. We change the definition of the set $T_2$ slightly: instead of a vertex from row $R_s$ lying in the first column of $B_{i+1}$, we include a vertex from $R_{s+1}$,  lying in the first column of $B_{i+1}$. The definitions of the paths $\lset^1$ and $\lset^3$ remain unchanged. The set $\lset^2$ also remains unchanged, except for the path contained in the row $R_s$. We replace that path with the following path: we include a segment of $R_s$ between with the unique vertex in $T_1\cap V(R_s)$ and $u$, the path $P$, and the segment of $R_{s+1}$ between $u'$ and the unique vertex in $T_2\cap V(R_{s+1})$. We then let $\lset(B_i)$ be the concatenation of $\lset^1,\lset^2,\lset^3$. Then  the $q$th path of $\lset(B_i)$ in their natural top-to-bottom order must contain $x^j_q$.

For each wall $B_i\in \tilde{\bset}$, we have defined a collection $\lset(B_i)$ of $t$ disjoint paths, that are contained in $W'[\nset(B_i)]$, where for $1\leq j\leq t$, the $j$th path starts and terminates at the $j$th row of $\rset$. We now connect all these paths together, as follows. For each consecutive pair $B_i,B_{i'}$ of walls in $\tilde{\bset}$, for each $1\leq j\leq t$, let $u$ be the last endpoint of the $j$th path in $\lset(B_i)$, and let $u'$ be the first endpoint of the $j$th path in $\lset(B_{i'})$. Both $u$ and $u'$ must belong to the $j$th row of $\rset$. We use a segment of that row to connect $u$ to $u'$. Once we process every consecutive pair of walls in $\tilde{\bset}$, we obtain a set $\lset$ of $t$ disjoint paths, such that, for each $1\leq q\leq t$, all vertices whose label is $q$ are contained in the $q$th path of $\rset$. Moreover, all paths in $\lset$ are contained in $W'\setminus \tilde A$.

Let $L_1,L_2,L_3,L_4$ be the first four paths of $\lset$. We extend the four paths slightly to the right and to the left, as follows. For each $1\leq j\leq 4$, let $R_{i_j}$ be the row to which the endpoints of the path $L_j$ belong (that is, $R_{i_j}$ is the $j$th row of $\rset$). Let $C_1,C_2,C_3$ be the first three columns of $B_{i^*}$, and let $C_4$ be the first column of $B_{i^{**}}$. (Recall that $B_{i^*}\in \set{B_1,\ldots,B_{t-3}}$ and $B_{i^{**}}\in\set{B_{N-t+4,\ldots,B_N}}$, and they do not contain vertices of  $\tilde A$.) We extend $L_1$ to the left along the row $R_{i_1}$ until it contains a vertex of $C_1$ (see Figure~\ref{fig: 4-clique}). We extend $L_2$ to the left along the row $R_{i_2}$ until it contains a vertex of $C_3$, and we extend it to the right along $R_{i_2}$ until it contains a vertex of $C_4$. We extend $L_3$ to the left along the row $R_{i_3}$ until it contains a vertex of $C_2$. Finally, we extend $L_4$ to the left along the row $R_{i_4}$ until it contains a vertex of $C_1$, and we extend $L_4$ to the right along the row $R_{i_4}$ until it contains a vertex of $C_4$. This finishes the definition of the paths $L_1,\ldots,L_t$. We embed, for each $1\leq i\leq t$, the vertex $v_i$ of graph $Z$ into the path $L_i$. Each edge $(u_j,v_i)$ of $E_1$, for $1\leq j\leq t-4$, $1\leq i\leq t$, is embedded into the path $Q_i^j$. Since $L_i$ is guaranteed to contain the endpoint $x_i^j$ of $Q_i^j$ (whose label is $i$), this is a valid embedding. Finally, we need to show how to embed the $6$ edges of $E_2$. Observe that each of the four paths $L_1,L_2,L_3,L_4$ intersect the column $C_3$, partitioning it into three segments, each connecting a consecutive pair of these paths. We use these three segments of $C_3$ to embed the edges $(v_1,v_2)$, $(v_2,v_3)$ and $(v_3,v_4)$. Column $C_2$ is only intersected by paths $L_1,L_3$ and $L_4$. We use the segment of $C_2$ between rows $R_{i_1}$ and $R_{i_3}$ to embed the edge $(v_1,v_3)$. Column $C_1$ is only intersected by $L_1$ and $L_4$. We use a segment of $C_1$ between rows $R_{i_1}$ and $R_{i_4}$ to embed the edge $(v_1,v_4)$. Finally, column $C_4$ is only intersected by $L_2$ and $L_4$, and we use it similarly to embed the edge $(v_2,v_4)$. 

\begin{figure}[h]
\scalebox{0.4}{\includegraphics{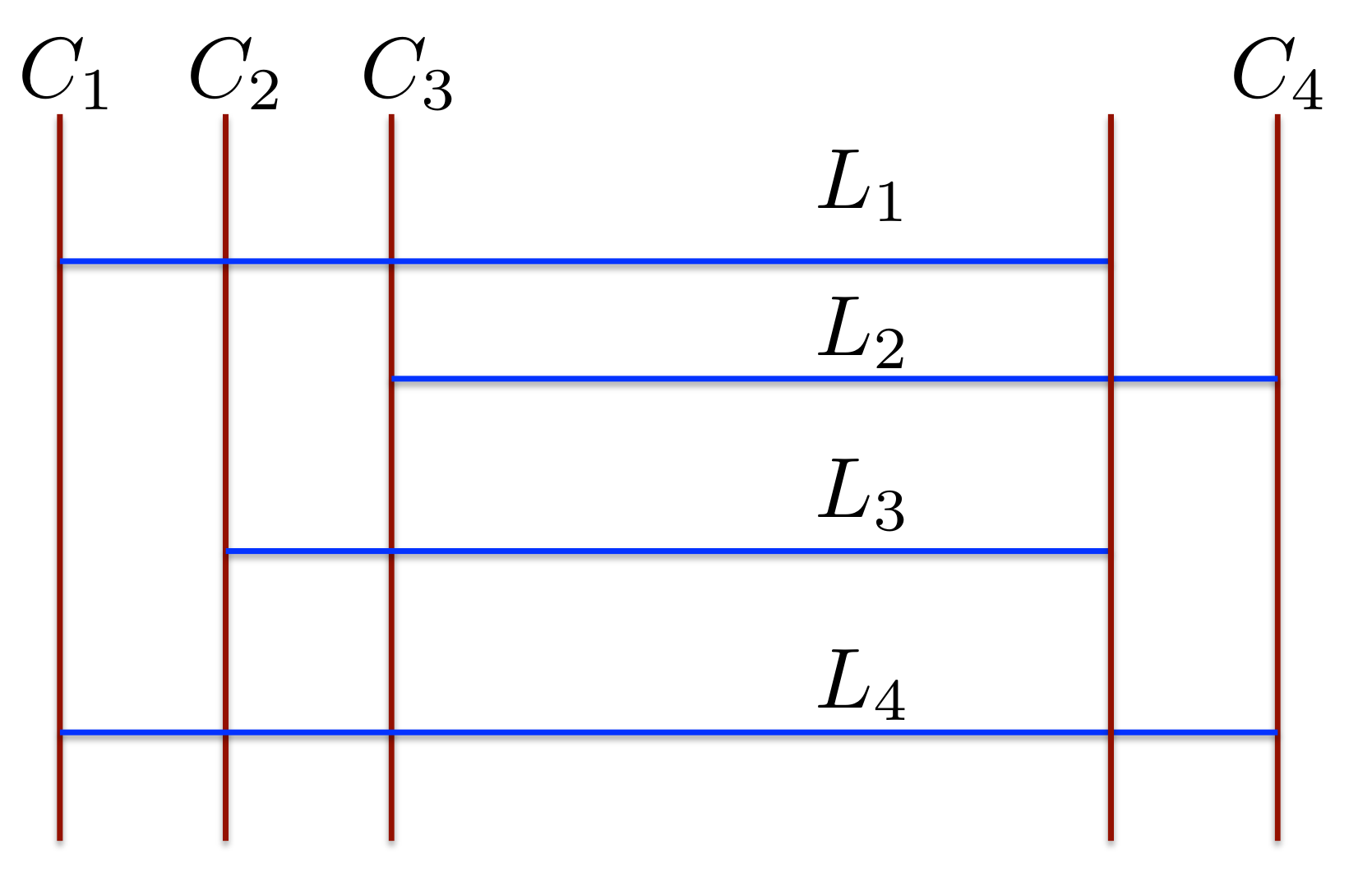}}\caption{Constructing the paths in $\sset$\label{fig: 4-clique}}
\end{figure}

This finishes the definition of the embedding of the graph $Z$ into $G$. As observed before, we can now find a model of a $K_t$-minor in $G$. It is easy to see that the model of the $K_t$ minor is grasped by $W'$, since there is at least one basic wall $B_i$, such that every path in $\lset$ intersects every column of $B_i$, and all basic walls $B_i$ are sub-walls of $W$.

\paragraph{Case 2}
If Case 2 happens, then $|\bset''\setminus \bset^*|\geq 80T+40+6t^2$. 
For each vertex $a_j\in A$, for each path $Q\in \qset_j$, let $x(Q)$ denote the endpoint of $Q$ that is different from $a_j$. 
We say that a wall $B_i\in \bset''$ is \emph{bad} iff $\nset(B_i)$ contains a vertex of $A$. Since for each pair $B_i,B_{i'}\in \bset''$, with $i\neq i'$, $|i-i'|\geq 3$ holds, the number of bad walls in $\bset''\setminus \bset^*$ is at most $|A|\leq 40T+20$.

For each vertex $a_j\in A$, we delete from $\qset_j$ all paths $Q$ where $x(Q)$ belongs to a bad wall. Then after this procedure, $\sum_{a_j\in A}|\qset_j|\geq |\bset''\setminus\bset^*|-(40T+20)\geq 40T+20+6t^2$, and $\sum_{a_j\in A}(|\qset_j|-1)\geq 6t^2$. 

We discard some additional vertices from $A$, to obtain the final set $\tilde A$. First, for each vertex $a_j\in A$, if $|\qset_j|\leq 1$, then we delete $a_j$ from $A$. It is easy to see that 
$\sum_{a_j\in A}(|\qset_j|-1)$ does not decrease. Finally, let $A'\subset A$ be the subset of vertices that lie in the first $2t$ rows of $W'$, and let $A''\subset A$ be the subset of vertices lying in the last $2t$ rows of $W'$. We assume w.l.o.g. that $\sum_{a_j\in A''}(|\qset_j|-1)\leq \sum_{a_j\in A'}(|\qset_j|-1)$. We delete the vertices of $A''$ from $A$, obtaining the final set $\tilde A$. From our construction, $\sum_{a_j\in \tilde A}(|\qset_j|-1)\geq 3t^2$. Let $\rset$ be the set of the bottom $2t$ rows of $W'$. Then no vertex of $\tilde A$ belongs to a row of $\rset$.

Our next step is to find a collection $A_1,\ldots,A_t$ of $t$ disjoint subsets of $\tilde A$, such that for each $1\leq r\leq t$, $\sum_{a_j\in A_r}(|\qset_j|-1)\geq t$. We will also delete some paths from sets $\qset_j$ for $a_j\in A_r$, to ensure that this summation is exactly $t$ for each set $A_r$. 

We find the partition of $\tilde A$ via a simple greedy procedure. Assume w.l.o.g. that $\tilde A=\set{a_1,\ldots,a_{m'}}$. Let $j$ be the smallest index, such that $\sum_{i=1}^j(|\qset_i|-1)\geq t$. Then, since for all $a_i\in \tilde A$, $2\leq |\qset_i|\leq 2t$, $\sum_{i=1}^j(|\qset_i|-1)\leq 3t$. We delete paths from $\qset_j$, until $\sum_{i=1}^j(|\qset_i|-1)= t$ holds. From our choice of $j$, $|\qset_j|\geq 2$ continues to hold. We set $A_1=\set{a_1,\ldots,a_j}$, delete the vertices of $A_1$ from $\tilde A$, and continue to the next iteration. Since $\sum_{a_i\in \tilde A}(|\qset_i|-1)\geq 3t^2$, we can continue this process for $t$ iterations, and find the desired collection $A_1,\ldots,A_t$ of subsets of $\tilde A$.

Consider some set $A_r$, for $1\leq r\leq t$. Recall that for each vertex $a_i\in A_r$, $|\qset_i|\geq 2$. We select one arbitrary path $Q_i^r\in \qset_i$, and we assign to its endpoint $x(Q_i^r)$ the label $(t+r)$. Let $\sset_1^r=\set{Q_i^r\mid a_i\in A_r}$ be the set of all paths whose endpoint is assigned the label $(t+r)$, and let $\sset_2^r=\left(\bigcup_{a_i\in A_r}\qset_i\right )\setminus \sset_1^r$ be the set of the remaining paths. Then $|\sset_2^r|=t$. For each $Q\in \sset_2^r$, we assign to $x(Q)$ a label in $\set{1,\ldots,t}$, such that each label is assigned to exactly one endpoint $x(Q)$ of a path in $\sset_2^r$.

In the rest of the proof, we will embed a $K_{t,t}$-minor into $G$, as follows. We denote the two sets of vertices in the bi-partition of $K_{t,t}$ by $\set{v_1,\ldots,v_t}$ and $\set{u_1,\ldots,u_t}$. We will build $2t$ paths $P_1,\ldots,P_{2t}$ in $W'\setminus \tilde A$, such that for each $1\leq i\leq 2t$, path $P_i$ contains all vertices with label $i$. This is done very similarly to the algorithm used in Case 1. 

For each $1\leq r\leq t$, let $C_r$ be the union of the path $P_{t+r}$ and all paths in $\sset_1^r$. Note that $C_r$ is a connected graph, as for each $Q\in \sset_1^r$, the label of $x(Q)$ is $t+r$, and so $P_{t+r}$ must contain $x(Q)$. For each $1\leq r\leq t$, we embed the vertex $v_r$ of $K_{t,t}$ into $C_r$. For each $1\leq i\leq t$, we embed the vertex $u_i$ of $K_{t,t}$ into the path $P_i$. The edge $(v_r,u_j)$ is then embedded into the unique path $Q\in\sset_2^r$ whose endpoint $x(Q)$ has label $j$. Then one endpoint of $Q$ belongs to $A_r$, and hence to $C_r$, while the other must lie on $P_j$. This finishes the description of the embedding of $K_{t,t}$ into $G$. It is now easy to obtain an embedding of $K_t$ into $G$: for each $1\leq i\leq t$, we take the union of the embeddings of $v_i,u_i$ and edge $(v_i,u_i)$ to obtain an embedding of the $i$th vertex of $K_t$. The edge connecting the $i$th and the $j$th vertices of $K_t$ is then embedded into the same path as the edge $(v_i,u_j)$ of $K_{t,t}$. We will ensure that each path $P_i$ that we construct intersects at least $t$ columns of $W'$. It then follows that the corresponding model of $K_t$ is grasped by $W$.

It now only remains to define the set $\set{P_1,\ldots,P_{2t}}$ of disjoint paths in $W'\setminus \tilde A$, such that for each $1\leq i\leq 2t$, all vertices whose label is $i$ belong to $P_i$, and each path $P_i$ intersects at least $t$ columns of $W'$. Recall that we have defined a set $\rset$ of $2t$ rows of $W'$ that do not contain the vertices of $\tilde A$. The remainder of the proof closely follows the analysis for Case 1.
For $1\leq j\leq 2t$, let $U_j$ be the set of vertices whose label is $j$, and let $U=\bigcup_{j=1}^{2t}U_j$. Let $\tilde \bset\subseteq \bset''$ be the set of walls containing the vertices of $U$.
Recall that for each such wall $B_i\in \tilde {\bset}$, exactly one vertex of $U$ belongs to $B_i$, and it must lie in $X_i$. Moreover, $\nset(B_i)\cap \tilde A=\emptyset$, and if $B_i,B_{i'}\in \tilde{\bset}$ with $i\neq i'$, then $|i-i'|\geq 3$ must hold. 

Consider now some vertex $x\in U$, and assume that $x$ belongs to some basic wall $B_i\in \tilde \bset$, and that the label of $x$ is $q$. Let $S_1$ be the set of $2t$ vertices lying in the first column of $B_{i-1}$, that belong to the rows in $\rset$, such that exactly one vertex from each row in $\rset$ belongs to $S_1$. Define $S_2$ similarly for the last column of $B_{i+1}$. 
We construct a set $\lset(B_i)$ of $2t$ disjoint paths, contained in $W'[\nset(B_i)]$, connecting the vertices of $S_1$ to the vertices of $S_2$, such that the $q$th path of $\lset(B_i)$ in their natural top-to-bottom order contains $x$. This is done exactly like in the analysis of Case 1.

For each wall $B_i\in \tilde{\bset}$, we have therefore defined a collection $\lset(B_i)$ of $2t$ disjoint paths, that are contained in $W'[\nset(B_i)]$, where for $1\leq j\leq 2t$, the $j$th path starts and terminates at the $j$th row of $\rset$. We now connect all these paths together, as before: for each consecutive pair $B_i,B_{i'}$ of walls in $\tilde{\bset}$, for each $1\leq j\leq 2t$, let $u$ be the last endpoint of the $j$th path in $\lset(B_i)$, and let $u'$ be the first endpoint of the $j$th path in $\lset(B_{i'})$. Both $u$ and $u'$ must belong to the $j$th row of $\rset$. We use a segment of that row to connect $u$ to $u'$. Once we process every consecutive pair of walls in $\tilde{\bset}$, we obtain a set $\lset$ of $2t$ disjoint paths, such that, for each $1\leq q\leq 2t$, all vertices whose label is $q$ are contained in the $q$th path of $\lset$. Moreover, all paths in $\lset$ are contained in $W'\setminus \tilde A$, and each path intersects at least $t$ columns of $W'$.

\paragraph{Case 3}
If Case $3$ happens, then $|A^*|\leq t-5$, and $|\bset^*|\geq 132T+50-(80T+40+6t^2)\geq 14T+6$. Recall that we have the following properties for the walls in $\bset^*$:

\begin{properties}{P}
\item For each $B_i,B_{i'}\in \bset^*$ with $i\neq i'$, $|i-i'|\geq 3$; \label{prop: separated}

\item For each $B_i\in \bset^*$, $A\cap V(B_i)=\emptyset$;

\item For each $B_i\in \bset^*$, there is no path $Q$ connecting a vertex of $X_i$ to a vertex of $A\setminus A^*$ in $G$, such that $Q$ is internally disjoint from $W'\cup A$; \label{prop: no path Q}
\item For each $B_i\in \bset^*$, there is no path $P$ routing $M_i$ in $G\setminus A^*$.\label{prop: not routable}
\end{properties}

In order to see that the last property holds, assume for contradiction that there is a path $P$ routing $M_i$ in $G\setminus A^*$. Since from Theorem~\ref{thm: find apex vertices} $M_i$ is not routable in $G\setminus A$, path $P$ must contain a vertex of $A\setminus A^*$, contradicting Property~(\ref{prop: no path Q}).

We say that $B_i\in \bset^*$ is a type-$\tone$ wall iff there is a path $P_i$ connecting $X_i$ to $V(W'\setminus B_i')$ in $G\setminus A^*$, such that $P_i$ is internally disjoint from $W'\cup A^*$. Let $B_i\in \bset^*$ be a wall of type-$\tone$, and let $P_i$ be the corresponding path. We denote by $x_i$ the endpoint of $P_i$ lying in $X_i=V(B'_i\setminus\Gamma_i')$, and by $y_i$ its other endpoint. From Property~(\ref{prop: not routable}), $y_i\not\in Y_i$, and so $y_i\in \nset_i\setminus V(B_i')$ must hold. Observe that if $B_i,B_{i'}\in \bset^*$ are type-$\tone$ walls, then $P_i$ and $P_{i'}$ must be disjoint - otherwise, by combining $P_i$ and $P_{i'}$, we can obtain a path $P$ connecting $x_i\in X_i$ to $x_{i'}\in Y_i$ in graph $G\setminus A^*$, contradicting Property~(\ref{prop: not routable}).
Let $\bset^*_1$ be the set of all type-$\tone$ walls in $\bset^*$. 

We say that Case 3a happens if $|\bset^*_1|\geq 12T+6$. Assume that Case 3a happens, and consider the sub-graph $G'$ of $G$, obtained by taking the union of $W'$ and the paths $P_i$ for all $B_i\in \bset^*_1$. Then $W'$ is a chain of $N$ walls of height at least $z'$ in $G'$, and for every wall $B_i\in \bset^*_1$, the corresponding $\tau$-core wall $B_i'$ is a type-1 wall in $G'$, with $P_i$ being the neighborhood bridge for $B_i$. Applying Theorem~\ref{thm: type-1 walls} to $G'$ and $W'$, we obtain an efficient algorithm to find a model of a $K_t$-minor in $G'$ (and hence in $G$), grasped by $W$.

From now on we assume that Case 3a does not happen. Let $\bset^*_2=\bset^*\setminus\bset^*_1$ and consider any wall $B_i\in \bset^*_2$. 
Then $(G\setminus A^*)\setminus \Gamma_i'$ consists of at least two connected components, with one of them containing $B'_i\setminus \Gamma'_i$. Therefore, there is a separation $(X,Y)$ of $G\setminus A^*$, with $B'_i\subseteq X$, $X\cap Y\subseteq \Gamma_i'$, and for each $B_j\in \bset^*_2$ with $j\neq i$, $B_j\subseteq Y$. Recall that the corners $a_i',b'_i,c_i',d_i'$ of the wall $B'_i$ are fixed. We assume that they appear on $\Gamma_i'$ in this order clockwise. If graph $X$ contains a wall-cross for $B'_i$ (that is, a pair of disjoint paths connecting $a_i'$ to $c_i'$ and $b_i'$ to $d_i'$), then we say that wall $B_i$ is of type $\ttwo$. Otherwise, it is of type $\tthree$. 

If at least one wall $B_i\in \bset^*_2$ is a type-$\tthree$ wall, then from Lemma~\ref{lemma: type 4 gives flat wall}, we can efficiently find a flat wall $B'$ of size $((z'-2\tau)\times(z'-2\tau))=((z'-4t)\times(z'-4t))=(w\times w)$ in $G\setminus A^*$, such that $B'$ is a sub-wall of $B_i$ and hence of $W$.

From now on we assume that all walls in $\bset^*_2$ are of type $\ttwo$. Recall that $|\bset^*_2|\geq 14T+6-(12T+6)\geq 2T$. For each wall $B_i\in \bset^*_2$, let $Q_i^1,Q_i^2$ be the pair of disjoint paths realizing the wall-cross for $B_i'$, and let $\tilde{Q}=\bigcup_{B_i\in \bset^*_2}\set{Q_i^1,Q_i^2}$. As in Case 3a, all paths in $\tilde Q$ must be completely disjoint, since otherwise we can combine two such paths to obtain a routing of some demand $M_i$ for $B_i\in \bset^*_2$ in graph $G\setminus A^*$, contradicting Property~\ref{prop: not routable}. Consider the sub-graph $G'$ of $G$, obtained by taking the union of $W'$ and the paths in $\tilde Q$. Then $W'$ is a chain of $N$ walls of height at least $z'$ in $G'$, and for every wall $B_i\in \bset^*_2$, the corresponding $\tau$-core wall $B_i'$ is a type-3 wall in $G'$, with $Q_i^1,Q_i^2$ being the corresponding wall-cross. Applying Theorem~\ref{thm: type-3 walls} to $G'$ and $W'$, we can find a model of a $K_t$ minor in $G'$ (and hence in $G$), grasped by $W$.

\label{------------------------------negative result----------------------------}
\section{A Lower Bound}\label{sec: lower bound}
In this section we prove Theorem~\ref{thm: lower bound}.
We can assume that $w,t\geq 4000$: otherwise, we can use a graph $G$ consisting of a single vertex.
We round $w$ down to the closest integral multiple of $4$, and
we set $w'=w/4-8$ and $t'=\floor{t/30}$. In order to construct the graph $G$, we start with a grid whose height and width is $(w't'-1)$. For each $0\leq i< t'$, $0\leq j<t'$, vertex $v(iw',jw')$ is called a \emph{special vertex}. The unique cell of the grid for which $v(iw',jw')$ is the left top corner is denoted by $Q(iw',jw')$, and we call it a \emph{black cell}. For each black cell $Q(iw',jw')$, we add the two diagonals $(v(iw',jw'),v(iw'+1,jw'+1))$ and $(v(iw'+1,jw'),v(iw',jw'+1))$ to the graph. This completes the definition of the graph $G$. Clearly, $G$ contains a wall of size $\Omega(w't')=\Omega(wt)$ as a minor.
 We next prove that $G$ does not contain a $K_t$-minor.

\begin{theorem}\label{thm: no clique minor}
Graph $G$ does not contain a $K_t$-minor.
\end{theorem}
\begin{proof}
The proof uses the notions of graph drawing and graph crossing number. A drawing of a graph $H$ in the plane is a mapping, in which every vertex of $H$ is mapped into a point in the plane, and every edge into a continuous curve connecting the images of its endpoints, such that no three curves meet at the same point, and no curve contains an image of any vertex other than its endpoints.
A \emph{crossing} in such a drawing is a point where the images of two edges intersect, and the \emph{crossing number} of a graph $H$, denoted by $\optcro{H}$, is the smallest number of crossings achievable by any drawing of $H$ in the plane.
We use the following well-known theorem~\cite{ajtai82,leighton_book}.

\begin{theorem}\label{thm: crossing number of a clique}
For any graph $G=(V,E)$ with $|E|>7.5|V|$, $\optcro{G}\geq\frac{|E|^3}{33.75|V|^2}$. In particular, for all $n>16$, $\optcro{K_n}>(n-1)^4/272$.
\end{theorem}


Assume for contradiction that $G$ contains a $K_t$-minor, and consider its model $f$. The main idea is to use the natural drawing $\psi$ of $G$, that contains $(t'-1)^2$ crossings, together with the model $f$, to obtain a drawing $\psi'$ of $K_t$ with fewer than $(t-1)^4/272$ crossings, leading to a contradiction.
For convenience, instead of defining a drawing of $K_t$, we define a drawing of another graph $H$, obtained from $K_t$ by subdividing each edge $e\in E(K_t)$ by two new vertices, $u(e)$ and $u'(e)$. Clearly, a drawing of $H$ with $z$ crossings immediately gives a drawing of $K_t$ with $z$ crossings. We let $V_1=V(K_t)$, and $V_2=V(H)\setminus V_1$. For each vertex $v\in V_1$, let $\delta(v)$ be the set of edges of $H$ incident on $e$. Let $E'$ be the set of edges of $H$ whose both endpoints belong to $V_2$.

We first define the drawings of the vertices of $V_1$.
For each vertex $v\in V(K_t)$, we select an arbitrary vertex $x_v\in f(v)$. The drawing of $v$ in $\psi'$ is at the same point as the drawing of $x_v$ in $\psi$.

We now turn to define the drawings of the vertices of $V_2$ and the edges of $H$. Along the way, for each vertex $v\in V_1$, we will define a set $\pset(v)$ of paths contained in $f(v)$. The paths in $\pset(v)$ will be used to define the drawings of the edges in $\delta(v)$. 
For each edge $e\in E(H)$, we will associate a path $Q(e)\subseteq G$ with $e$, and we will draw the edge $e$ \emph{along} the drawing of the path $Q(e)$ in $\psi$. In other words, let $\gamma$ be the drawing of the path $Q(e)$ in $G$. The drawing $\psi'(e)$ of $e$ will start at the first endpoint of $\gamma$, and then will continue very close to $\gamma$, in parallel to it and without crossing it, so that $\psi'(e)$ does not contain the images of any vertices of $G$, except for the endpoints of $\gamma$. It will then terminate at the drawing of the second endpoint of $\gamma$. Notice that for now we allow $\psi'(e)$ to self-intersect arbitrarily. Consider now two edges $e,e'\in E(H)$, and their corresponding paths $Q(e)$, $Q(e')$. 
We distinguish between three types of crossings between $\psi'(e)$ and $\psi'(e')$. Type-1 crossings arise whenever an edge $e^*\in Q(e)$ crosses an edge $e^{**}\in Q(e')$ in $\psi$. The number of type-1 crossings between the images of  $e$ and $e'$ in $\psi'$ is bounded by the number of crossings between the edges of $Q(e)$ and the edges of $Q(e')$ in $\psi$. We will ensure that for each edge $e^*\in E(G)$, at most $t-1$ paths in $\set{Q(e)\mid e\in E(H)}$ contain $e^*$. Therefore, the number of type-1 crossings can be bounded by $(t-1)^2$ times the number of crossings in $\psi$, giving the total bound of $(t-1)^2\cdot (t'-1)^2$. If two paths $Q(e),Q(e')$ share some edge $e^*\in E(G)$, then the portions of the images of $e$ and $e'$ that are drawn along $e^*$ may cross arbitrarily. Similarly, if $Q(e)$ and $Q(e')$ share some vertex $v\in V(G)$ where $v$ is an inner vertex on both paths, then the images of $e$ and $e'$ may cross arbitrarily next to  $\psi(v)$. We call all such crossings type-2 crossings. We also include among type-2 crossings the self-crossings of an image of any edge $e\in E(H)$, that are not type-1 crossings. (We will eventually eliminate all type-2 crossings.) Finally, if an endpoint $v$ of some path $Q(e)$ also belongs to some path $Q(e')$, where $e,e'\in E(H)$ and $e\neq e'$, then we allow the images of $e$ and $e'$ to cross once due to this containment. We call all such crossings \emph{type-3 crossings}. We will ensure that each vertex $v\in V(G)$ may serve as an inner vertex in at most $t-1$ paths $\set{Q(e)\mid e\in E(H)}$, and, since the number of vertices of $G$ serving as endpoints of paths in $\set{Q(e)\mid e\in E(H)}$ is at most $2|E(H)|$, the number of all type-3 crossings will be bounded by $2|E(H)|\cdot (t-1)\leq 6t(t-1)^2$.


We now proceed to define the drawings of the vertices of $V_2$ and the edges of $H$, along with the sets $\pset(v)$ of paths for all $v\in V_1$.
We start with $\pset(v)=\emptyset$ for all $v\in V_1$.

Let $e=(v,v')$ be any edge of $K_t$. Recall that $f(e)$ is an edge $e'$, connecting some vertex $a\in f(v)$ to some vertex $b\in f(v')$. Since $f(v)$ induces a connected sub-graph in $G$, let $P_1$ be any simple path connecting $x_v$ to $a$ in $G[f(v)]$. Similarly, let $P_2$ be any path connecting $b$ to $x_{v'}$ in $G[f(v')]$. 
Consider the vertices $u(e),u'(e)$ that subdivide the edge $e$ in $H$, and assume that $u(e)$ lies closer to $v$ than $u'(e)$ in the subdivision.
We denote the edges $e_1=(v,u(e))$, $e_2=(u(e),u'(e))$, and $e_3=(u'(e),v')$.
We draw the edge $e_2$ along the drawing $\psi(e')$, where $u(e)$ is drawn at $\psi(a)$, and $u'(e)$ is drawn at $\psi(b)$, and we set $Q(e_2)=(e')$.
We draw the edge $e_1=(v,u(e))$ of $H$ along the path $P_1$, setting $Q(e_1)=P_1$, and we add $P_1$ to $\pset(v)$. Similarly,  we draw the edge $(v',u'(e))$ along the path $P_2$, setting $Q(e_2)=P_2$. We then add $P_2$ to $\pset(v')$.

Recall that the graphs $\set{G[f(v)]\mid v\in V(K_t)}$ are completely disjoint. Moreover, the edges $\set{f(e)\mid e\in E(K_t)}$ are all distinct, and they do not belong to the graphs $\set{G[f(v)]\mid v\in V(K_t)}$.
Therefore, for $v,v'\in V_1$ with $v\neq v'$, the paths in $\pset(v)$ and $\pset(v')$ are completely disjoint. It is then easy to see that type-2 crossings in $\psi'$ are only possible between the images of edges $e,e'\in E(H)$, where $e,e'\in \delta(v)$ for some $v\in V_1$. Our next step is to re-route the edges of $\delta(v)$ along the paths contained in $f(v)$ in such a way that their corresponding drawings do not have type-2 crossings. In order to do so, we perform a simple un-crossing procedure. 
Given a pair $e,e'\in \delta(v)$ of edges, whose images have a type-2 crossing in $\psi'$, we remove one of the type-2 crossings, without increasing the total number of crossings in the current drawing, by un-crossing the images of the two edges, as shown in Figure~\ref{fig: uncrossing}. We continue performing this procedure, until for each vertex $v\in V_1$, for every pair $e,e'\in \delta(v)$, the images of $e$ and $e'$ do not have a type-2 crossing. We can still associate, with each edge $e\in \delta(v)$, a path $Q(e)\subseteq G[f(v)]$, such that $e$ is drawn along $Q(e)$. We also eliminate type-2 self-crossings of an edge by simply shortcutting the image of the edge at the crossing point.
Eventually, only type-1 and type-3 crossings remain in the graph. An edge $e^*\in E(G)$ may belong to at most $(t-1)$ paths in $\set{Q(e)\mid e\in E(H)}$ (if $e$ lies in $G[f(v)]$ for some $v\in V_1$, then it may only belong to the paths in $\pset(v)$; otherwise, it may belong to at most one path $Q(e')$ for $e'\in E'$); similarly, a vertex $u\in V(G)$ may be an inner vertex on at most $(t-1)$ paths in $\set{Q(e)\mid e\in E(H)}$. Therefore, as observed before, the total number of crossings in $\psi'$ is bounded by $(t-1)^2(t'-1)^2+6t(t-1)^2\leq (t-1)^2(\frac t{30}-1)^2+6t(t-1)^2<(t-1)^4/272$, contradicting Theorem~\ref{thm: crossing number of a clique}.

\begin{figure}[h]
\scalebox{0.5}{\includegraphics{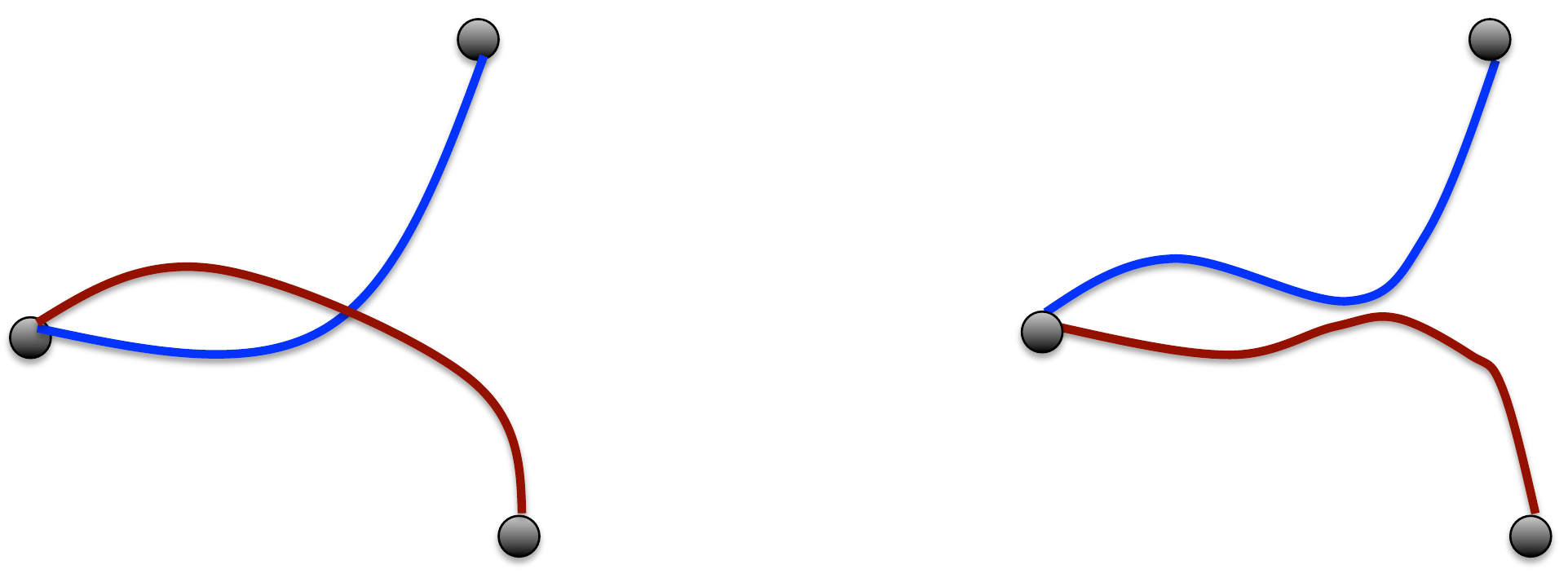}}\caption{Uncrossing the drawings of a pair of edges to eliminate a type-2 crossing.}\label{fig: uncrossing}
\end{figure}

\end{proof}

The following theorem completes the proof of Theorem~\ref{thm: lower bound}.
\begin{theorem}\label{thm: no flat wall}
Graph $G$ does not contain a flat wall of size $(w\times w)$.
\end{theorem}

\begin{proof}
Assume otherwise. Let $W$ be the flat wall of size $(w\times w)$, $U$ the set of the pegs of $W$, $\Gamma$ its boundary, and $(A,B)$ the separation of $G$ certifying the flatness of $W$: that is, $W\subseteq B$, $A\cap B\subseteq \Gamma$, $U\subseteq A\cap B$, and $B$ is $A\cap B$-flat.

For $1\leq i\leq w/4$, let $W_i$ be the sub-wall of $W$ spanned by rows $(R_i,\ldots,R_{w-i+1})$ and columns $(C_i,\ldots,C_{w-i+1})$, so $W_1=W$, and let $\Gamma_i$ be the boundary of $W_i$.

Fix some $1< i\leq w/4$. Observe that every path $P$ in graph $G$ connecting a vertex of $W_i\setminus \Gamma_i$ to a vertex of $W\setminus W_i$ must contain a vertex of $\Gamma_i$: otherwise, there is a path $P'$ whose endpoints belong to  $W_i\setminus \Gamma_i$ and $W\setminus W_i$, respectively, and $P'$ is internally disjoint from $W$. Path $P'$ must be contained in $B$, since $\Gamma$ separates $A$ from $B$, and $P'$ is internally disjoint from $\Gamma$. Then we can use Theorem~\ref{thm: cross-wall} to build a wall-cross for $W$ in graph $B$, contradicting the fact that $W$ is a flat wall.
Therefore, there is a separation $(A_i,B_i)$ of $G$, with $W_i\subseteq B_i$ and $A_i\cap B_i\subseteq \Gamma_i$, such that $A\subseteq A_i$ and $W\setminus W_i\subseteq A_i$.

Following is the central lemma in the proof of Theorem~\ref{thm: no flat wall}.

\begin{lemma}\label{lemma: no black cell in Bi}
Let $3\leq i\leq w/4$, let $Q$ be any black cell, and let $S$ be the set of the $4$ vertices serving as the corners of $Q$. Then $S\not\subseteq B_i$.
\end{lemma}

Before we prove Lemma~\ref{lemma: no black cell in Bi}, let us complete the proof of Theorem~\ref{thm: no flat wall} using it.
%
%
%
%
Let $v^*$ be one of the vertices in the intersection of row $R_{w/2}$ and column $C_{w/2}$ of $W$. Then there must be a black cell $Q$ in $G$, such that there is a path $P$ of length at most $w'$ from $v^*$ to one of the corners of the cell $Q$ in $G$. Let $S'$ be the set of the vertices on $P$, and the vertices that serve as corners of $Q$, so $|S'|\leq w'+4$. Consider the cycles $\Gamma_4,\ldots,\Gamma_{w/4}$. Since $|S'|\leq w'+4<w/4-3$, at least one of these cycles $\Gamma_i$ does not contain any vertex of $S'$. Therefore, in $G\setminus \Gamma_i$, $v^*$ is connected to all vertices of $S'$, and in particular $S'\subseteq B_i$, a contradiction.
From now on we focus on proving Lemma~\ref{lemma: no black cell in Bi}. 
Our starting point is the following simple claim.
\begin{claim}\label{claim: four paths}
Let $1\leq i\leq w/4$, let $Q=Q(i'w',j'w')$ be any black cell, and let $S$ be the set of the four vertices serving as the corners of $Q$. Assume further that $S\subseteq B_i$. Then there are four disjoint paths in $B_i$ connecting the vertices of $S$ to the vertices of $\Gamma_i$.
\end{claim}

\begin{proof}
We first prove that there are four disjoint paths in $G$  connecting the vertices of $S$ to the vertices of $\Gamma$. Assume otherwise.
Then there is a set $X$ of $3$ vertices separating $S$ from $V(\Gamma)$ in $G$.
Since $w>32$, $|V(\Gamma)|> 32$.

Let $\cset'$ be the subset of the columns of the grid, such that for each column $C\in \cset'$, $C\cap \Gamma\neq \emptyset$. Similarly, 
let $\rset'$ be the subset of the rows of the grid, such that for each row $R\in \rset'$, $R\cap \Gamma\neq \emptyset$. Clearly, either $|\cset'|\geq 4$, or $|\rset'|\geq 4$ must hold: otherwise, all vertices of $\Gamma$ are contained in the columns of $\cset'$ and the rows of $\rset'$, and, if $|\cset'|,|\rset'|<4$, then $|V(\Gamma)|\leq 9$ must hold, a contradiction.
We assume w.l.o.g., that $|\cset'|\geq 4$. Then at least one column of $\cset'$ contains no vertex of $X$. Let $C^*$ be that column.

Consider the following four paths: $P_1$ is obtained by connecting the top left corner of $Q$ to the top of the grid $G'$, along the column $C_{j'w'}$; $P_2$ is obtained similarly by connecting the top right corner of $Q$ to the top of the grid, along the column $C_{j'w'+1}$. Define $P_3,P_4$ similarly, by connecting the bottom two corners of $Q$ to the bottom of the grid. Then all four paths are node-disjoint, and have length at least $w'-1$ each. At least one of these paths does not contain a vertex of $X$. Assume w.l.o.g. that it is $P_1$. Finally, let $R^1,R^2,R^3,R^4$ be the set of any 4 rows of the grid that intersect $P_1$: since $|V(P_1)|\geq w'-1$, such a set exists. At least one of these four rows contains no vertex of $X$ - assume w.l.o.g. that it is $R^1$. Then $R^1$ intersects both $P_1$, and either $V(\Gamma\setminus X)$ or $V(C^*\setminus \Gamma)$. Therefore, the union of $P_1,R^1,C^*$ and $\Gamma\setminus X$ must contain a path connecting a vertex of $S$ to a vertex of $\Gamma$ in $G\setminus X$, a contradiction.

Therefore, $G$ contains four disjoint paths connecting the vertices of $S$ to the vertices of $\Gamma$. Each such path must contain at least one vertex of $\Gamma_i$, since $\Gamma_i$ separates $B_i$ from $\Gamma$. By truncating these paths, we can ensure that they terminate at the vertices of $\Gamma_i$, and they do not contain the vertices of $\Gamma_i$ as inner vertices. The resulting set of paths must then be contained in $B_i$.
\end{proof}

We will also repeatedly use the following simple claim, whose proof can be found, e.g. in~\cite{RS90}, and is included here for completeness.

\begin{claim}\label{claim: re-routing paths}
Let $H$ be any graph, $X,Y$ any pair of disjoint vertex subsets of $H$, and assume that there is a set $\pset$ of $k$ disjoint paths connecting the vertices of $X$ to the vertices of $Y$ in $H$. Let $X'\subseteq X$, and assume that there is a set $\pset'$ of $k-1$ disjoint paths connecting the vertices of $X'$ to the vertices of $Y$ in $H$, such that the paths in $\pset'$ are internally disjoint from $X\cup Y$. 
Then there is a set $\pset''$ of $k$ disjoint paths connecting the vertices of $X$ to the vertices of $Y$ in $H$, such that the paths in $\pset''$ are internally disjoint from $X\cup Y$, and $k-1$ of the paths in $\pset''$ originate at the vertices of $X'$.
\end{claim}

\begin{proof}
Let $H'$ be the graph obtained from $H$ by unifying the vertices of $X\setminus X'$ into a single vertex $v^*$. Let $X^*=X'\cup\set{v^*}$. We claim that there is a set $\pset''$ of $k$ disjoint paths in $H'$ connecting the vertices of $X^*$ to the vertices of $Y$. Assume otherwise. Then there is a set $Z$ of $k-1$ vertices, separating $X^*$ from $Y$ in $H'$. Since the paths of $\pset'$ are contained in $H'$, each vertex of $Z$ must lie on a distinct path of $\pset'$, and in particular $v^*\not\in Z$. Set $\pset$ of paths defines $k$ internally disjoint paths connecting the vertices of $X^*$ to the vertices of $Y$ in $H'$. The only vertex that the paths in $\pset$ may share is $v^*$. Since $v^*\not\in Z$, there is at least one path $P\in \pset$ that does not contain a vertex of $Z$, a contradiction. Therefore, there is a set $\pset''$ of $k$ disjoint paths in $H'$ connecting the vertices of $X^*$ to the vertices of $Y$. These paths naturally define the desired set of paths in $H$.
\end{proof}

We are now ready to complete the proof of Lemma~\ref{lemma: no black cell in Bi}. 
Fix some $3\leq i\leq w/4$, and let $Q$ be some black cell, such that the set $S$ of the four vertices serving as the corners of $Q$ is contained in $B_i$. From Claim~\ref{claim: four paths}, we can find two disjoint paths, $P_1,P_2$, connecting the vertices of $S$ to the vertices of $\Gamma_i$ in $B_i$. Let $a,b,c,d$ be the four corners of the wall $W_{i-1}$, in this clock-wise order, where $a$ is the top left corner. It is easy to see that we can extend the two paths $P_1,P_2$, using the edges of $W_i\setminus W_{i-1}$, so that they connect two vertices of $S$ to $a$ and $c$, the two paths remain disjoint, and are contained in $B_{i-1}$. 

From Claim~\ref{claim: four paths}, there are three disjoint paths in $B_{i-1}$, connecting the vertices of $S$ to the vertices of $\Gamma_{i-1}$. Using Claim~\ref{claim: re-routing paths}, we can assume that two of these paths terminate at $a$ and $c$, respectfully. The third path can then be extended, using the edges of $\Gamma_{i-1}$, so that it terminates at either $c$ or $d$, and it remains disjoint from the first two paths. We assume w.l.o.g. that it terminates at $c$. Let $P'_1,P'_2,P'_3$ be the resulting three paths.

Let $a',b',c',d'$ be the four corners $W$, that appear on $\Gamma$ in this clock-wise order, where $a'$ is the top left corner. We can extend the three paths $P'_1,P'_2,P'_3$, using the edges of $E(W)\setminus E(W_{i-1})$, so that they connect three vertices of $S$ to $a',b'$ and $c'$, such that the three paths remain disjoint. It is easy to see that all three paths are contained in $B=B_1$. Finally, using Claim~\ref{claim: four paths}, there are four disjoint paths in $B$, connecting the vertices of $S$ to the vertices of $\Gamma$. Using Claim~\ref{claim: re-routing paths}, we can assume that three of these paths terminate at $a'$, $b'$ and $c'$. The vertices $a',b',c'$ partition $\Gamma$ into three segments, each of which contains at least two pegs. Let $x$ be the endpoint of the fourth path. Then we can always extend the fourth path along $\Gamma$, so that it remains disjoint from the first three paths, and it terminates at a peg of $W$. As the corners of a wall are a subset of its pegs, we now obtained a set $\pset$ of four disjoint paths, connecting the vertices of $S$ to the vertices of $U$, where $\pset\subseteq B$. Using the paths in $\pset$, and the edges of the cell $Q$, we can route any matching between the four corresponding pegs in graph $B$. This contradicts the fact that $B$ is $A\cap B$-flat.
This completes the proof of Lemma~\ref{lemma: no black cell in Bi}, and of Theorem~\ref{thm: no flat wall}.
\end{proof}

\bibliographystyle{alpha}
\bibliography{flat-wall-thm}

\end{document}